\documentclass[11pt,a4paper]{article}
\pdfoutput=1

\usepackage{amssymb}
\usepackage{amsmath}
\usepackage{amsfonts}
\usepackage{bbm}
\usepackage{amsthm}
\usepackage{mathrsfs}
\usepackage{hyperref}
\usepackage{color}
\usepackage[margin=2.41cm]{geometry}
\usepackage[all,cmtip]{xy}
\usepackage[utf8]{inputenc}
\usepackage{graphicx}
\usepackage{varwidth}
\usepackage{comment}
\usepackage{faktor}
\usepackage[shortlabels]{enumitem}
\usepackage{euscript}

\usepackage{mathtools}

\usepackage{upgreek}
\usepackage{rotating}

\usepackage{tikz}
\usetikzlibrary{shapes.geometric}

\usepackage{tikz-cd}
\usetikzlibrary{matrix,arrows,calc,decorations.pathmorphing,fit,shapes.geometric,decorations.pathreplacing,positioning}

\usepackage{amssymb}
\tikzset{curve/.style={settings={#1},to path={(\tikztostart)
    .. controls ($(\tikztostart)!\pv{pos}!(\tikztotarget)!\pv{height}!270:(\tikztotarget)$)
    and ($(\tikztostart)!1-\pv{pos}!(\tikztotarget)!\pv{height}!270:(\tikztotarget)$)
    .. (\tikztotarget)\tikztonodes}},
    settings/.code={\tikzset{quiver/.cd,#1}
        \def\pv##1{\pgfkeysvalueof{/tikz/quiver/##1}}},
    quiver/.cd,pos/.initial=0.35,height/.initial=0}

\definecolor{darkred}{rgb}{0.8,0.1,0.1}

\hypersetup{
	colorlinks=true,         
	urlcolor=darkred, 
	linkcolor=darkred,
	citecolor=blue,
}

\theoremstyle{plain}
\newtheorem{theo}{Theorem}[section]
\newtheorem{lem}[theo]{Lemma}
\newtheorem{propo}[theo]{Proposition}
\newtheorem{cor}[theo]{Corollary}

\theoremstyle{definition}
\newtheorem{defi}[theo]{Definition}

\newtheorem{assu}[theo]{Assumption}

\newtheorem{exx}[theo]{Example}
\newenvironment{ex}
{\pushQED{\qed}\exx}
{\popQED\endexx}

\newtheorem{remm}[theo]{Remark}
\newenvironment{rem}
{\pushQED{\qed}\remm}
{\popQED\endremm}

\numberwithin{equation}{section}

\def\nn{\nonumber}

\def\bbR{\mathbb{R}}
\def\bbC{\mathbb{C}}

\def\bbZ{\mathbb{Z}}
\def\bbS{\mathbb{S}}

\def\Hom{\mathrm{Hom}}

\def\id{\mathrm{id}}
\def\Id{\mathrm{Id}}

\def\cc{\mathrm{c}}

\def\1{I}

\def\op{\mathrm{op}}

\def\Set{\mathbf{Set}}
\def\Alg{\mathbf{Alg}}

\def\Vec{\mathbf{Vec}}

\def\astAlg{{}^{\ast}\mathbf{Alg}_\bbC}
\def\CastAlg{C^{\ast}\mathbf{Alg}_\bbC}
\def\preCastAlg{C_{\mathrm{pre}}^{\ast}\mathbf{Alg}_\bbC}

\def\astCat{{}^{\ast}\mathbf{Cat}_\bbC}
\def\CastCat{C^{\ast}\mathbf{Cat}_\bbC}

\def\astAQFT{{}^{\ast}\mathbf{AQFT}}
\def\CastAQFT{C^{\ast}\mathbf{AQFT}}
\def\preCastAQFT{C_{\mathrm{pre}}^{\ast}\mathbf{AQFT}}

\def\Sub{\mathbf{Sub}}
\def\Fin{\mathbf{Fin}}

\def\CC{\mathbf{C}}
\def\DD{\mathbf{D}}

\def\Cat{\mathbf{Cat}}

\def\Fun{\mathbf{Fun}}

\def\Cone{\mathbf{Cone}}
\def\Disk{\mathbf{Disk}}
\def\AAlg{\EuScript{A}\mathsf{lg}}

\def\SSS{\mathbf{SSS}}

\def\AAA{\mathfrak{A}}

\def\BBB{\mathfrak{B}}

\def\CCC{\mathfrak{C}}

\def\FFF{\mathfrak{F}}

\def\A{\mathcal{A}}
\def\B{\mathcal{B}}
\def\C{\mathcal{C}}

\def\O{\mathcal{O}}
\def\P{\mathcal{P}}

\def\V{\EuScript{V}}
\def\CCastCat{C^\ast\EuScript{C}\mathsf{at}}

\def\colim{\mathrm{colim}}

\def\pt{\mathrm{pt}}

\newcommand\und[1]{\underline{#1}}

\newcommand{\norm}[1]{\vert\!\vert #1 \vert\!\vert}

\DeclareMathOperator{\cdiamond}{\diamond	\hspace*{-1.48mm}\cdot	\hspace*{0.8mm}}

\DeclareMathOperator*{\bigboxtimes}{\text{\raisebox{-0.5ex}{\scalebox{1.4}{$\boxtimes$}}}}

\def\sk{\vspace{2mm}}

\makeatletter
\let\@fnsymbol\@alph
\makeatother

\title{%
$C^\ast$-categorical prefactorization algebras\\
for superselection sectors and topological order
}

\author{%
Marco Benini$^{1,2,a}$, Victor Carmona$^{3,b}$, Pieter Naaijkens$^{4,c}$\ and\ Alexander Schenkel$^{5,d}$\vspace{4mm}\\
{\small ${}^1$ Dipartimento di Matematica, Dipartimento di Eccellenza 2023-27, Universit\`a di Genova,}\\
{\small Via Dodecaneso 35, 16146 Genova, Italy.}\vspace{2mm}\\
{\small ${}^2$ INFN, Sezione di Genova,}\\
{\small Via Dodecaneso 33, 16146 Genova, Italy.}\vspace{2mm}\\
{\small ${}^3$ Max Planck Institut f\"ur Mathematik in den Naturwissenschaften,}\\
{\small Inselstra\ss e 22, 04103 Leipzig, Germany.}\vspace{2mm}\\
{\small ${}^4$ School of Mathematics, Cardiff University,}\\
{\small Senghennydd Road, Cardiff CF24 4AG, United Kingdom.}\vspace{2mm}\\
{\small ${}^5$ School of Mathematical Sciences, University of Nottingham,}\\
{\small University Park, Nottingham NG7 2RD, United Kingdom.}\vspace{4mm}\\
{\small \begin{tabular}{ll}
Email: & ${}^a$~\href{mailto:marco.benini@unige.it}{\texttt{marco.benini@unige.it}}\\
& ${}^b$~\href{mailto:victor.carmona@mis.mpg.de}{\texttt{victor.carmona@mis.mpg.de}}\\
& ${}^c$~\href{mailto:NaaijkensP@cardiff.ac.uk}{\texttt{NaaijkensP@cardiff.ac.uk}}\\
& ${}^d$~\href{mailto:alexander.schenkel@nottingham.ac.uk}{\texttt{alexander.schenkel@nottingham.ac.uk}}
\vspace{2mm}
\end{tabular}
}
}

\date{November 2025}

\begin{document}

\maketitle

\begin{abstract}
\noindent This paper presents a conceptual and efficient geometric framework to encode the algebraic structures on the category of superselection sectors of an algebraic quantum field theory on the $n$-dimensional lattice $\mathbb{Z}^n$. It is shown that, under the typical assumption of Haag duality, the monoidal $C^\ast$-categories of localized superselection sectors carry the structure of a locally constant prefactorization algebra over the category of cone-shaped subsets of $\mathbb{Z}^n$. Employing techniques from higher algebra, one extracts from this datum an underlying locally constant prefactorization algebra defined on open disks in the cylinder $\mathbb{R}^1\times\mathbb{S}^{n-1}$. While the sphere $\mathbb{S}^{n-1}$ arises geometrically as the angular coordinates of cones, the origin of the line $\mathbb{R}^1$ is analytic and rooted in Haag duality. The usual braided (for $n=2$) or symmetric (for $n\geq 3$) monoidal $C^\ast$-categories of superselection sectors are recovered by removing a point of the sphere $\mathbb{R}^1\times(\mathbb{S}^{n-1}\setminus\mathrm{pt}) \cong\mathbb{R}^n$ and using the equivalence between $\mathbb{E}_n$-algebras and locally constant prefactorization algebras defined on open disks in $\mathbb{R}^n$. The non-trivial homotopy groups of spheres induce additional algebraic structures on these $\mathbb{E}_n$-monoidal $C^\ast$-categories, which in the case of $\mathbb{Z}^2$ is given by a braided monoidal self-equivalence arising geometrically as a kind of `holonomy' around the circle $\mathbb{S}^1$. The locally constant prefactorization algebra structures discovered in this work generalize, under some mild geometric conditions, to other discrete spaces and thereby provide a clear link between the geometry of the localization regions and the algebraic structures on the category of superselection sectors.
\end{abstract}
\vspace{-1mm}

\vspace{-2mm}

\vspace{-2mm}

\newpage 

\tableofcontents


\section{Introduction and summary}
Topologically ordered phases of matter are a modern and exciting
research area in the intersection of condensed matter physics and mathematics~\cite{Zeng}. 
A key feature of such systems in two spatial dimensions is that their ground state admits anyonic 
excitations exhibiting braided statistics.
From a condensed matter point of view, topologically ordered phases are interesting 
because they go beyond the Landau paradigm of symmetry breaking, and their classification 
is now a major area of the condensed matter theory, see e.g.\ \cite{Wen} for a review.
Their topological features are robust against perturbations, making them a candidate 
for applications to fault-tolerant quantum computing, see e.g.\ \cite{Kitaev,Nayak}.
\sk

The mathematical description of topologically ordered phases can be approached
from different angles, which can be classified roughly into `microscopic', `mesoscopic' and `macroscopic'.
In the `microscopic' approach, one starts from suitable lattice quantum systems
and studies topological excitations of their ground state. These are then
shown to assemble into a braided monoidal category encoding the 
anyons together with their fusion and braiding, see e.g.\
\cite{Kitaev} and \cite{NaaijkensChapter} for an illustration
in the context of Kitaev's quantum double model, or~\cite{Kawagoe} for an operational approach.
In the `mesoscopic' approach, one takes physically-informed 
collective features of two-dimensional quantum gapped systems, such as their entanglement patterns,  as input
and derives from there the relevant algebraic structures encoding the fusion and 
braiding rules for anyons. This is, for instance, the approach adopted in \cite{bootstrap1,bootstrap2}. 
In the `macroscopic' approach, one provides a direct 
axiomatization of the relevant types of categories encoding 
the extended operators in topological order, typically in terms of
a structure richer than that of a braided monoidal category, such as a 
modular tensor category, see e.g.\ \cite{Johnson-Freyd}.
This bird's-eye view is very useful for addressing classification questions (in two and also higher dimensions), 
but it leaves open more concrete questions such as whether or not these 
categories can be realized by lattice quantum systems.
For further information about topological order and its interplay with category theory,
we point the reader to the informative review article \cite{Kong}.
\sk

The focus of our article is on the `microscopic' approach to topological order.
We are mainly interested in formalizing a novel and conceptual framework
to extract from a lattice quantum system the associated
category of topological excitations of the ground state, 
together with its algebraic structures such as tensor products (fusion) and braidings.
Our inspiration comes from  works which employ techniques from algebraic 
quantum field theory (AQFT) to study lattice quantum systems and their representation theory,
see in particular \cite{Naaijkens1,Naaijkens2,NaaijkensChapter,FiedlerNaaijkens,Ogata}. 
In a nutshell, the key features of this AQFT-based approach to topological order are as follows:
The starting point is a family $\{\AAA(s)\}_{s\in S}$ of $C^\ast$-algebras indexed by 
the points of a discrete set $S$, which is interpreted as a discretization of the spatial dimensions.
A typical choice is the $n$-dimensional lattice $S=\bbZ^n$, but it is an interesting question
to which extent the discrete geometry of $S$ influences the main features of the model, 
see e.g.\ \cite{Elokl} for existing results in this direction.
This family of $C^\ast$-algebras is then extended to an AQFT $\AAA$ over $S$
which assigns to each subset $U\subseteq S$ a certain $C^\ast$-algebra
obtained by forming tensor products and inductive limits. 
While the entire representation category ${}^\ast\mathbf{Rep}_{\AAA}$
of this extended AQFT is in general rather wild and unstructured, typically there
exists an interesting full subcategory which can be endowed 
with the structure of a braided or symmetric monoidal $C^\ast$-category,
under suitable additional hypotheses on both the geometry of $S$ and the AQFT $\AAA$. The key idea 
to extract this full subcategory of so-called \textit{superselection sectors} \cite{DHR,BuchholzFredenhagen}
is to fix a reference representation $\pi_0 \in {}^\ast\mathbf{Rep}_{\AAA}$, 
interpreted as the ground state representation
of the extended AQFT $\AAA$, and consider only those representations 
$\pi\in {}^\ast\mathbf{Rep}_{\AAA}$ which are unitarily equivalent to 
$\pi_0$ on the complement of a suitable class of subsets $U\subseteq S$,
i.e.\ $\pi\vert_{U^\cc} \cong \pi_0\vert_{U^\cc}$ for all $U\in \CC(S)$ belonging
to some chosen category $\CC(S)$ of localization regions. In the case of
$S=\bbZ^n$, one typically takes $\CC(S) = \Cone(\bbZ^n)$ to be the category of
cone-shaped subsets. Note that such superselection sectors encode only topological information of the system,
which can be localized in any choice of localization region independently of 
its location or size, hence one can interpret
them as describing topological excitations of the ground state. In the special
case of $S=\bbZ^n$ with $n\geq 2$, one can construct directly 
a tensor product and braiding (which is symmetric for $n\geq 3$)
on the $C^\ast$-category of superselection sectors through a computationally
intricate method which is rooted in \cite{DHR,BuchholzFredenhagen}.
These methods were recently modernized and generalized in \cite{SSS},
where in particular rather general geometric conditions on the underlying category 
of localization regions were identified which allow for the construction of a braiding.
\sk

The main innovation of our paper is the new observation that
the monoidal $C^\ast$-categories $\SSS_{(\AAA,\pi_0)}(U)$ of
superselection sectors which are strictly localized in a region $U\in\CC(S)$
carry the structure of a \textit{locally constant prefactorization algebra}
over the category $\CC(S)$ of localization regions, provided that 
certain rather mild geometric and algebraic assumptions are satisfied. 
(See Assumptions \ref{assu:0} and \ref{assu:1}, and note that 
by Theorem \ref{theo:HaagDualityPFA} the latter follows from the more standard Haag duality property.)
This means that, for every mutually disjoint family of inclusions 
of localization regions $(U_1,\dots,U_n)\to V$, i.e.\ $U_i\subseteq V$ for all $i$ and 
$U_i\cap U_j=\varnothing$ for all $i\neq j$, there exists a \textit{factorization product} 
\begin{flalign}
\bigotimes_{i=1}^n \SSS_{(\AAA,\pi_0)}(U_i)~\xrightarrow{\;\quad\;}~\SSS_{(\AAA,\pi_0)}(V)
\end{flalign}
that allows us to `fuse' superselection sectors localized in disjoint regions.
The term \textit{local constancy} refers to the property that an equivalence of $C^\ast$-categories
$\SSS_{(\AAA,\pi_0)}(U)\xrightarrow{\;\sim\;}\SSS_{(\AAA,\pi_0)}(V)$ 
is assigned to every $1$-ary inclusion $U\subseteq V$ of localization regions,
which formalizes the intuition that superselection sectors are topological objects
that are insensitive to the size and location of the localization region.
This locally constant behavior of superselection sectors 
is implicitly recognized and also computationally utilized since the early works 
\cite{DHR,BuchholzFredenhagen}, see also \cite[Theorem 3.19]{SSS} for a more explicit
statement, however its full scope becomes manifest only in 
combination with the prefactorization algebra structure which we develop in the present paper.
The algebraic structure described above can also be interpreted as a topological categorified
AQFT, see Corollary \ref{cor:2AQFT} and Remark \ref{rem:2AQFT}.
\sk

Prefactorization algebras are a very general and versatile class of algebraic structures.
They were introduced by Costello and Gwilliam \cite{CG1,CG2} in their works 
on perturbative QFT, but they are also of independent 
interest in factorization homology \cite{AyalaFrancis},
higher algebra \cite{LurieHA}, representation theory \cite{BZBJ,BZBJ2,Hataishi}
and algebraic QFT \cite{Henriques,BPSWcategorified}.
The key feature of a prefactorization algebra is that it describes a family of objects 
together with algebraic operations which are dictated intuitively by the 
underlying geometry of the problem. This geometry is abstractly encoded by a choice
of category $\CC$, collecting the `regions of interest' and their inclusions, 
together with a distinguished class of pairs of morphisms $U_1\to V \leftarrow U_2$ encoding 
when two regions $U_1$ and $U_2$ are `geometrically independent' in $V$. One of the prime
examples is given by the category $\mathbf{Open}(X)$ of open subsets of a manifold $X$,
together with the natural notion of geometric independence $U_1\to V \leftarrow U_2$
given by disjointness $U_1\cap U_2=\varnothing$ of subsets. In the context of our paper,
the geometry is encoded by the category of localization regions $\CC(S)$ 
and geometric independence by disjointness of subsets.
\sk

The reader might now rightfully ask how our locally constant prefactorization algebra structures 
relate to the more familiar braided or symmetric monoidal structures 
arising in traditional superselection theory and topological order. The key insight
is that locally constant prefactorization algebras over very simple types 
of geometries specialize to these familiar algebraic structures. For example,
it is shown in \cite[Theorem 5.4.5.9]{LurieHA} that a locally constant
prefactorization algebra which is defined over open disks in $\bbR^n$
is equivalent to the datum of an algebra over the little $n$-disks operad $\mathbb{E}_n$.
This implies that a category-valued locally constant
prefactorization algebra over open disks in $\bbR^n$ is equivalent to
an $\mathbb{E}_n$-monoidal category, i.e.\ a monoidal category for $n=1$, 
a braided monoidal category for $n=2$, and a symmetric monoidal category for 
$n\geq 3$.\footnote{This degenerate behavior of $\mathbb{E}_n$-monoidal categories 
for $n\geq 3$ is a consequence of our restriction to $1$-categories. In the context
of higher category theory, $\mathbb{E}_n$-monoidal $\infty$-categories are distinct 
concepts for all $n\geq 0$, where increasing $n$ leads to more and more commutative objects.}
When translated to our context,
this means that our locally constant prefactorization algebras of superselection
sectors are a priori much more general objects and they only specialize
to the familiar braided or symmetric monoidal categories if the 
category of localization regions $\CC(S)$ is sufficiently simple.
We will investigate these aspects in Section \ref{sec:lattice} 
for the typical example of the category $\CC(S)= \Cone(\bbZ^n)$ 
of cone-shaped subsets in $\bbZ^n$ and show how to recover the
familiar $\mathbb{E}_n$-monoidal categories in this case.
In our analysis we discover also new algebraic structures on 
these categories, related to the homotopy groups of the sphere $\bbS^{n-1}$
encoding the angular coordinates of cones, 
which do not seem to have been observed before. In the
particular case of a $2$-dimensional lattice $\bbZ^2$,
these additional algebraic structures consist of a
self-equivalence of the braided monoidal category of superselection sectors
which admits a geometric interpretation in terms of a kind of `holonomy'
around the circle $\bbS^1$.
\sk

We will now explain our results in more detail by outlining the content of this paper.
In Section~\ref{sec:prelim}, we collect the relevant preliminary material which is needed
for our work. Subsection~\ref{subsec:Cast} recalls some basic notions of 
$C^\ast$-algebras and their categorical aspects.
Subsection \ref{subsec:AQFT} introduces a concept of AQFTs over discrete 
sets $S$, such as e.g.\ the lattice $S=\bbZ^n$, and develops categorical tools 
which allow us to extend any family $\{\AAA(s)\}_{s\in S}$ of $C^\ast$-algebras
indexed by the points of $S$ to an AQFT over $S$. These tools are rooted in
the operadic approach to AQFT \cite{BSWoperad,BSWinvolutive}, which we combine
with the $C^\ast$-algebraic concepts from Subsection \ref{subsec:Cast},
see in particular the main result in Proposition \ref{prop:LKECast}.
Subsection \ref{subsec:CastCat} recalls some basic aspects of the theory of $C^\ast$-categories
following \cite{CastCat} and introduces a variant of prefactorization algebras 
taking values in $C^\ast$-categories, see in particular Definition \ref{def:PFA} and Remark \ref{rem:PFA}.
\sk

In Section \ref{sec:PFA}, we study the $C^\ast$-category 
$\SSS_{(\AAA,\CC(S),\pi_0)}\subseteq  {}^\ast\mathbf{Rep}_{\AAA}$
of superselection sectors (see Definition \ref{def:SSSloc})
of an AQFT over a discrete set $S$ and explore the algebraic structures
which may be defined on this category. One of the key features entering our approach,
which is implicitly known since \cite{DHR,BuchholzFredenhagen} and explicitly stated in \cite[Facts 2.26, (SSS1)]{SSS},
is that there exists a family of unitarily equivalent models 
$\SSS_{(\AAA,\pi_0)}(U)\xrightarrow{\;\sim\;}\SSS_{(\AAA,\CC(S),\pi_0)}$
for this $C^\ast$-category, indexed by the objects $U\in\CC(S)$ 
of the category of localization regions (e.g.\ cones for $S=\bbZ^n$), that describes
superselection sectors which are strictly localized in $U$. This family of $C^\ast$-categories 
assembles into a functor $\SSS_{(\AAA,\pi_0)} : \CC(S)\to \CastCat$ which is locally constant
in the sense that a unitary equivalence $\SSS_{(\AAA,\pi_0)}(U)\xrightarrow{\;\sim\;}\SSS_{(\AAA,\pi_0)}(V)$
of $C^\ast$-categories is assigned to every inclusion $U\subseteq V$ in $\CC(S)$, 
see Lemma \ref{lem:SSSU}. In Subsection \ref{subsec:PFAstructure},
we show that this functor can be extended to a prefactorization algebra 
$\SSS_{(\AAA,\pi_0)} : \P_{\CC(S)^\perp}\to \CastCat$ over the category
of localization regions $\CC(S)^{\perp}$ 
(endowed with the orthogonality relation ${\perp}$ given by disjointness),
provided that certain geometric and algebraic assumptions are satisfied, see Proposition \ref{prop:PFA}. 
On the one hand, our geometric Assumption \ref{assu:0} is rather mild (see also Remark \ref{rem:assu0}) and it is
satisfied by the category of cone-shaped subsets of $S=\bbZ^n$, 
see Proposition \ref{prop:cones}. On the other hand, our algebraic Assumption \ref{assu:1} 
is designed  to facilitate the construction of the factorization products. 
In particular, it holds if the AQFT satisfies Haag duality for all localization regions 
$U\in \CC(S)$, see Theorem \ref{theo:HaagDualityPFA}.
\sk

In Subsections \ref{subsec:monoidal} and \ref{subsec:compatibility}, we review the usual 
monoidal structure on the $C^\ast$-category of superselection sectors from \cite{DHR,BuchholzFredenhagen}
and prove that it is compatible with our prefactorization algebra structure. This compatibility implies
that we obtain a locally constant prefactorization algebra 
$\SSS_{(\AAA,\pi_0)} : \P_{\CC(S)^\perp}\to\Alg_{\mathsf{uAs}}\big(\CastCat\big)$
taking values in the category of strict monoidal $C^\ast$-categories, see Theorem \ref{theo:PFAinMonCat}.
In contrast to the geometric origin of our prefactorization algebra structure,
this monoidal structure is of an analytic nature. It is rooted
in the fact that, assuming Haag duality, every superselection sector $\pi$ can be 
presented by a suitable endomorphism $\rho$, hence one can define a monoidal product
by composition of endomorphisms. The existence of these two compatible types of algebraic structures, i.e.\
a geometric prefactorization algebra structure and an analytic object-wise monoidal structure,
is the key insight to recover from our perspective the results 
from \cite{Naaijkens1,Naaijkens2,NaaijkensChapter,FiedlerNaaijkens,Ogata} 
about the braided monoidal structures on the categories 
of superselection sectors in $2$-dimensional lattice quantum systems.
\sk

In Section \ref{sec:lattice}, we specialize our results to the case where the set $S=\bbZ^n$
is the $n$-dimensional lattice and the category of localization regions $\CC(S)=\Cone(\bbZ^n)$ is given by 
cone-shaped subsets of $\bbZ^n$. In this context, our general results from Section \ref{sec:PFA} 
imply that, assuming Haag duality for all cone-shaped subsets, the categories of localized 
superselection sectors carry the structure of a locally constant prefactorization algebra
$\SSS_{(\AAA,\pi_0)} : \P_{\Cone(\bbZ^n)^\perp}\to\Alg_{\mathsf{uAs}}\big(\CastCat\big)$
taking values in the category of strict monoidal $C^\ast$-categories, see Corollary \ref{cor:latticePFA}.
In Subsection \ref{subsec:infty}, we use powerful techniques from higher algebra, in particular
from $\infty$-categories and $\infty$-operads \cite{LurieHA},
in order to relate these algebraic structures to more familiar ones.
One of the key results is Corollary \ref{cor:latticePFAcylinder}, which shows that 
$\SSS_{(\AAA,\pi_0)}$ has an underlying locally constant prefactorization algebra 
$\und{\SSS}_{(\AAA,\pi_0)}$ defined on open disks in the cylinder $\bbR^1\times \bbS^{n-1}$.
We would like to emphasize that the sphere $\bbS^{n-1}$ arises from the geometry of our context
(it describes the angular coordinates of cones centered at the origin of $\bbZ^n$),
while the line $\bbR^1$ has an analytic origin rooted in Haag duality and the associated 
object-wise monoidal structure. By removing a point of the sphere 
$\bbR^1\times(\bbS^{n-1}\setminus\pt) \cong \bbR^{n}$, we obtain an underlying
locally constant prefactorization algebra $\und{\und{\SSS}}_{(\AAA,\pi_0)}$ 
defined on open disks in $\bbR^n$, which is equivalent
to an $\mathbb{E}_n$-monoidal structure on the $C^\ast$-category of superselection sectors, 
see Corollary \ref{cor:latticeEn}. This provides a conceptual explanation
for the origin of the braided (for $n=2$) or symmetric (for $n\geq 3$) monoidal categories
of superselection sectors in lattice quantum systems on $\bbZ^n$
\cite{Naaijkens1,Naaijkens2,NaaijkensChapter,FiedlerNaaijkens,Ogata}.
It is important to highlight that these $\mathbb{E}_n$-monoidal categories 
are obtained by removing a point of the sphere $\bbS^{n-1}$, hence they forget those 
algebraic structures carried by the locally constant prefactorization algebra $\und{\SSS}_{(\AAA,\pi_0)}$
on $\bbR^1\times\bbS^{n-1}$ which are linked to the homotopy groups of $\bbS^{n-1}$. 
In the special case of a $2$-dimensional lattice $\bbZ^2$, the additional algebraic structure
arising from the circle $\bbS^1$ can be identified explicitly with a self-equivalence
of the braided monoidal category of superselection sectors, see Corollary \ref{cor:selfequivalence}.
This self-equivalence can be interpreted as a kind of `holonomy' arising from 
rotating cones in $\bbZ^2$ by $2\pi$ around their apex.
\sk

In Subsection \ref{subsec:explicit}, we spell out some more 
concrete computational details for the case of $\bbZ^2$. In particular, we show 
how one can compute the braided monoidal structure on the category of superselection sectors  $\SSS_{(\AAA,\CC(S),\pi_0)}$
explicitly from our locally constant prefactorization algebra 
$\SSS_{(\AAA,\pi_0)} : \P_{\Cone(\bbZ^2)^\perp}\to\Alg_{\mathsf{uAs}}\big(\CastCat\big)$
and observe that the result 
agrees with the traditional approach in \cite{DHR,BuchholzFredenhagen}.
We further explain how to compute the self-equivalence of this braided monoidal category 
which is associated to the circle $\bbS^1$ and show that it is trivial 
(i.e.\ the identity) for Kitaev's quantum double model
for an Abelian group, see Example \ref{ex:Kitaevmonodromy}.
We believe that this triviality result is not due to 
the simplicity of this example, but it is rooted deeper in the type
of superselection sectors which are commonly used for lattice quantum systems.
These correspond to the traditional superselection sectors from \cite{DHR,BuchholzFredenhagen},
which in other contexts are known to neglect information about the underlying `spacetime' topology. 
In more detail, such traditional superselection sectors
correspond in the terminology of Brunetti and Ruzzi to `topologically trivial sectors' \cite{BrunettiRuzzi},
which are a subclass of the more general superselection sectors introduced in their work
building upon ideas of Roberts \cite{Roberts} on non-Abelian cohomology.
We believe that adapting these more general superselection sectors to lattice quantum systems
will produce a richer $C^\ast$-category of superselection sectors that depends 
on the homotopy groups of the category of cone-shaped subsets $\Cone(\bbZ^2)$
and thereby leads to a non-trivial self-equivalence in Corollary \ref{cor:selfequivalence}.
We hope to come back to this in future work.


\section{\label{sec:prelim}Preliminaries}

\subsection{\label{subsec:Cast}$C^\ast$-algebras}
In this subsection we recall some basic facts about 
$\ast$-algebras, pre-$C^\ast$-algebras and $C^\ast$-algebras
over the field $\bbC$ of complex numbers, see e.g.~\cite{Murphy}.
For our purposes, it will be most convenient to approach this subject 
from a categorical point of view following \cite{Bunke,BunkeLecture}.
\begin{defi}\label{def:astAlg}
A \textit{$\ast$-algebra} is a unital associative algebra $A$ over $\bbC$
with a complex antilinear involution $\ast : A\to A\,,~a\mapsto a^\ast$
that reverses the order of multiplication,
i.e.\ 
\begin{flalign}
(\lambda \,a + \lambda^\prime\, a^\prime)^\ast \,=\, \overline{\lambda} \, a^\ast + 
\overline{\lambda^\prime}\,a^{\prime \ast}
~,\quad
(a\,a^\prime)^\ast \,=\, a^{\prime\ast}\,a^\ast~,\quad
a^{\ast\ast}\,=\,a~,
\end{flalign}
for all $a,a^\prime\in A$ and $\lambda,\lambda^\prime\in\bbC$.
We denote by $\astAlg$ the category whose objects are all 
$\ast$-algebras and whose morphisms are all $\ast$-homomorphisms, i.e.\ 
$f:A\to B$ is a morphism of unital associative algebras such that 
$f(a^\ast) = f(a)^\ast$, for all $a\in A$.
\end{defi}

Let us recall that the traditional definition 
of a $C^\ast$-algebra is that of a Banach $\ast$-algebra $A$
whose norm $\norm{\cdot}_A: A\to [0,\infty)$ satisfies the 
$C^\ast$-identity $\norm{a^\ast\,a}_A = \norm{a}_A^2$, for all $a\in A$.
The norm of a $C^\ast$-algebra turns out to be fixed uniquely
by its other structures, see e.g.\ \cite[Corollary 2.1.2]{Murphy},
and every $\ast$-homomorphism $f:A\to B$
between two $C^\ast$-algebras is automatically continuous, 
see e.g.\ \cite[Theorem 2.1.7]{Murphy}. This justifies the following definition.
\begin{defi}\label{def:CastAlg}
A \textit{$C^\ast$-algebra} is $\ast$-algebra $A\in \astAlg$ 
which has the property that it admits a
$C^\ast$-norm turning it into a Banach $\ast$-algebra.
The category of $C^\ast$-algebras is defined as the full subcategory
\begin{flalign}
\CastAlg\,\subseteq\,\astAlg
\end{flalign}
of the category of $\ast$-algebras from Definition \ref{def:astAlg}
whose objects are all $C^\ast$-algebras.
\end{defi}

For every $\ast$-algebra $A\in \astAlg$, one can define a function
$\norm{\cdot}_{\max} : A\to [0,\infty]$ by setting
\begin{flalign}\label{eqn:maxnorm}
\norm{a}_{\max}\,:=\,\sup_{f: A\to B}\norm{f(a)}_B~,
\end{flalign}
for all $a\in A$, where the supremum is taken over all $\ast$-homomorphisms
$f: A\to B$ to a $C^\ast$-algebra $B\in\CastAlg$ with $C^\ast$-norm denoted by $\norm{\cdot}_B$.
\begin{defi}\label{def:preCastAlg}
A \textit{pre-$C^\ast$-algebra} is a $\ast$-algebra $A\in \astAlg$ such that $\norm{a}_{\max} < \infty$
is finite, for all $a\in A$. The category of pre-$C^\ast$-algebras is defined as the full subcategory
\begin{flalign}
\preCastAlg\,\subseteq\,\astAlg
\end{flalign}
of the category of $\ast$-algebras from Definition \ref{def:astAlg}
whose objects are all pre-$C^\ast$-algebras.
\end{defi}
\begin{rem}\label{rem:chainofinclusions}
For every $C^\ast$-algebra $A\in\CastAlg$ one can show that $\norm{\cdot}_{\max}  = \norm{\cdot}_A$
agrees with its $C^\ast$-norm. This implies that the full subcategories
from Definitions \ref{def:CastAlg} and \ref{def:preCastAlg} satisfy
\begin{flalign}
\CastAlg\,\subseteq\, \preCastAlg\,\subseteq\, \astAlg~,
\end{flalign}
see e.g.\ \cite[Example 2.11]{Bunke} or \cite[Lemma 7.7]{BunkeLecture}.
\end{rem}

The following useful result has been shown in \cite[Section 3]{Bunke} 
and \cite[Section 7]{BunkeLecture}, which also provides 
explicit models for the adjoint functors displayed below.
\begin{propo}\label{prop:Castadjunctions}
\begin{itemize}
\item[(a)] There exists an adjunction
\begin{equation}
\begin{tikzcd}
\mathrm{incl}\,:\,\preCastAlg
\ar[r, shift left=0.75ex] & \ar[l, shift left=0.75ex]
\astAlg\,:\,\mathrm{Bd}^\infty
\end{tikzcd}
\end{equation}
whose left adjoint $\mathrm{incl}$ is the full subcategory inclusion $\preCastAlg\subseteq \astAlg$
from Definition~\ref{def:preCastAlg}. Hence, the category of pre-$C^\ast$-algebras 
$\preCastAlg$ is a coreflective full subcategory of the category of $\ast$-algebras $\astAlg$.

\item[(b)] There exists an adjunction
\begin{equation}
\begin{tikzcd}
\mathrm{compl}\,:\, \preCastAlg 
\ar[r, shift left=0.75ex] & \ar[l, shift left=0.75ex]
\CastAlg\,:\,\mathrm{incl}
\end{tikzcd}
\end{equation}
whose right adjoint $\mathrm{incl}$ is the full subcategory inclusion
$\CastAlg\subseteq \preCastAlg$ from Remark \ref{rem:chainofinclusions}.
Hence, the category of $C^\ast$-algebras 
$\CastAlg$ is a reflective full subcategory of the category of pre-$C^\ast$-algebras $\preCastAlg$.
\end{itemize}
\end{propo}

Using that the category of $\ast$-algebras $\astAlg$ is both complete and cocomplete
(i.e.\ bicomplete) and a general categorical argument as in \cite[Proposition 4.5]{BunkeReview}, 
it follows that the categories $\preCastAlg$ and $\CastAlg$ are bicomplete too. One can say even
more about how limits and colimits may be computed, see also
\cite[Section 8]{Bunke} and \cite[Section 7]{BunkeLecture}.
\begin{cor}\label{cor:bicompleteness}
The categories $\preCastAlg$ and $\CastAlg$ are bicomplete. Furthermore:
\begin{itemize}
\item[(a)] Colimits in $\preCastAlg$ are 
created by the inclusion functor $\mathrm{incl}:\preCastAlg\to \astAlg$ and 
the limit of a diagram $X : \DD \to \preCastAlg$ can be computed by
\begin{flalign}
\lim\Big(X : \DD \to \preCastAlg\Big)\,\cong\, \mathrm{Bd}^\infty\Big(
\lim\Big(X : \DD \to \astAlg\Big)\Big)
\end{flalign}
from the limit of this diagram in $\astAlg$.

\item[(b)] Limits in $\CastAlg$ are 
created by the inclusion functor $\mathrm{incl}:\CastAlg\to \preCastAlg$ 
and the colimit of a diagram $X : \DD \to \CastAlg$ can be computed by
\begin{flalign}
\colim\Big(X : \DD \to \CastAlg\Big)\,\cong\, \mathrm{compl}\Big(
\colim\Big(X : \DD \to \preCastAlg\Big)\Big)
\end{flalign}
from the colimit of this diagram in $\preCastAlg$.
\end{itemize}
\end{cor}

\begin{rem}\label{rem:bicompleteness}
Combining items (a) and (b) from Corollary \ref{cor:bicompleteness},
it follows that the colimit of a diagram $X : \DD \to \CastAlg$ can be computed by
\begin{flalign}
\colim\Big(X : \DD \to \CastAlg\Big)\,\cong\, \mathrm{compl}\Big(
\colim\Big(X : \DD \to \astAlg\Big)\Big)
\end{flalign}
from the colimit of this diagram in $\astAlg$. 
\end{rem}

\subsection{\label{subsec:AQFT}Algebraic quantum field theories over sets}
In this subsection we describe a variant of algebraic quantum
field theories (AQFTs) in the context where the 
underlying ``space(time)''\footnote{In our context, it is more accurate to think about the underlying space only.
The time direction does not play a significant role, and is only used 
to inform the selection of a suitable reference representation.
The AQFTs we consider could hence be thought of as living on a lattice over a Cauchy surface.}
is given by a possibly infinite set $S\in \Set$, 
such as e.g.\ the lattice $\bbZ^n$. For this we use some aspects of 
the operadic approach to AQFT from \cite{BSWoperad} and its generalization
to the $\ast$-involutive setting from \cite{BSWinvolutive}.
The aim of this subsection is to fix our notations and formalize through a universal property
(see Proposition \ref{prop:LKECast}) a standard, but rather ad hoc, extension construction 
(see Remark \ref{rem:qspin}) which is frequently used in the context of quantum spin systems.
We believe that our perspective will be particularly valuable for categorically-minded
readers to understand conceptually this extension construction,
however we do not expect that it will lead to new 
practical or computational advantages for quantum spin systems.
\begin{defi}\label{def:SubS}
Given any set $S\in \Set$, we denote by $\Sub(S)$ the category
whose objects are all subsets $U\subseteq S$ and whose morphisms
are subset inclusions $U\subseteq V$. We say that two morphisms
$U_1\subseteq V$ and $U_2\subseteq V$ are \textit{orthogonal},
written as $(U_1\subseteq V)\perp(U_2\subseteq V)$, if 
$U_1\cap U_2=\varnothing$ are disjoint.
\end{defi}

\begin{defi}\label{def:AQFT}
A \textit{$\ast$-AQFT\,\footnote{The terminology $\ast$-AQFT
is chosen as in \cite{BSWinvolutive} to emphasize the presence of
$\ast$-involutive structures and thereby to distinguish this concept
from the `purely algebraic' AQFTs (taking values in an arbitrary symmetric monoidal category)
appearing in other works of some of the authors, 
see e.g.\ \cite{BSWoperad} and \cite{Carmona}.} 
over a set $S\in \Set$} is a functor
$\AAA : \Sub(S)\to \astAlg$ from the category of subsets of $S$ (see Definition \ref{def:SubS})
to the category of $\ast$-algebras (see Definition \ref{def:astAlg})
which satisfies the following $\perp$-commutativity property: 
For all orthogonal pairs $(U_1\subseteq V)\perp (U_2\subseteq V)$,
the diagram
\begin{equation}\label{eqn:perpcommutativity}
\begin{tikzcd}
\ar[d] \AAA(U_1)\otimes \AAA(U_2) \ar[r]& \AAA(V)\otimes \AAA(V)\ar[d,"\mu^\op"] \\
\AAA(V)\otimes\AAA(V) \ar[r,"\mu"']& \AAA(V)
\end{tikzcd}
\end{equation}
of linear maps between vector spaces commutes, where $\mu^{(\op)}$ denotes the (opposite)
multiplication in the $\ast$-algebra $\AAA(V)\in\astAlg$
and the unlabeled arrows are induced by the $\ast$-homomorphisms
$\AAA(U_i)\to \AAA(V)$ obtained by applying the functor $\AAA$ to the 
subset inclusions $U_i\subseteq V$, for $i=1,2$.
The category of $\ast$-AQFTs over $S$ is defined as the full subcategory
\begin{flalign}
\astAQFT(S)\,\subseteq\, \Fun\big(\Sub(S),\astAlg\big)
\end{flalign}
of the functor category whose objects are all $\ast$-AQFTs over $S$.
\end{defi}

\begin{rem}\label{rem:AQFT}
This definition of $\ast$-AQFTs in terms of functors satisfying 
the $\perp$-commutativity property is similar to the traditional 
description of AQFTs in \cite{HaagKastler}. 
(In this relativistic context, the orthogonality relation $\perp$ is 
given by causal disjointness of subsets of the Minkowski spacetime.)
There also exists an equivalent, but more elegant and powerful, definition of ($\ast$-)AQFTs
in terms of ($\ast$-)algebras over ($\ast$-)operads, see \cite{BSWoperad} and \cite{BSWinvolutive}.
Indeed, Definition \ref{def:SubS} defines an orthogonal category $\Sub(S)^\perp := (\Sub(S),\perp)$
in the sense of \cite[Definition 3.4]{BSWoperad}, hence there exists an associated
AQFT ($\ast$-)operad $\O_{\Sub(S)^\perp}$, see \cite[Definition 3.8]{BSWoperad}
and \cite[Proposition 7.8]{BSWinvolutive}. 
The result in \cite[Proposition 7.11]{BSWinvolutive} then shows that 
the category from Definition \ref{def:AQFT} is equivalent
\begin{flalign}\label{eqn:operadalgebra}
\astAQFT(S) \,\simeq\, {}^\ast\Alg_{\O_{\Sub(S)^\perp}}\big(\Vec_\bbC\big)
\end{flalign}
to the category of $\ast$-algebras over the AQFT $\ast$-operad with values
in the involutive closed symmetric monoidal category $\Vec_\bbC$ of complex vector spaces.
\end{rem}

The operadic perspective from Remark \ref{rem:AQFT} allows
us to provide a conceptual formalization of a useful 
construction in the context of lattice quantum systems 
which extends a family of $\ast$-algebras $\{\AAA(s)\}_{s\in S}$
indexed by the points of $S$ to a $\ast$-AQFT over $S$, see e.g.\ 
\cite{Naaijkens1,Naaijkens2,NaaijkensChapter,FiedlerNaaijkens,Ogata}
and also Remark \ref{rem:qspin} below.
Regarding the set $S\in\Set$ as a discrete category,
we denote by 
\begin{flalign}\label{eqn:jfunctor}
j\,:\, S~\xrightarrow{\;\quad\;}~\Sub(S)~,~~s~\xmapsto{\;\quad\;}~j(s)\,:=\,\Big(\{s\}\subseteq S\Big)
\end{flalign}
the functor which assigns to an element $s\in S$ 
the singleton subset $j(s) = \{s\}\in\Sub(S)$.
Note that this functor is fully faithful and that 
restricting the orthogonality relation $\perp$ on $\Sub(S)$ from Definition
\ref{def:SubS} defines the trivial (i.e.\ empty) orthogonality relation on $S$.
The result in \cite[Proposition 4.6]{BSWoperad}, combined with its 
$\ast$-involutive generalization in \cite[Corollary 7.13]{BSWinvolutive},
then specializes in our context to the following
\begin{propo}\label{prop:LKE}
The orthogonal functor \eqref{eqn:jfunctor} induces an adjunction
\begin{equation}
\begin{tikzcd}
j_!\,:\, \astAlg^S 
\ar[r, shift left=0.75ex] & \ar[l, shift left=0.75ex]
\astAQFT(S)\,:\,j^\ast
\end{tikzcd}
\end{equation}
which exhibits the product category $\astAlg^S := \prod_{s\in S}\astAlg$
as a coreflective full subcategory of the category of $\ast$-AQFTs from Definition \ref{def:AQFT},
i.e.\ the left adjoint $j_!$, which is given by operadic left Kan extension, is a fully faithful functor.
The right adjoint functor $j^\ast$ is given concretely by assigning 
to each object $\AAA\in\astAQFT(S)$ the family of $\ast$-algebras 
$j^\ast(\AAA) = \big\{\AAA(\{s\})\big\}_{s\in S}\in \astAlg^S$. 
\end{propo}

We would like to emphasize that operadic constructions as
in Proposition \ref{prop:LKE} exist a priori only for AQFTs assigning 
$\ast$-algebras, but not necessarily for AQFTs assigning (pre-)$C^\ast$-algebras. 
The reason for this is that (pre-)$C^\ast$-algebras can not, to the best of our knowledge, be presented as 
categories of $\ast$-algebras in a suitable cocomplete involutive closed symmetric 
monoidal category, which is however essential for the operadic approach in \cite{BSWoperad,BSWinvolutive}.
In the remaining part of this subsection we will study the interplay between the constructions
presented above and the (pre-)$C^\ast$-algebraic concepts from Subsection \ref{subsec:Cast}.
For this we introduce the following terminology.
\begin{defi}\label{def:CastAQFT}
An object $\AAA\in\astAQFT(S)$ is called a \textit{(pre-)$C^\ast$-AQFT} if, for each $U\in\Sub(S)$,
the $\ast$-algebra $\AAA(U)\in C^\ast_{(\mathrm{pre})}\Alg_\bbC \subseteq \astAlg$ is contained 
in the full subcategory of (pre-)$C^\ast$-algebras from Definitions  \ref{def:CastAlg} and  \ref{def:preCastAlg}. 
We denote by
\begin{flalign}
\CastAQFT(S)\,\subseteq\,\preCastAQFT(S) \,\subseteq\,\astAQFT(S)
\end{flalign}
the full subcategories of $C^\ast$-AQFTs and pre-$C^\ast$-AQFTs over $S$.
\end{defi}

\begin{lem}\label{lem:AQFTcompletion}
The adjunction $\mathrm{compl}:\preCastAlg \rightleftarrows \CastAlg : \mathrm{incl}$ 
from Proposition \ref{prop:Castadjunctions} (b) induces via post-composition an adjunction
\begin{equation}
\begin{tikzcd}
\mathrm{compl}\,:\,\preCastAQFT(S)
\ar[r, shift left=0.75ex] & \ar[l, shift left=0.75ex]
\CastAQFT(S)\,:\,\mathrm{incl}
\end{tikzcd}
\end{equation}
between the categories of pre-$C^\ast$-AQFTs and $C^\ast$-AQFTs.
\end{lem}
\begin{proof}
We have to show that, given any $\AAA\in \preCastAQFT(S)$, post-composing with the completion functor yields a functor
$\mathrm{compl}(\AAA) : \Sub(S)\to \CastAlg$ which satisfies the $\perp$-commutativity property \eqref{eqn:perpcommutativity}.
Let us recall from e.g.\ \cite[Lemma 7.17]{BunkeLecture}
that, given any $A\in\preCastAlg$, its completion
$\mathrm{compl}(A)\in \CastAlg$ can be modeled by forming
the quotient $A/I_A\in \astAlg$ by the $\ast$-ideal $I_A := \{a\in A\,:\,\norm{a}_{\max}=0\}$
and then completing the result with respect to the norm $\norm{\cdot}_{\max}$.
Since all maps in the analog of the diagram \eqref{eqn:perpcommutativity}
for the completions are continuous, it suffices to verify that the diagram
\begin{equation}
\begin{tikzcd}
\ar[d] \faktor{\AAA(U_1)}{I_{\AAA(U_1)}} \otimes \faktor{\AAA(U_2)}{I_{\AAA(U_2)}}  
\ar[r]& \faktor{\AAA(V)}{I_{\AAA(V)}} \otimes \faktor{\AAA(V)}{I_{\AAA(V)}} \ar[d,"\mu^\op"] \\
\faktor{\AAA(V)}{I_{\AAA(V)}} \otimes \faktor{\AAA(V)}{I_{\AAA(V)}}\ar[r,"\mu"']& \faktor{\AAA(V)}{I_{\AAA(V)}}
\end{tikzcd}
\end{equation}
of the uncompleted quotients commutes, for all $(U_1\subseteq V)\perp(U_2\subseteq V)$.
This is a direct consequence of the $\perp$-commutativity property of $\AAA\in \preCastAQFT(S)$. 
\sk

The resulting functor $\mathrm{compl}\colon \preCastAQFT(S)\to \CastAQFT(S)$ 
is clearly left adjoint to the canonical inclusion functor. 
\end{proof}

The following result adapts Proposition \ref{prop:LKE}
to the $C^\ast$-algebraic context.
\begin{propo}\label{prop:LKECast}
The adjunction from Proposition \ref{prop:LKE} together with the completion functor $\mathrm{compl}:
\preCastAQFT(S)\to\CastAQFT(S)$ from Lemma \ref{lem:AQFTcompletion} induce the adjunction
\begin{equation}\label{eqn:j!Cast}
\begin{tikzcd}
\mathrm{compl}\circ j_!\, :\, 
\CastAlg^S \ar[r, shift left=0.75ex] & \ar[l, shift left=0.75ex]
\CastAQFT(S)\,:\, j^\ast
\end{tikzcd}
\end{equation}
between the corresponding $C^\ast$-algebraic full subcategories.
This adjunction exhibits the product category $\CastAlg^S := \prod_{s\in S}\CastAlg$
as a coreflective full subcategory of $\CastAQFT(S)$, i.e.\ the left adjoint
$\mathrm{compl}\circ j_!$ is a fully faithful functor.
\end{propo}
\begin{proof}
From the definition of the right adjoint functor $j^\ast(\AAA) = \big\{\AAA(\{s\})\big\}_{s\in S}$, 
it is evident that it restricts to the $C^\ast$-algebraic full subcategories.
It thus remains to show that our candidate in \eqref{eqn:j!Cast} for a left adjoint 
is well-defined and that it is indeed a left adjoint for the restricted $j^\ast$.
\sk

To show well-definedness, let us start with observing that the 
orthogonal functor \eqref{eqn:jfunctor} factorizes
\begin{equation}
\begin{tikzcd}
j\,:\, S \ar[r,"j_1"]& \Fin(S) \ar[r,"j_2"] & \Sub(S)
\end{tikzcd}
\end{equation}
through the full subcategory $\Fin(S)\subseteq \Sub(S)$ of all \textit{finite} subsets of $S$, which we endow 
with the restriction of the orthogonality relation $\perp$ on $\Sub(S)$.
This implies that the operadic left Kan extension $j_!$ can be computed in two steps
\begin{equation}\label{eqn:jfactorized}
\begin{tikzcd}
j_!\,:\, \astAlg^S \ar[r,"j_{1!}"]& \astAQFT^{\mathrm{fin}}(S) \ar[r,"j_{2!}"] & \astAQFT(S)
\end{tikzcd}~,
\end{equation}
where $\astAQFT^{\mathrm{fin}}(S)$ is defined similarly to Definition \ref{def:AQFT}
by replacing the orthogonal category $\Sub(S)^\perp$ of all subsets 
with the orthogonal category $\Fin(S)^\perp$ of all finite subsets.
The first operadic left Kan extension $j_{1!}$ can be computed explicitly
as follows: To a family of $\ast$-algebras $\AAA = \{\AAA(s)\}_{s\in S}\in \astAlg^S$, it assigns the object
$j_{1!}(\AAA)\in \astAQFT^{\mathrm{fin}}(S)$ which is defined by the functor 
$j_{1!}(\AAA) : \Fin(S)\to \astAlg$ that assigns to every finite subset 
$U\subseteq S$ the (finite!) tensor product $\ast$-algebra
\begin{subequations}\label{eqn:explicitj!}
\begin{flalign}\label{eqn:tensoralgebra}
j_{1!}(\AAA)(U) \,=\, \bigotimes_{s\in U}\AAA(s)\,\in\,\astAlg 
\end{flalign}
and to every inclusion $U\subseteq V$ of finite subsets the canonical $\ast$-homomorphism.
Since $j_2$ satisfies the closedness condition from \cite[Definition 5.3]{BSWoperad},
the result in \cite[Theorem 5.4]{BSWoperad} shows that the
second operadic left Kan extension $j_{2!}$ can be computed by an ordinary \textit{categorical}
left Kan extension of $\astAlg$-valued functors. More concretely, this means that
to an object $\BBB \in \astAQFT^{\mathrm{fin}}(S)$ it assigns the object
$j_{2!}(\BBB)\in \astAQFT(S)$ which is defined by the functor 
$j_{2!}(\BBB) : \Sub(S)\to \astAlg$ that assigns to every subset 
$U\subseteq S$ the filtered colimit
\begin{flalign}\label{eqn:colimitalgebra}
j_{2!}(\BBB)(U) \,=\,\colim\Big(
\Fin(S)/U ~\xrightarrow{\quad\quad}~\Fin(S) ~\xrightarrow{\quad \BBB\quad} ~\astAlg
\Big) 
\end{flalign}
\end{subequations}
and to every inclusion $U\subseteq V$ of subsets the canonical $\ast$-homomorphism.
\sk

Given any family $\AAA= \{\AAA(s)\}_{s\in S}\in \CastAlg^S$ of $C^\ast$-algebras,
the finite algebraic tensor products in \eqref{eqn:tensoralgebra} are pre-$C^\ast$-algebras,
see e.g.\ \cite[Lemma 7.13]{BunkeLecture}. Hence, the functor $j_{1!}$ restricts to a functor
$j_{1!} : \CastAlg^S\to\preCastAQFT^{\mathrm{fin}}(S)$. Using also
\eqref{eqn:colimitalgebra} and Corollary \ref{cor:bicompleteness} (a), 
it follows that the functor $j_{2!}$ restricts to a functor
$j_{2!} : \preCastAQFT^{\mathrm{fin}}(S)\to \preCastAQFT(S)$. Hence,
by \eqref{eqn:jfactorized} one obtains the restriction
\begin{flalign}
j_!\,:\, \CastAlg^S~\xrightarrow{\;\quad\;}~\preCastAQFT(S)~,
\end{flalign}
which implies that our candidate for the left adjoint functor in \eqref{eqn:j!Cast} is well-defined.
\sk

To show that \eqref{eqn:j!Cast} is an adjunction, we observe that, 
for any pair of objects $\AAA\in \CastAlg^S$ and $\CCC\in \CastAQFT(S)$, there exist natural bijections
\begin{flalign}
\Hom_{\astAQFT(S)}\big(\mathrm{compl}\big(j_!(\AAA)\big),\CCC\big)\,\cong\,\Hom_{\astAQFT(S)}\big(j_!(\AAA),\CCC\big)\,
\cong\, \Hom_{\astAlg^S}\big(\AAA,j^\ast(\CCC)\big)~,
\end{flalign} 
where in the first step we use Lemma \ref{lem:AQFTcompletion}
and in the second step we use Proposition \ref{prop:LKE}. 
To show fully faithfulness of the left adjoint functor, one 
can equivalently verify that the adjunction unit is a natural isomorphism.
This follows from the observation that,
for each $ \AAA\in \CastAlg^S$, the unit 
$\eta_\AAA : \AAA \to j^\ast\mathrm{compl} j_!(\AAA)$
is given component-wise by the identity maps 
\begin{flalign}
\AAA(s)  ~\xrightarrow{\;\quad\;}~ \mathrm{compl}\big(j_!(\AAA)(\{s\})\big) \,=\, 
\mathrm{compl}\big(\AAA(s)\big) \,=\, \AAA(s)~,
\end{flalign} 
for all $s\in S$, where in the first equality we use \eqref{eqn:explicitj!}
and in the second equality we use that $\AAA(s)\in \CastAlg$ 
is by hypothesis a $C^\ast$-algebra.
\end{proof}

\begin{rem}\label{rem:qspin}
We would like to highlight that the
left adjoint functor from Proposition \ref{prop:LKECast}
provides a formalization of the following standard extension construction
in the context of lattice quantum systems, see e.g.\ 
\cite{BRI} and \cite{Naaijkens1,Naaijkens2,NaaijkensChapter,FiedlerNaaijkens,Ogata}.
Any family of matrix algebras $\{\AAA(s)\cong \mathrm{Mat}_d(\bbC)\}_{s\in S}\in \CastAlg^S$,
indexed by the points of the set $S$,
can be extended to all \textit{finite} subsets $U^{\mathrm{fin}}\subseteq S$ 
by forming the tensor product $C^\ast$-algebra $\AAA(U^{\mathrm{fin}}) := 
\bigotimes_{s\in U^{\mathrm{fin}}} \AAA(s)$. This can be extended
further to all (possibly infinite) subsets $U \subseteq S$ by taking
the inductive limit $\AAA(U) := \colim_{U^{\mathrm{fin}}\subseteq U}^{}\,\AAA(U^{\mathrm{fin}})$
over all finite subsets $U^{\mathrm{fin}}\subseteq U$. One easily checks by hand 
that this extension procedure defines a $C^\ast$-AQFT over $S$.
From the proof of Proposition \ref{prop:LKECast}, it is evident that the result 
of this ad hoc construction coincides with the application of the left adjoint
functor $\mathrm{compl} \circ j_!$ on $\AAA=\{\AAA(s)\}_{s\in S}\in \CastAlg^S$, hence
our extension construction agrees with the usual one from the literature.
Some of the advantages of our more abstract point of view are as follows: 
1.)~The extended $\mathrm{compl} \big(j_!(\AAA)\big)\in \CastAQFT(S)$ has, and 
it is characterized uniquely (up to isomorphism) by, the following useful universal property:
For every $\BBB\in\CastAQFT(S)$, 
the set of $\CastAQFT(S)$-morphisms $\mathrm{compl} \big(j_!(\AAA)\big)\to \BBB$
is in bijective correspondence to the set of families 
$\{\AAA(s)\to \BBB(\{s\})\}_{s\in S}$ of $\CastAlg$-morphisms.
2.)~The concept of operadic left Kan extension provides a conceptual explanation
for why one first extends from points to finite subsets via tensor products
and then from finite subsets to all subsets via inductive limits.
\end{rem}

\subsection{\label{subsec:CastCat}$C^\ast$-categorical prefactorization algebras}
In this subsection we recall some well-known aspects of the theory of $C^\ast$-categories~\cite{GhezLimaRoberts},
following mostly the presentation in \cite{CastCat}.
See also \cite{Bunke,BunkeLecture,BunkeReview} for related results.
We will then introduce a concept of prefactorization algebras 
which are defined over an orthogonal category $\CC^\perp = (\CC,\perp)$,
for example the one from Definition \ref{def:SubS} or full orthogonal
subcategories thereof, and take values in $C^\ast$-categories.
\begin{defi}\label{def:astCat}
A \textit{$\ast$-category} is a (unital) $\bbC$-linear category $\A$
with an involutive complex antilinear contravariant endofunctor $\ast : \A^\op\to \A$
acting as the identity on objects. In more detail, an involution 
consists of a family of complex antilinear maps $\ast : \A(x,y)\to \A(y,x)\,,~f\mapsto f^\ast$, 
for all objects $x,y\in\A$, such that
\begin{flalign}
(\lambda\,f + \lambda^\prime\,f^\prime)^\ast\,=\, 
\overline{\lambda}\,f^\ast + \overline{\lambda^\prime}\, f^{\prime\ast}~,\quad
(g\,f)^\ast\,=\, f^\ast\,g^\ast~,\quad
f^{\ast\ast} \,=\, f~,
\end{flalign}
for all $f,f^\prime\in \A(x,y)$, $g\in \A(y,z)$ and $\lambda,\lambda^\prime\in\bbC$.
We denote by $\astCat$ the category whose objects are all small $\ast$-categories
and whose morphisms are all $\ast$-functors, i.e.\ $\bbC$-linear functors
$F:\A\to \B$ such that $F(f^\ast) = F(f)^\ast$, for all morphisms $f$ in $\A$.
\end{defi}

Similarly to the case of $C^\ast$-algebras in Subsection \ref{subsec:Cast},
one can define a $C^\ast$-category as a $\ast$-category which has additional properties.
\begin{defi}\label{def:CastCat}
A \textit{$C^\ast$-category} is a $\ast$-category $\A\in\astCat$ which has the following properties:
\begin{itemize}
\item[(i)] There exists a family of norms $\norm{\cdot}_\A^{}: \A(x,y)\to [0,\infty)$, for all $x,y\in \A$,
which satisfy the $C^\ast$-property $\norm{f^\ast f}_\A^{} = \norm{f}_\A^2$, for all morphisms $f$ in $\A$,
and turn $\A$ into a Banach $\ast$-category.
That is, each $\A(x,y)$ is complete with respect to the norm
$\norm{\cdot}_\A^{}$ and $\norm{ g\,f }_\A^{} \leq \norm{g}_\A^{} \, \norm{f}_\A^{}$ 
whenever the composition is defined.

\item[(ii)] For every morphism $f\in \A(x,y)$, the endomorphism $f^\ast\,f\in\A(x,x)$
is positive, i.e.\ its spectrum is contained in $[0,\infty)$.
\end{itemize}
The category of $C^\ast$-categories is defined as the full subcategory
\begin{flalign}
\CastCat\,\subseteq\,\astCat
\end{flalign}
of the category of $\ast$-categories from Definition \ref{def:astCat}
whose objects are all $C^\ast$-categories.
\end{defi}

A particularly important construction for $C^\ast$-categories that 
is used in our work is the \textit{maximal tensor product}
\begin{flalign}\label{eqn:maximaltensor}
\boxtimes\,:=\,\otimes_{\max}\,:\, \CastCat\times\CastCat~\xrightarrow{\;\quad\;}~\CastCat~,
\end{flalign}
which we denote by the symbol $\boxtimes$ in order to avoid subscripts.
Let us recall its description from \cite[Proposition 3.12]{CastCat}:
Given two $C^\ast$-categories $\A,\B\in\CastCat$, one forms
first their algebraic tensor product $\A\otimes\B\in\astCat$, which is the $\ast$-category
whose objects are pairs of objects $(x,p)\in\A\times \B$ and whose morphisms
are defined by the tensor product of vector spaces
\begin{flalign}
(\A\otimes \B)\big((x,p),(y,q)\big)\,:=\, \A(x,y)\otimes \B(p,q)~,
\end{flalign}
for all $(x,p), (y,q)\in\A\times \B$. The involution on $\A\otimes \B$ is defined
by the complex antilinear extension of $(f\otimes g)^\ast := f^\ast\otimes g^\ast : (y,q)\to (x,p)$, 
for all $f\otimes g : (x,p)\to (y,q)$. The maximal tensor product
$\A\boxtimes \B\in\CastCat$ is then defined by completing the algebraic
tensor product with respect to the maximal norm, which is defined in analogy to
the case of $(C)^\ast$-algebras \eqref{eqn:maxnorm} by considering
$\ast$-functors to all $C^\ast$-categories. By construction,
there exists a canonical $\ast$-functor
\begin{subequations}\label{eqn:maxtensoralgtensor}
\begin{flalign}
\A\otimes\B ~\xrightarrow{\;\quad\;}~\A\boxtimes\B
\end{flalign}
which induces a natural bijection of Hom-sets
\begin{flalign}
\Hom_{\astCat}\big(\A\boxtimes\B,\C\big) ~\xrightarrow{\quad\cong\quad}~
\Hom_{\astCat}\big(\A\otimes\B,\C\big)~,
\end{flalign}
\end{subequations}
for all $C^\ast$-categories $\A,\B,\C\in\CastCat$. 
This means that specifying a $\ast$-functor $\A\boxtimes \B\to \C$ between $C^\ast$-categories
is equivalent to specifying a $\ast$-functor $\A\otimes \B\to \C$ out of the algebraic tensor product,
which is practically very convenient.
\sk

The following result follows from \cite{CastCat} and \cite[Corollary 11.4]{Bunke}.
\begin{theo}\label{theo:CastCatmodel}
\begin{itemize}
\item[(a)] The category $\CastCat$ from Definition \ref{def:CastCat}
is closed symmetric monoidal with respect to the maximal tensor product $\boxtimes:=\otimes_{\max}$ from
\eqref{eqn:maximaltensor} and the monoidal unit
is given by the $C^\ast$-category  $\B\bbC\in\CastCat$ with a single object and morphisms $\bbC$.

\item[(b)] It is further a combinatorial simplicial symmetric monoidal model
category  with respect to the following classes of morphisms: A
$\ast$-functor $F:\A\to \B$ in $\CastCat$ is
\begin{itemize}
\item a weak equivalence if it is a unitary equivalence,
i.e.\ it is quasi-invertible such that every component of the natural isomorphisms
$F^{-1}\,F\cong \id_{\A}$ and $F\,F^{-1}\cong \id_{\B}$ is unitary.
(As a consequence of \cite[Lemma 4.6]{CastCat}, this is equivalent to the 
underlying functor between ordinary categories
being fully faithful and essentially surjective.)

\item a cofibration if it is injective on objects.

\item a fibration if
every unitary morphism $v : F(x)\to p$ in $\B$ admits a lift, i.e.\
there exists a unitary morphism $u:x\to y$ in $\A$ such that $F(u) = v$.
\end{itemize}
\end{itemize}
\end{theo}

\begin{rem}\label{rem:CastCatmodel}
The symmetric monoidal structure from Theorem \ref{theo:CastCatmodel} (a)
allows us to define a $1$-categorical concept of $\CastCat$-valued prefactorization algebras, 
which will be sufficient for our constructions in Section \ref{sec:PFA}.
The symmetric monoidal \textit{model} structure from Theorem \ref{theo:CastCatmodel} (b)
determines a richer $\infty$-categorical context for such objects, which will be discussed
in more detail in Section \ref{sec:lattice}. In particular, this $\infty$-categorical context 
is needed to identify and characterize the algebraic structures on superselection sectors
for lattice $C^\ast$-AQFTs in Section \ref{sec:lattice}.
\sk

We would like to note that $C^\ast$-categories also assemble naturally into a symmetric monoidal
$2$-category, see e.g.\ \cite{Voigt} and \cite{Hataishi}. Using tools
from $2$-operad theory, see e.g.\ \cite{Corner,BPSWcategorified}, it would be possible 
to formulate our results directly in this $2$-categorical context, which however
would come with the following drawbacks: 1.)~Our explicit constructions and results
in Section \ref{sec:PFA} behave manifestly $1$-categorical, hence a $2$-categorical
framework would unnecessarily complicate their presentation. 2.)~Our more abstract
constructions and results in Section \ref{sec:lattice} apply to an
arbitrary presentably symmetric monoidal $\infty$-category as target for
operad algebras, hence a $2$-categorical
framework would unnecessarily restrict their range of applicability. In particular, this would
limit potential future applications to more sophisticated 
lattice models which include stacky and/or derived geometric structures, 
whose representation categories take the form of higher categories
or even dg-categories, see e.g.\ \cite{2dYM,CalaqueCarmona}.
\end{rem}

We will now introduce a concept of $\CastCat$-valued prefactorization algebras 
which are defined over any orthogonal category $\CC^\perp$.
For our applications in Section \ref{sec:PFA}, the relevant examples
are full orthogonal subcategories of the orthogonal category
$\Sub(S)^\perp$ of subsets of a (possibly infinite) set $S\in \Set$ from Definition \ref{def:SubS}.
The following definition of the prefactorization operad associated to an orthogonal
category is taken from \cite[Definition 2.5]{BPSWcategorified} and 
was motivated by the works of Costello and Gwilliam \cite{CG1,CG2}.
\begin{defi}\label{def:Poperad}
The \textit{prefactorization operad} $\P_{\CC^\perp}$ associated to an
orthogonal category $\CC^\perp$ is the $\Set$-valued colored symmetric operad defined
by the following data:
\begin{enumerate}[(1)]
\item The objects of $\P_{\CC^\perp}$ are the objects of the category $\CC$.

\item The sets of operations are
\begin{flalign}
\P_{\CC^\perp}\big(\substack{V \\ \und{U}}\big)\,:=\, \bigg\{ \und{f} := (f_1,\dots,f_n)\in \prod_{i=1}^n \CC(U_i,V)\,:\, f_i\perp f_j \text{~for all~}i\neq j\,\bigg\}~,
\end{flalign}
for each object $V\in\CC$ and each tuple of objects 
$\und{U} :=(U_1,\dots,U_n)\in \CC^n$. For the empty tuple $\und{U}=()$,
we set $\P_{\CC^\perp}\big(\substack{V \\ ()}\big):=\{\pt_V^{}\}$ to be a singleton.

\item The composition maps 
\begin{subequations}
\begin{flalign}
\gamma \,:\, \P_{\CC^\perp}\big(\substack{V \\ \und{U}}\big)\times\prod_{i=1}^n 
\P_{\CC^\perp}\big(\substack{U_i \\ \und{W_i}}\big)~\xrightarrow{\;\quad\;}~ \P_{\CC^\perp}\big(\substack{V \\ \und{\und{W}}}\big)~,
\end{flalign}
where $\und{\und{W}} := (\und{W_1},\dots,\und{W_n})$ denotes the concatenation of tuples, 
are given by composition in the category $\CC$, i.e.\
\begin{flalign}
 \gamma \big(\und{f},(\und{g_1},\dots,\und{g_n}) \big)\,:=\, \und{f}\,\und{\und{g}} \,:=\, \big(f_1\,g_{11},\dots, f_1\, g_{1 k_1},\dots,f_n\, g_{n1},\dots,f_n \, g_{n k_n}\big)~.
\end{flalign}
\end{subequations}

\item The identity operations are $\id_V \in \P_{\CC^\perp}\big(\substack{V \\ V}\big)$.

\item The permutation actions $\P_{\CC^\perp}(\sigma) : \P_{\CC^\perp}\big(\substack{V \\ \und{U}}\big)
\to \P_{\CC^\perp}\big(\substack{V \\ \und{U}\sigma}\big)$, for $\sigma\in\Sigma_n$, are given by
\begin{flalign}
\P_{\CC^\perp}(\sigma)(\und{f}) \,:=\, \und{f}\sigma := (f_{\sigma(1)},\dots, f_{\sigma(n)})~.
\end{flalign}
\end{enumerate}
\end{defi}

\begin{defi}\label{def:PFA}
The category of \textit{$C^\ast$-categorical prefactorization algebras}
over $\CC^\perp$ is defined as the category
\begin{flalign}
C^\ast\mathbf{PFA}_{\CC^\perp}\,:=\, \Alg_{\P_{\CC^\perp}}\big(\CastCat\big)
\end{flalign}
of algebras over the operad $\P_{\CC^\perp}$ from Definition \ref{def:Poperad}
with values in the closed symmetric monoidal category $(\CastCat,\boxtimes,\B\bbC)$ 
of $C^\ast$-categories from Theorem \ref{theo:CastCatmodel} (a).
\end{defi}

\begin{rem}\label{rem:PFA}
For the convenience of readers who are not familiar with operadic concepts, let 
us spell out explicitly the data and properties of a $C^\ast$-categorical 
prefactorization algebra $\FFF : \P_{\CC^\perp}\to \CastCat$:
\begin{itemize}
\item[(1)] To any object $U\in \CC$ is assigned a $C^\ast$-category
$\FFF(U)\in\CastCat$.

\item[(2)] To any operation $\und{f}: \und{U} = (U_1,\dots,U_n)\to V$ in $\P_{\CC^\perp}$
is assigned a $\ast$-functor
\begin{flalign}\label{eqn:PFAstructuremaps}
\FFF(\und{f})\,:\, \bigboxtimes_{i=1}^n \FFF(U_i) ~\xrightarrow{\;\quad\;}~\FFF(V)
\end{flalign}
from the maximal tensor product $\boxtimes$ of $C^\ast$-categories. In the special case
where $\und{U} = ()$ is the empty tuple (i.e.\ $n=0$), this amounts to a $\ast$-functor
\begin{flalign}\label{eqn:PFAunit}
\FFF(\pt_V^{})\,:\,\B\bbC~\xrightarrow{\;\quad\;}~\FFF(V)
\end{flalign}
from the monoidal unit $\B\bbC\in \CastCat$.
\end{itemize}
These data have to satisfy the following axioms:
\begin{itemize}
\item[(i)] For all operations $\und{f}: \und{U} = (U_1,\dots,U_n)\to V$
and $\und{g_i}: \und{W_i} = (W_{i1},\dots,W_{ik_i})\to U_i$ in $\P_{\CC^\perp}$, 
for $i=1,\dots,n$, the diagram
\begin{equation}\label{eqn:PFAdiagramcomp}
\begin{tikzcd}
\ar[rd,"\FFF(\und{f}\,\und{\und{g}})"']\bigboxtimes\limits_{i=1}^n \bigboxtimes\limits_{j=1}^{k_i} \FFF(W_{ij}) \ar[r,"\boxtimes_{i}\FFF(\und{g_i})"] &[10mm] 
\bigboxtimes\limits_{i=1}^n \FFF(U_i)\ar[d,"\FFF(\und{f})"]\\
&[10mm] \FFF(V)
\end{tikzcd}
\end{equation}
in $\CastCat$ commutes.

\item[(ii)] For all identity operations $\id_V: V\to V$ in $\P_{\CC^\perp}$,  the $\ast$-functor
\begin{flalign}
\FFF(\id_V) \,=\,\id_{\FFF(V)}\,:\, \FFF(V)~\xrightarrow{\;\quad\;}~\FFF(V)
\end{flalign}
coincides with the identity $\ast$-functor.

\item[(iii)] For all operations $\und{f} : \und{U} = (U_1,\dots,U_n)\to V$ in $\P_{\CC^\perp}$ 
and permutations $\sigma\in\Sigma_n$, the diagram
\begin{equation}\label{eqn:PFAdiagramperm}
\begin{tikzcd}
\ar[dr,"\FFF(\und{f})"'] \bigboxtimes\limits_{i=1}^n\FFF(U_i)\ar[rr,"\cong"] && 
\bigboxtimes\limits_{i=1}^n\FFF(U_{\sigma(i)}) \ar[dl,"\FFF(\und{f}\sigma)"] \\
& \FFF(V)
\end{tikzcd}
\end{equation}
in $\CastCat$ commutes, where the unlabeled isomorphism is obtained from the symmetric braiding
on $\CastCat$.
\end{itemize}
As a consequence of the universal property \eqref{eqn:maxtensoralgtensor} of 
the maximal tensor product, we note that the structure maps \eqref{eqn:PFAstructuremaps}
can be specified equivalently by $\ast$-functors 
\begin{flalign}
\FFF(\und{f})\, :\, \bigotimes_{i=1}^n\FFF(U_i)~\xrightarrow{\;\quad\;}~ \FFF(V)
\end{flalign}
out of the algebraic tensor product and \eqref{eqn:PFAunit} is equivalent to the choice of an object 
\begin{flalign}
\FFF(\pt_V^{})\,\in\, \FFF(V)~.
\end{flalign} 
Furthermore, commutativity of the diagrams \eqref{eqn:PFAdiagramcomp} and 
\eqref{eqn:PFAdiagramperm} can be verified
by restricting to algebraic tensor products. These observations are useful in practice.
\end{rem}


\section{\label{sec:PFA}Prefactorization algebras of localized superselection sectors}
Consider any set $S\in \Set$ and any family $\AAA = \{\AAA(s)\}_{s\in S}\in\CastAlg^S$
of $C^\ast$-algebras. Then the construction in Proposition \ref{prop:LKECast}
defines a $C^\ast$-AQFT
\begin{flalign}\label{eqn:Aoverline}
\overline{\AAA}\,:=\,\mathrm{compl}\big(j_!(\AAA)\big)\,\in\,\CastAQFT(S)
\end{flalign}
on the orthogonal category $\Sub(S)^\perp$
of all subsets of $S$ from Definition \ref{def:SubS}.
This means that we can assign functorially to each subset $U\subseteq S$
a $C^\ast$-algebra $\overline{\AAA}(U)\in\CastAlg$, which is constructed explicitly
in the proof of Proposition \ref{prop:LKECast} in terms of tensor products, filtered colimits
and $C^\ast$-completions, and to each subset inclusion $U\subseteq V$ a $\ast$-homomorphism
$\overline{\AAA}(U)\to \overline{\AAA}(V)$. This functor satisfies by construction
the $\perp$-commutativity property \eqref{eqn:perpcommutativity}.
It assigns to the empty subset $\varnothing\subseteq S$ the 
initial object $\overline{\AAA}(\varnothing) = \bbC\in C^\ast\Alg$ 
and to each singleton $\{s\}\subseteq S$ the input datum
$\overline{\AAA}(\{s\}) = \AAA(s)\in \CastAlg$ of the construction. 
In applications to quantum spin systems
\cite{Naaijkens1,Naaijkens2,NaaijkensChapter,FiedlerNaaijkens,Ogata},
one usually considers a lattice $S=\bbZ^n$ and
a family of finite-dimensional matrix algebras 
$\AAA = \{\AAA(s) \cong \mathrm{Mat}_d(\bbC)\}_{s\in S}\in\CastAlg^S$.
\sk

The aim of this section is to construct $C^\ast$-categories
of suitably localized superselection sectors for such 
$C^\ast$-AQFTs over $S$ and to explore the algebraic structures 
which can be defined, under suitable additional assumptions,
on these $C^\ast$-categories. We propose the following simple
definition of $\ast$-representations of the extended $\overline{\AAA}\in \CastAQFT(S)$ 
in \eqref{eqn:Aoverline} and argue in Remark \ref{rem:SSStraditional} 
below that this is equivalent to the more standard definition in the literature.
\begin{defi}\label{def:SSS}
Let $\AAA = \{\AAA(s)\}_{s\in S}\in\CastAlg^S$ be any family of $C^\ast$-algebras
which is indexed by a set $S\in \Set$. The $C^\ast$-category ${}^\ast\mathbf{Rep}_{\AAA}$
of \textit{$\ast$-representations} of the corresponding
extended $\overline{\AAA}\in \CastAQFT(S)$ 
in \eqref{eqn:Aoverline} is defined as follows:
\begin{itemize}
\item[$\mathsf{Obj}$:] An object in ${}^\ast\mathbf{Rep}_{\AAA}$ is a pair $(H,\pi)$ consisting
of a complex Hilbert space $H$ and a family $\pi = \{\pi_s : \AAA(s)\to B(H)\}_{s\in S}$
of mutually commuting $\ast$-homomorphisms to the $C^\ast$-algebra $B(H)\in \CastAlg$
of bounded operators on $H$, i.e.\ the diagram
\begin{equation}
\begin{tikzcd}
\ar[d,"\pi_{s}\otimes\pi_{s^\prime}"'] \AAA(s)\otimes \AAA(s^\prime) \ar[r,"\pi_{s}\otimes\pi_{s^{\prime}}"] &[10mm]  B(H)\otimes B(H) \ar[d,"\circ^{\op}"] \\
B(H)\otimes B(H) \ar[r,"\circ"'] &[10mm] B(H)  
\end{tikzcd}
\end{equation}
of linear maps between vector spaces commutes, for all $s,s^\prime\in S$ with $s\neq s^\prime$,
where $\circ^{(\op)}$ denotes the (opposite) composition of bounded operators.

\item[$\mathsf{Mor}$:] Given two objects $(H,\pi)$ and $(H^\prime,\pi^\prime)$,
the Banach space of morphisms is the closed subspace
\begin{flalign}
{}^\ast\mathbf{Rep}_{\AAA}\big((H,\pi),(H^\prime,\pi^\prime)\big)\,\subseteq\, B(H,H^\prime)
\end{flalign}
of the Banach space of bounded operators from $H$ to $H^\prime$ which consists
of all bounded operators $L : H\to H^\prime$ such that the diagram
\begin{equation}\label{eqn:SSSmorphisms}
\begin{tikzcd}
\ar[d,"\pi_s"'] \AAA(s)\ar[r,"\pi_s^\prime"] &[4mm] B(H^\prime)\ar[d,"(-)\circ L"]\\
B(H)\ar[r,"L\circ(-)"'] &[4mm] B(H,H^\prime)
\end{tikzcd}
\end{equation}
of linear maps between vector spaces commutes, for all $s\in S$.
\end{itemize}
The identity morphisms are given by $\id_H : (H,\pi)\to (H,\pi)$
and the composition of two morphisms $L : (H,\pi)\to (H^\prime,\pi^\prime)$
and $L^\prime : (H^\prime,\pi^\prime)\to (H^{\prime\prime},\pi^{\prime\prime})$
is the composition $L^\prime\circ L: (H,\pi) \to (H^{\prime\prime},\pi^{\prime\prime})$
of bounded operators. The $\ast$-involution is defined by forming adjoints $L^\ast : (H^\prime,\pi^\prime)\to
(H,\pi)$ of bounded operators $L : (H,\pi)\to (H^\prime,\pi^\prime)$.
\end{defi}

\begin{rem}\label{rem:SSStraditional}
The usual definition of $\ast$-representations of the 
extended $\overline{\AAA}\in \CastAQFT(S)$ in \eqref{eqn:Aoverline}
consists of $\ast$-homomorphisms $\overline{\AAA}(S)\to B(H)$ 
from the global $C^\ast$-algebra $\overline{\AAA}(S)\in \CastAlg$
which is assigned to the entire set $S\in\Sub(S)$.
We would like to note that this definition is equivalent to our Definition \ref{def:SSS}:
As a consequence of the adjunction in Proposition \ref{prop:Castadjunctions} (b),
a $\ast$-homomorphism $\overline{\AAA}(S) = \mathrm{compl}\big(j_!(\AAA)(S)\big)\to B(H)$
is equivalent to the datum of a $\ast$-homomorphism $\pi : j_!(\AAA)(S)\to B(H)$.
From our construction
of $j_!(\AAA)(S)$ in the proof of  Proposition \ref{prop:LKECast}, in terms of a colimit 
over the slice category $\Fin(S)/S$, it follows that this datum is equivalent to a family 
\begin{flalign}
\bigg\{\pi_U \,:\, \AAA(U)\,:=\,\bigotimes_{s\in U}\AAA(s)~\xrightarrow{\;\quad\;}~B(H)\bigg\}_{U\in \Fin(S)}
\end{flalign}
of $\ast$-homomorphisms, indexed by all finite subsets $U\subseteq S$, such that the diagram
\begin{equation}\label{eqn:trianglerep}
\begin{tikzcd}
& B(H) &\\
\AAA(U) \ar[ru,"\pi_U"]\ar[rr] && \AAA(V)\ar[lu,"\pi_V"']
\end{tikzcd}
\end{equation}
in $\astAlg$ commutes, for all inclusions $U\subseteq V$ of finite subsets.
Since $\AAA(U) =\bigotimes_{s\in U}\AAA(s)$ is a tensor product $\ast$-algebra,
for all $U\in \Fin(S)$, such datum is equivalent to an $S$-indexed family $\{\pi_s : \AAA(s)\to B(H)\}_{s\in S}$
of mutually commuting $\ast$-homomorphisms as in Definition \ref{def:SSS}.
\sk

By a similar argument as above, it follows that every object $(H,\pi)\in {}^\ast\mathbf{Rep}_{\AAA}$
in the category from Definition \ref{def:SSS} comes with a canonical $\ast$-homomorphism
\begin{flalign}\label{eqn:piUextend}
\pi_U\,:\,\overline{\AAA}(U)~\xrightarrow{\;\quad\;}~B(H)
\end{flalign}
from the $C^\ast$-algebra $\overline{\AAA}(U)\in \CastAlg$
which is assigned by the extended $\overline{\AAA}\in \CastAQFT(S)$ 
to any (possibly infinite) subset $U\in\Sub(S)$.
This family of $\ast$-homomorphisms satisfies commutative triangles as in \eqref{eqn:trianglerep},
for all subset inclusion $U\subseteq V$, and it is related to the input datum via
$\pi_{\{s\}} = \pi_s$, for all elements $s\in S$.
\end{rem}

The $C^\ast$-category of all $\ast$-representations ${}^\ast\mathbf{Rep}_{\AAA}$
is typically too large and unstructured,
which has led Doplicher, Haag and Roberts \cite{DHR} and later
also Buchholz and Fredenhagen \cite{BuchholzFredenhagen} to consider only specific
classes of $\ast$-representations which are called \textit{superselection sectors}. 
These $\ast$-representations are defined relative to a fixed faithful $\ast$-representation 
$(H,\pi_0)\in {}^\ast\mathbf{Rep}_{\AAA}$  (the ``vacuum representation'') and
are required to be localizable in a suitable class of subsets $U\subseteq S$,
i.e.\ they are unitarily equivalent to the vacuum representation in all 
complements $U^\cc := S\setminus U$.
In our context, it is convenient to describe such localization regions 
by means of a full subcategory $\CC(S)\subseteq \Sub(S)$ 
of the orthogonal category $\Sub(S)^\perp$ from Definition \ref{def:SubS},
which inherits the structure of an orthogonal full subcategory 
$\CC(S)^\perp\subseteq \Sub(S)^\perp$ by 
restricting the orthogonality relation.
\begin{defi}\label{def:SSSloc}
Let $\AAA = \{\AAA(s)\}_{s\in S}\in\CastAlg^S$ be any family of $C^\ast$-algebras
which is indexed by a set $S\in \Set$. Fix any object $(H,\pi_0)\in {}^\ast\mathbf{Rep}_\AAA$ 
such that the $\ast$-homomorphism $\pi_{0U} : \overline{\AAA}(U)\to B(H)$ is injective, for all $U\subseteq S$,
and any full subcategory $\CC(S)\subseteq \Sub(S)$.
The $C^\ast$-category of \textit{$(\CC(S),\pi_0)$-localizable superselection sectors} 
of the corresponding extended $\overline{\AAA}\in \CastAQFT(S)$ 
in \eqref{eqn:Aoverline} is defined as the full $C^\ast$-subcategory
\begin{subequations}\label{eqn:localizablerep}
\begin{flalign}
\SSS_{(\AAA,\CC(S),\pi_0)} \,\subseteq\,{}^\ast\mathbf{Rep}_\AAA
\end{flalign}
consisting of all objects $(K,\pi)\in {}^\ast\mathbf{Rep}_\AAA$ which satisfy
the following localizability conditions: For every $U\in \CC(S)$, there exists
a unitary operator $u : K\to H$ such that the diagram
\begin{equation}
\begin{tikzcd}
\ar[d,"\pi_{s}"'] \AAA(s) \ar[r,"\pi_{0s}"] &[4mm]  B(H) \ar[d,"(-)\circ u"] \\
B(K) \ar[r,"u\circ(-)"'] &[4mm] B(K,H)  
\end{tikzcd}
\end{equation}
\end{subequations}
of linear maps between vector spaces commutes, for all $s\in U^\cc := S\setminus U$ in the complement of $U\subseteq S$.
\end{defi}

\begin{rem}\label{rem:smalllocalizableisenough}
To show that a $\ast$-representation $(K,\pi)\in {}^\ast\mathbf{Rep}_\AAA$
lies in the full $C^\ast$-subcategory $\SSS_{(\AAA,\CC(S),\pi_0)} \subseteq {}^\ast\mathbf{Rep}_\AAA$,
it suffices to verify the localizability conditions \eqref{eqn:localizablerep}
only for sufficiently small subsets $U\in \CC(S)$. More precisely, given any morphism
$U\subseteq V$ in $\CC(S)$, localizability of $(K,\pi)\in {}^\ast\mathbf{Rep}_\AAA$ 
in the smaller subset $U\in \CC(S)$ implies localizability in the larger subset $V\in\CC(S)$. 
Indeed, if there exists a unitary operator $u: K\to H$ such that $u \circ \pi_s(-) = \pi_{0s}(-)\circ u$,
for all $s\in U^\cc$, then this property holds true in particular for all $s\in V^\cc \subseteq U^\cc$.
This simple observation will become useful below.
\end{rem}

In our constructions below, it is convenient to work with the following family
of unitarily equivalent models for the $C^\ast$-category $\SSS_{(\AAA,\CC(S),\pi_0)}$ 
of $(\CC(S),\pi_0)$-localizable superselection sectors from Definition \ref{def:SSSloc}.
\begin{lem}\label{lem:SSSU}
For each $U\in\CC(S)$, we denote by
\begin{subequations}\label{eqn:localizedrep}
\begin{flalign}
\SSS_{(\AAA,\pi_0)}(U)\,\subseteq\,\SSS_{(\AAA,\CC(S),\pi_0)}
\end{flalign}
the full $C^\ast$-subcategory consisting of all objects $(H,\pi)\in \SSS_{(\AAA,\CC(S),\pi_0)}$ 
which share the same Hilbert space with the reference $\ast$-representation $(H,\pi_0)\in{}^\ast\mathbf{Rep}_\AAA$ 
and satisfy the strict localization condition
\begin{flalign}
\pi_s\,=\,\pi_{0s}\,:\AAA(s)~\xrightarrow{\;\quad\;}~B(H)~,
\end{flalign}
\end{subequations}
for all $s\in U^\cc$. Then the full $C^\ast$-subcategory inclusion
\eqref{eqn:localizedrep} is a unitary equivalence of $C^\ast$-categories, for every $U\in\CC(S)$.
\end{lem}
\begin{proof} By Definition \ref{def:SSSloc}, 
given any $(K,\pi)\in \SSS_{(\AAA,\CC(S),\pi_0)}$, there exists 
a unitary operator $u: K\to H$ such that $u\circ \pi_s(-) = \pi_{0s}(-)\circ u$,
for all $s\in U^\cc$. Then the family of $\ast$-homomorphisms
\begin{flalign}\label{eqn:piuformula}
\pi^u\,:=\,\Big\{u \circ \pi_s(-)\circ u^\ast \,:\, \AAA(s)~\xrightarrow{\;\quad\;}~ B(H)\Big\}_{s\in S}
\end{flalign}
defines an object $(H,\pi^u)\in \SSS_{(\AAA,\CC(S),\pi_0)}$ which satisfies by construction 
the strict localization condition
$\pi^u_s = \pi_{0s}$, for all $s\in U^\cc$. In other words, $(H,\pi^{u})$ is contained in the
full $C^\ast$-subcategory $\SSS_{(\AAA,\pi_0)}(U)\subseteq\SSS_{(\AAA,\CC(S),\pi_0)}$.
The unitary $u: K\to U$ defines a unitary isomorphism $u : (K,\pi)\to (H,\pi^u)$ in $\SSS_{(\AAA,\CC(S),\pi_0)}$,
which implies that the full $C^\ast$-subcategory inclusion $\SSS_{(\AAA,\pi_0)}(U)\subseteq 
\SSS_{(\AAA,\CC(S),\pi_0)}$ is a unitary equivalence of $C^\ast$-categories.
\end{proof}

\begin{rem}
We would like to note that similar $C^\ast$-categories of superselection sectors
which are also strictly localized in some fixed region appeared before in the literature.
See for example \cite[Setting 5.1]{Ogata} for the context of $2$-dimensional lattice quantum 
systems and \cite[Definition 2.25]{SSS} for a general von Neumann algebraic 
setting for AQFTs over posets.
\end{rem}

\begin{rem}
Since all objects $(H,\pi)\in \SSS_{(\AAA,\pi_0)}(U)$ in the $C^\ast$-category
from \eqref{eqn:localizedrep} share the same Hilbert space $H$,
fixed by the choice of reference $\ast$-representation $(H,\pi_0)\in {}^\ast\mathbf{Rep}_\AAA$,
we will drop the Hilbert space from our notations and simply write $\pi\in \SSS_{(\AAA,\pi_0)}(U)$
for an object and $L : \pi\to \pi^\prime$ for a morphism in $\SSS_{(\AAA,\pi_0)}(U)$.
\end{rem}

\begin{rem}\label{rem:transporter}
From Definition \ref{def:SSSloc}, it follows that, given any
$\pi\in \SSS_{(\AAA,\pi_0)}(U) \subseteq \SSS_{(\AAA,\CC(S),\pi_0)}$ which is
localized strictly in $U\in\CC(S)$ and any other localization region $V\in\CC(S)$, there exists
a unitary operator $u : H\to H$ such that $u\circ \pi_s(-)  = \pi_{0s}(-)\circ u$, for all $s\in V^\cc$.
This means that $\pi^u\in \SSS_{(\AAA,\pi_0)}(V)\subseteq \SSS_{(\AAA,\CC(S),\pi_0)}$ defined
in \eqref{eqn:piuformula} is localized strictly in $V\in\CC(S)$. We further have a unitary isomorphism
\begin{flalign}\label{eqn:transporter}
u\,:\, \pi~\xrightarrow{\;\quad\;}~\pi^u
\end{flalign}
in $\SSS_{(\AAA,\CC(S),\pi_0)}$, which in the literature is often called a charge transporter.
\end{rem}

A direct consequence of the definition in \eqref{eqn:localizedrep} is that
we have inclusions of full $C^\ast$-subcategories
\begin{flalign}\label{eqn:SSSUVinclion}
\SSS_{(\AAA,\pi_0)}(U)\,\subseteq\,\SSS_{(\AAA,\pi_0)}(V)\,\subseteq\,\SSS_{(\AAA,\CC(S),\pi_0)}~,
\end{flalign}
for all morphisms $U\subseteq V$ in $\CC(S)$. Note that these inclusions are unitary equivalences
by Lemma~\ref{lem:SSSU} and the $2$-out-of-$3$ property of unitary equivalences. 
This implies the following result.
\begin{propo}\label{prop:SSSfunctor}
The $C^\ast$-categories from \eqref{eqn:localizedrep} assemble into a functor
\begin{flalign}
\SSS_{(\AAA,\pi_0)} \,:\,\CC(S)~\xrightarrow{\;\quad\;}~\CastCat
\end{flalign}
which assigns to each morphism $U\subseteq V$ in $\CC(S)$ the 
full $C^\ast$-subcategory inclusion in \eqref{eqn:SSSUVinclion}.
This functor is locally constant in the sense that it 
assigns to every morphism $U\subseteq V$ in $\CC(S)$ a unitary equivalence of
$C^\ast$-categories.
\end{propo}

\begin{rem}
Another direct consequence of the definition in \eqref{eqn:localizedrep}
is that there exists, for each $U\in\CC(S)$, 
a distinguished object
\begin{flalign}\label{eqn:pointing}
\pi_0\,\in\, \SSS_{(\AAA,\pi_0)}(U)
\end{flalign}
given by the reference $\ast$-representation.
Note that these objects are compatible with the
functor structure from Proposition \ref{prop:SSSfunctor}, hence one can also
think of this functor as a functor to the category of pointed 
(i.e.\ $\mathbb{E}_0$-monoidal) $C^\ast$-categories.
\end{rem}

\subsection{\label{subsec:PFAstructure}Prefactorization algebra structure}
Note that the functor $\SSS_{(\AAA,\pi_0)} :\CC(S) \to \CastCat$ from 
Proposition \ref{prop:SSSfunctor} and the distinguished objects \eqref{eqn:pointing}
describe the arity $0$ and $1$ data of a $C^\ast$-categorical
prefactorization algebra over $\CC(S)^\perp$, see in particular Remark \ref{rem:PFA}. 
In order to obtain the entire data of a $C^\ast$-categorical prefactorization algebra,
one has to define also actions of the $(n\geq 2)$-ary operations $\und{U} = (U_1,\dots,U_n)\to V$ in 
the operad $\P_{\CC(S)^\perp}$, i.e.\ $U_i\subseteq V$ for all $i$ and
$U_i\cap U_j=\varnothing$ for all $i\neq j$. 
The goal of this section is to show that such an extension of the functor
to a prefactorization algebra exists, provided that certain
geometric and algebraic assumptions are satisfied. 
Let us start with introducing suitable geometric assumptions on 
the full subcategory $\CC(S)\subseteq \Sub(S)$
encoding the localization regions of superselection sectors 
in Definition \ref{def:SSSloc} which ensure the existence
of our constructions in this subsection and also in Subsection \ref{subsec:monoidal} below.
We show in Section \ref{sec:lattice} that these assumptions are satisfied
in the case of cone-shaped subsets in the lattice $S=\bbZ^n$.
\begin{assu}\label{assu:0}
We assume that the full subcategory $\CC(S)\subseteq \Sub(S)$
satisfies the following properties: 
\begin{itemize}
\item[(1)] Given any $U\in\CC(S)$, there exists an object $V\in \CC(S)$
such that $V\subseteq U^\cc$.

\item[(2)] Given any $U,V\in\CC(S)$, there exist objects  $V^\prime, W\in\CC(S)$
such that $V^\prime\subseteq V$ and $U\cup V^\prime\subseteq W$.

\item[(3)] Given any $U\in\CC(S)$ and any finite 
tuple $(U_1,\dots,U_n)\in\CC(S)^n$ of mutually disjoint subsets, i.e.\ $U_i\cap U_j=\varnothing$
for all $i\neq j$, there exists an object $V\in\CC(S)$ 
such that $V\subseteq U$ and $V\cap U_i\neq \varnothing$
for at most one index $i\in\{1,\dots,n\}$.
\end{itemize}
\end{assu}

\begin{rem}\label{rem:assu0}
We would like to emphasize that our geometric Assumption \ref{assu:0}
is designed such that the strictly localized superselection sectors admit 
the structure of a prefactorization algebra over $\CC(S)^\perp$ taking values in monoidal
$C^\ast$-categories, see in particular Theorem \ref{theo:PFAinMonCat} below. 
Since such prefactorization algebras are more general 
algebraic objects than braided monoidal categories, it is not surprising
that our geometric assumptions are weaker than the geometric conditions 
(GA0--GA3) in \cite[Section 1.1]{SSS}. The latter are designed to achieve 
a more specialized result, namely that strictly localized superselection sectors form a 
\textit{braided} monoidal $C^\ast$-category.
We return to this point in Section~\ref{sec:lattice}.
\end{rem}

The following algebraic assumptions control the local behavior 
of strictly localized superselection sectors in the sense of \eqref{eqn:localizedrep}.
We show in Theorem \ref{theo:HaagDualityPFA} below that these assumptions hold
true in particular in the case where the extended $\overline{\AAA}\in \CastAQFT(S)$ 
satisfies Haag duality in the reference $\ast$-representation $(H,\pi_0)$,
for all $U\in \CC(S)$. The reason why we prefer presenting Assumption \ref{assu:1}
in the form below, instead of assuming directly the more standard Haag duality property,
is that it is designed to provide a streamlined proof for the key Lemma \ref{lem:PFAstructure}.
\begin{assu}\label{assu:1}
We assume that the following properties hold, for all
pairs $U_1,U_2\in \CC(S)$ of disjoint subsets, i.e.\
$U_1\cap U_2=\varnothing$:
\begin{itemize}
\item[(1)] For any pair of objects $\pi_1\in \SSS_{(\AAA,\pi_0)}(U_1)$ and $\pi_2\in \SSS_{(\AAA,\pi_0)}(U_2)$,
the diagram
\begin{equation}\label{eqn:assu11}
\begin{tikzcd}
\ar[d,"\pi_{1s_1}\otimes\pi_{2s_2}"'] \AAA(s_1)\otimes \AAA(s_2) \ar[r,"\pi_{1s_1}\otimes\pi_{2s_2}"] &[10mm]  B(H)\otimes B(H) \ar[d,"\circ^{\op}"] \\
B(H)\otimes B(H) \ar[r,"\circ"'] &[10mm] B(H)  
\end{tikzcd}
\end{equation}
of linear maps between vector spaces commutes, for all $s_1\in U_1$ and $s_2\in U_2$.

\item[(2)] 
For any pair of morphisms $L_1 : \pi_1\to\pi_1^\prime$ in $\SSS_{(\AAA,\pi_0)}(U_1)$
and $L_2 : \pi_2\to\pi_2^\prime$ in  $ \SSS_{(\AAA,\pi_0)}(U_2)$, we have that 
\begin{subequations}\label{eqn:assu12}
\begin{flalign}\label{eqn:assu12com}
L_1\circ L_2 \,= \,L_2\circ L_1
\end{flalign}
and the diagrams
\begin{equation}\label{eqn:assu12diag}
\begin{tikzcd}
\ar[d,"\pi_{1s_1}"'] \AAA(s_1) \ar[r,"\pi_{1s_1}"] &[4mm]  B(H) \ar[d,"(-)\circ L_2"] \\
B(H) \ar[r,"L_2\circ (-)"'] &[4mm] B(H)  
\end{tikzcd} 
\qquad \qquad
\begin{tikzcd}
\ar[d,"\pi_{2s_2}"'] \AAA(s_2) \ar[r,"\pi_{2s_2}"] &[4mm]  B(H) \ar[d,"(-)\circ L_1"] \\
B(H) \ar[r,"L_1\circ (-)"'] &[4mm] B(H) 
\end{tikzcd}
\end{equation}
\end{subequations}
of linear maps between vector spaces commute, for all $s_1\in U_1$ and $s_2\in U_2$.

\item[(3)] For any pair of objects $\pi_1\in \SSS_{(\AAA,\pi_0)}(U_1)$ and $\pi_2\in \SSS_{(\AAA,\pi_0)}(U_2)$
and any $V\in\CC(S)$ such that $V\cap U_1=\varnothing$ are disjoint (but $V$ and $U_2$ not necessarily disjoint), 
the unitary isomorphism $u_2 : \pi_2\to \pi_2^{u_2}$ from \eqref{eqn:transporter} to some
$\pi_2^{u_2}\in \SSS_{(\AAA,\pi_0)}(V)$ strictly localized in $V$ 
can be chosen such that the diagram
\begin{equation}\label{eqn:assu13}
\begin{tikzcd}
\ar[d,"\pi_{1s_1}"'] \AAA(s_1) \ar[r,"\pi_{1s_1}"] &[4mm]  B(H) \ar[d,"(-)\circ u_2"] \\
B(H) \ar[r,"u_2\circ (-)"'] &[4mm] B(H)  
\end{tikzcd}
\end{equation}
of linear maps between vector spaces commutes, for all $s_1\in U_1\cap V^\cc = U_1$.
\end{itemize}
\end{assu}

\begin{lem}\label{lem:PFAstructure}
Suppose that Assumptions \ref{assu:0} (3) and \ref{assu:1} are satisfied. For any 
$n$-ary operation $(U_1,\dots,U_n)\to V$
in the prefactorization operad $\P_{\CC(S)^\perp}$ from Definition \ref{def:Poperad},
i.e.\ $U_i\subseteq V$ for all $i$ and $U_i\cap U_j=\varnothing$ for all $i\neq j$, 
consider the assignment from the algebraic tensor product
\begin{subequations}\label{eqn:PFAstructure}
\begin{flalign}
\bullet\,:\, \bigotimes_{i=1}^n\SSS_{(\AAA,\pi_0)}(U_i)~&\xrightarrow{\;\quad\;}~\SSS_{(\AAA,\pi_0)}(V)~,\\
\nn (\pi_1,\dots,\pi_n)~&\xmapsto{\;\quad\;}~\pi_1\bullet\cdots\bullet\pi_n~,\\[4pt]
\nn L_1\otimes\cdots\otimes L_n ~&\xmapsto{\;\quad\;}~L_1\bullet \cdots\bullet L_n
\end{flalign}
given by
\begin{flalign}\label{eqn:PFAstructureformula}
(\pi_1\bullet\cdots\bullet\pi_n)_s\,:=\,\begin{cases}
\pi_{is} ~&\text{if } s\in U_i\text{ for }i=1,\dots,n\,, \\
\pi_{0s} ~&\text{if } s\in U_1^\cc\cap \cdots\cap U_n^\cc \,,
\end{cases}
\end{flalign}
for all $s\in S$, and by the composition of bounded operators
\begin{flalign}\label{eqn:PFAstructureformulamorphism}
L_1\bullet \cdots\bullet L_n\,:=\, L_1\circ \cdots\circ L_n\, \in\, B(H)~.
\end{flalign} 
\end{subequations}
Then \eqref{eqn:PFAstructure} is a well-defined $\ast$-functor.
\end{lem}
\begin{proof}
Note that mutual disjointness $U_i\cap U_j=\varnothing$, for all $i\neq j$,
implies that the case distinction in \eqref{eqn:PFAstructureformula} is valid.
Assumption \ref{assu:1} (1) implies that the family of $\ast$-homomorphisms
\eqref{eqn:PFAstructureformula} is mutually commuting and it is obvious from the 
definition that $(\pi_1\bullet\cdots\bullet\pi_n)_s = \pi_{0s}$, for all $s\in V^\cc$, because
$U_i\subseteq V$ for all $i$. It remains to show that \eqref{eqn:PFAstructureformula}
satisfies the localizability condition from Definition \ref{def:SSSloc}, for all $U\in\CC(S)$.
Using Remark \ref{rem:smalllocalizableisenough} and Assumption \ref{assu:0} (3),
it suffices to consider the case where $U\in\CC(S)$ intersects at most one of the 
$U_i$. Without loss of generality, we can take this to be $U_n$ (otherwise, we relabel
the subsets), hence we may assume that $U\cap U_i =\varnothing$ for all $i=1,\dots,n-1$.
Since $\pi_i\in \SSS_{(\AAA,\pi_0)}(U_i)$, there exists by definition a unitary operator
$u_i : H\to H$ such that $u_i\circ \pi_{is}(-) = \pi_{0s}(-)\circ u_i$, for all $s\in U^\cc$.
(Note that this condition specializes to the commutativity
property $u_i\circ \pi_{0s}(-) = \pi_{0s}(-)\circ u_i$, for all $s\in U_i^\cc\cap U^\cc$.)
We claim that the composition $u:= u_1\circ\cdots\circ u_n : H\to H$ of these unitaries
(with the distinguished $u_n$ on the right, but otherwise an arbitrary composition order) satisfies
\begin{subequations}
\begin{flalign}
u \circ (\pi_1\bullet\cdots\bullet\pi_n)_s(-) \,=\, \pi_{0s}(-)\circ u~,
\end{flalign}
for all $s\in U^\cc$, which would complete the proof
that $\pi_1\bullet\cdots\bullet\pi_n\in \SSS_{(\AAA,\pi_0)}(V)$ defines an object.
For $s\in U_1^\cc \cap\cdots \cap U_n^\cc\cap U^\cc$, this follows directly
from the commutativity of $u_i$ and $\pi_{0s}$ in this case. For $i=1,\dots,n-1$ and $s\in U_i\cap U^\cc$,
this follows from the calculation
\begin{flalign}
\nn u_1\circ\cdots\circ u_n \circ \pi_{is}(-) \,&=\, 
u_1\circ\cdots\circ u_i \circ \pi_{is}(-) \circ u_{i+1}\circ \cdots \circ u_n\\
\,&=\, u_1\circ\cdots\circ u_{i-1} \circ \pi_{0s}(-) \circ u_{i}\circ \cdots \circ u_n\,=\,
\pi_{0s}(-)\circ u_1\circ\cdots\circ u_n~,
\end{flalign}
where the first step uses Assumption \ref{assu:1} (3), the second step 
uses the intertwining property of $u_i$, and the last step follows from 
the commutativity of $u_j$ and $\pi_{0s}$, noting that $s\in U_i\cap U^\cc\subseteq U_j^\cc\cap U^\cc$
for all $j\neq i$. For $s\in U_n\cap U^\cc$, this follows directly
from a similar calculation 
\begin{flalign}
u_1\circ\cdots\circ u_n \circ \pi_{ns}(-) \,=\, 
u_1\circ\cdots\circ u_{n-1} \circ \pi_{0s}(-) \circ u_n
\,=\,
\pi_{0s}(-)\circ u_1\circ\cdots\circ u_n~,
\end{flalign}
\end{subequations}
which does not require Assumption \ref{assu:1} (3).
\sk

To verify that $L_1\bullet \cdots\bullet 
L_n: \pi_1\bullet\cdots\bullet\pi_n\to \pi_1^\prime\bullet\cdots\bullet\pi_n^\prime$
defines a morphism in $\SSS_{(\AAA,\pi_0)}(V)$, we have to show that
\begin{flalign}
(L_1\bullet \cdots\bullet L_n) \circ (\pi_1\bullet\cdots\bullet\pi_n)_s(-)\,=\, 
(\pi_1^\prime\bullet\cdots\bullet\pi_n^\prime)_s(-)\circ (L_1\bullet \cdots\bullet L_n)~,
\end{flalign}
for all $s\in S$. For $s\in U_1^\cc\cap\cdots \cap U_n^\cc$,
this follows directly from \eqref{eqn:SSSmorphisms} and the fact that $L_i : \pi_i\to \pi_i^\prime$
is a morphism in $\SSS_{(\AAA,\pi_0)}(U_i)$, for all $i=1,\dots,n$. For $s\in U_i$, this follows 
directly from Assumption \ref{assu:1} (2) and a similar calculation as in the previous paragraph.
\sk

The assignment in \eqref{eqn:PFAstructure} is clearly functorial
by using that $L_i\circ L_j = L_j\circ L_i$, for all $i\neq j$, from Assumption \ref{assu:1} (2).
By a similar argument one shows that this linear functor preserves the $\ast$-involutions
$(L_1\bullet \cdots\bullet L_n)^\ast=(L_1\circ \cdots\circ L_n)^\ast = (L_n\circ\cdots\circ L_1)^\ast = 
L_1^\ast\circ\cdots\circ L_n^\ast=L_1^\ast\bullet\cdots\bullet L_n^\ast$.
\end{proof}

\begin{propo}\label{prop:PFA}
Suppose that Assumptions \ref{assu:0} (3) and \ref{assu:1} are satisfied. 
Then the functor from Proposition \ref{prop:SSSfunctor}, the distinguished objects \eqref{eqn:pointing} 
and the arity $(n\geq 2)$ structure maps from Lemma \ref{lem:PFAstructure} define
a $C^\ast$-categorical prefactorization algebra $\SSS_{(\AAA,\pi_0)} : \P_{\CC(S)^\perp}\to \CastCat$
over the orthogonal category $\CC(S)^\perp$. This prefactorization algebra is 
locally constant in the sense that it assigns to every $1$-ary operation $U\to V$ in 
$\P_{\CC(S)^\perp}$ a unitary equivalence of $C^\ast$-categories.
\end{propo}
\begin{proof}
One has to check that the structure maps defined by these data satisfy the axioms of a 
$C^\ast$-categorical prefactorization algebra from Definition \ref{def:PFA}, 
see in particular the explicit form of these axioms in Remark \ref{rem:PFA}.
This is evident from the definitions. Local constancy follows directly from Proposition 
\ref{prop:SSSfunctor} since the $1$-ary operations $U\to V$ in $\P_{\CC(S)^\perp}$
are precisely the morphisms $U\subseteq V$ in $\CC(S)$.
\end{proof}

We conclude this subsection by showing that Haag duality for $\overline{\AAA}\in \CastAQFT(S)$ 
in the reference $\ast$-representation $(H,\pi_0)\in{}^\ast\mathbf{Rep}_{\AAA}$, for all $U\in\CC(S)$,
implies the algebraic Assumption \ref{assu:1}. 
Let us start with recalling the concept of Haag duality. Using the construction from the second paragraph
of Remark \ref{rem:SSStraditional}, we obtain injective $\ast$-homomorphisms
$\pi_{0U} : \overline{\AAA}(U)\to B(H)$, for all $U\in \Sub(S)$,
and we denote their images by $\pi_{0U}\big(\overline{\AAA}(U)\big)\subseteq B(H)$.
The $\perp$-commutativity property of $\overline{\AAA}\in \CastAQFT(S)$
implies that $\pi_{0U^\cc}\big(\overline{\AAA}(U^\cc)\big) \subseteq \pi_{0U}\big(\overline{\AAA}(U)\big)^\prime$
is contained in the commutant of $\pi_{0U}\big(\overline{\AAA}(U)\big)\subseteq B(H)$, for all $U\in \Sub(S)$.
Applying the commutant to this inclusion yields 
\begin{flalign}\label{eqn:bicommutant}
\pi_{0U}\big(\overline{\AAA}(U)\big)^{\prime\prime}\, \subseteq \, \pi_{0U^\cc}\big(\overline{\AAA}(U^\cc)\big)^\prime~,
\end{flalign}
for all $U\subseteq \Sub(S)$. 
\begin{defi}\label{def:HaagDuality}
The extended $\overline{\AAA}\in \CastAQFT(S)$ 
is said to satisfy \textit{Haag duality for $U\in\Sub(S)$}
in the reference $\ast$-representation $(H,\pi_0)$ 
if the inclusion in \eqref{eqn:bicommutant} is an equality 
\begin{flalign}
\pi_{0U}\big(\overline{\AAA}(U)\big)^{\prime\prime} \,=\, \pi_{0U^\cc}\big(\overline{\AAA}(U^\cc)\big)^\prime~.
\end{flalign}
\end{defi}

\begin{theo}\label{theo:HaagDualityPFA}
Suppose that $\overline{\AAA}\in \CastAQFT(S)$
satisfies Haag duality in the reference $\ast$-representation $(H,\pi_0)$, for all $U\in \CC(S)$.
Then the algebraic Assumption \ref{assu:1} is satisfied.
Hence, provided that the geometric Assumption \ref{assu:0} (3) is satisfied too,
Proposition \ref{prop:PFA} defines a locally constant $C^\ast$-categorical prefactorization algebra 
$\SSS_{(\AAA,\pi_0)} : \P_{\CC(S)^\perp}\to \CastCat$ in this case.
\end{theo}
\begin{proof}
Note that the family of commuting diagrams in \eqref{eqn:assu11} is equivalent to
the condition $\pi_{1U_1}\big(\overline{\AAA}(U_1)\big) \subseteq \pi_{2U_2}\big(\overline{\AAA}(U_2)\big)^\prime$.
We will check this equivalent condition in order to show that Assumption \ref{assu:1} (1) is satisfied.
The strict localization condition in \eqref{eqn:localizedrep} is equivalent to the identity
$\pi_{iU_i^\cc} = \pi_{0U_i^\cc} : \overline{\AAA}(U_i^\cc)\to B(H)$, for $i=1,2$. Using also
$\perp$-commutativity of $\overline{\AAA}\in\CastAQFT(S)$ and Haag duality for $U_i\in\CC(S)$, we obtain
\begin{flalign}\label{eqn:HDtmp1}
\pi_{iU_i}\big(\overline{\AAA}(U_i)\big)\,\subseteq\, \pi_{iU_i^\cc}\big(\overline{\AAA}(U_i^\cc)\big)^\prime \,=\,
\pi_{0U_i^\cc}\big(\overline{\AAA}(U_i^\cc)\big)^\prime\,=\, \pi_{0U_i}\big(\overline{\AAA}(U_i)\big)^{\prime\prime}~,
\end{flalign}
for $i=1,2$. Our claim then follows from
\begin{flalign}
\pi_{1U_1}\big(\overline{\AAA}(U_1)\big)\,\subseteq\, 
\pi_{0U_1}\big(\overline{\AAA}(U_1)\big)^{\prime\prime}\,\subseteq\,
\pi_{0U_2^\cc}\big(\overline{\AAA}(U_2^\cc)\big)^{\prime\prime} \,=\, 
\pi_{2U_2^\cc}\big(\overline{\AAA}(U_2^\cc)\big)^{\prime\prime}\,\subseteq\,
\pi_{2U_2}\big(\overline{\AAA}(U_2)\big)^{\prime}~.
\end{flalign}
The first step uses \eqref{eqn:HDtmp1} for $i=1$,
the second step follows from $U_1\subseteq U_2^\cc$
since $U_1\cap U_2=\varnothing$, the third step uses
the strict localization condition $\pi_{2U_2^\cc} = \pi_{0U_2^\cc}$,
and the last step follows by applying the commutant 
to the first inclusion in \eqref{eqn:HDtmp1} for $i=2$.
\sk

Concerning Assumption \ref{assu:1} (2), we start with observing that,
as a consequence of the strict localizations of the objects 
$\pi_1,\pi_1^\prime\in \SSS_{(\AAA,\pi_0)}(U_1)$
and $\pi_2,\pi_2^\prime\in \SSS_{(\AAA,\pi_0)}(U_2)$,
the property of the morphisms $L_1 : \pi_1\to \pi_1^\prime$
and $L_2: \pi_2\to \pi_2^\prime$ in \eqref{eqn:SSSmorphisms} 
implies that
\begin{flalign}\label{eqn:Llocalization}
L_1\,\in\,\pi_{0U_1^\cc}\big(\overline{\AAA}(U_1^\cc)\big)^\prime\,=\,
\pi_{0U_1}\big(\overline{\AAA}(U_1)\big)^{\prime\prime}~,\qquad
L_2\,\in\,\pi_{0U_2^\cc}\big(\overline{\AAA}(U_2^\cc)\big)^\prime\,=\,
\pi_{0U_2}\big(\overline{\AAA}(U_2)\big)^{\prime\prime}~,
\end{flalign}
where the equalities use Haag duality. Since
applying twice the commutant to the $\perp$-commutativity 
condition $\pi_{0U_1}\big(\overline{\AAA}(U_1)\big)\subseteq\pi_{0U_2}\big(\overline{\AAA}(U_2)\big)^{\prime}$,
for $U_1\cap U_2=\varnothing$, yields
$\pi_{0U_1}\big(\overline{\AAA}(U_1)\big)^{\prime\prime}\subseteq
\pi_{0U_2}\big(\overline{\AAA}(U_2)\big)^{\prime\prime\prime}$,
we find that $L_1\circ L_2 = L_2\circ L_1$ and hence 
the assumption in \eqref{eqn:assu12com} holds true.
The family of commutative diagrams in \eqref{eqn:assu12diag}
is equivalent to the conditions
$L_1\in \pi_{2U_2}\big(\overline{\AAA}(U_2)\big)^\prime$
and $L_2\in \pi_{1U_1}\big(\overline{\AAA}(U_1)\big)^\prime$.
Focusing on the case of $L_1$, we observe that
\begin{flalign}
\pi_{2U_2}\big(\overline{\AAA}(U_2)\big)\,\subseteq \, \pi_{0U_2^\cc}\big(\overline{\AAA}(U_2^\cc)\big)^\prime
\,\subseteq\,\pi_{0U_1}\big(\overline{\AAA}(U_1)\big)^\prime~,
\end{flalign}
where the first step is given by \eqref{eqn:HDtmp1} and the second step follows from $U_1\subseteq U_2^\cc$.
Applying the commutant to this inclusion, we obtain $\pi_{0U_1}\big(\overline{\AAA}(U_1)\big)^{\prime\prime}\subseteq
\pi_{2U_2}\big(\overline{\AAA}(U_2)\big)^\prime $, hence our claim that
$L_1\in \pi_{2U_2}\big(\overline{\AAA}(U_2)\big)^\prime$ follows from \eqref{eqn:Llocalization}.
Exchanging the labels $1$ and $2$, the same argument shows that $L_2\in \pi_{1U_1}\big(\overline{\AAA}(U_1)\big)^\prime$.
\sk

Concerning Assumption \ref{assu:1} (3), we observe that,
as a consequence of the strict localization properties
of the domain and codomain of the unitaries \eqref{eqn:transporter}, it follows that
\begin{flalign}\label{eqn:u2local}
u_2\,\in\, \pi_{0(U_2^\cc\cap U^\cc)}\big(\overline{\AAA}(U_2^\cc\cap U^\cc)\big)^\prime\,
\subseteq\, \pi_{0U_1}\big(\overline{\AAA}(U_1)\big)^{\prime}\,=\,
\pi_{0U_1^\cc}\big(\overline{\AAA}(U_1^\cc)\big)^{\prime\prime}~,
\end{flalign}
where the first subset inclusion follows from $U_1\subseteq U_2^\cc \cap U^\cc$.
The last equality follows from Haag duality
$\pi_{0U_1^\cc}\big(\overline{\AAA}(U_1^\cc)\big)^{\prime} = 
\pi_{0U_1}\big(\overline{\AAA}(U_1)\big)^{\prime\prime}$ for $U_1\in \CC(S)$
by taking once more the commutant $\pi_{0U_1}\big(\overline{\AAA}(U_1)\big)^{\prime\prime\prime}
=\pi_{0U_1^\cc}\big(\overline{\AAA}(U_1^\cc)\big)^{\prime\prime}$ 
and using that $(-)^{\prime\prime\prime} = (-)^\prime$.
The family of commutative diagrams \eqref{eqn:assu13} is equivalent to the condition
$u_2\in \pi_{1U_1}\big(\overline{\AAA}(U_1)\big)^\prime$. By the same arguments as in the
paragraphs above, we have that
\begin{flalign}
\pi_{1U_1}\big(\overline{\AAA}(U_1)\big) \,\subseteq\, \pi_{1U_1^\cc}\big(\overline{\AAA}(U_1^\cc)\big)^\prime
\,=\, \pi_{0U_1^\cc}\big(\overline{\AAA}(U_1^\cc)\big)^\prime~.
\end{flalign}
Applying the commutant to this inclusion we get $\pi_{0U_1^\cc}\big(\overline{\AAA}(U_1^\cc)\big)^{\prime\prime}
\subseteq \pi_{1U_1}\big(\overline{\AAA}(U_1)\big)^\prime$, hence our claim that
$u_2\in \pi_{1U_1}\big(\overline{\AAA}(U_1)\big)^\prime$ follows from \eqref{eqn:u2local}.
\end{proof}

\subsection{\label{subsec:monoidal}Object-wise monoidal structure}
We will show later in Section \ref{sec:lattice} that the locally constant $C^\ast$-categorical
prefactorization algebra $\SSS_{(\AAA,\pi_0)} : \P_{\CC(S)^\perp}\to \CastCat$
from Proposition \ref{prop:PFA} describes only a certain
part of the structure of the braided monoidal $C^\ast$-categories
of superselection sectors for lattice AQFTs on $S=\bbZ^2$ exhibited in 
\cite{Naaijkens1,Naaijkens2,NaaijkensChapter,FiedlerNaaijkens,Ogata}. The other part
is an (independent) monoidal structure on each $C^\ast$-category
$\SSS_{(\AAA,\pi_0)} (U)\in\CastCat$ which is compatible with this prefactorization algebra structure.
The construction of such object-wise monoidal structures is facilitated by Haag duality 
and it goes back to the early works of Doplicher, Haag and Roberts \cite{DHR} and 
Buchholz and Fredenhagen \cite{BuchholzFredenhagen}. 
Let us briefly recall the main tools entering these constructions. 
\begin{lem}\label{lem:Haagconsequences}
Suppose that the extended $\overline{\AAA}\in \CastAQFT(S)$ in \eqref{eqn:Aoverline}
satisfies Haag duality in the reference $\ast$-representation $(H,\pi_0)$, for all $U\in \CC(S)$.
\begin{itemize}
\item[(1)] Given any $\pi\in \SSS_{(\AAA,\pi_0)}(U)$, the $\ast$-homomorphism
$\pi_U : \overline{\AAA}(U) \to B(H)$ factorizes 
\begin{equation}\label{eqn:pifactorization}
\begin{tikzcd}
\ar[dr,"\pi_U"'] \overline{\AAA}(U) \ar[rr,"\pi_U"] && B(H)\\
 & \pi_{0U}\big(\overline{\AAA}(U)\big)^{\prime\prime} \ar[ru, hookrightarrow] 
\end{tikzcd}
\end{equation}
through the bicommutant $\pi_{0U}\big(\overline{\AAA}(U)\big)^{\prime\prime}\subseteq B(H)$.

\item[(2)] Suppose that the geometric Assumption \ref{assu:0} (1) is satisfied.
Given any $\pi\in \SSS_{(\AAA,\pi_0)}(U)$, the $\ast$-homomorphism $\pi_U : \overline{\AAA}(U)\to 
\pi_{0U}\big(\overline{\AAA}(U)\big)^{\prime\prime}$ in the factorization \eqref{eqn:pifactorization} extends uniquely 
along $\pi_{0U}: \overline{\AAA}(U)\to \pi_{0U}\big(\overline{\AAA}(U)\big)^{\prime\prime}$ 
to a weakly continuous $\ast$-endomorphism
\begin{equation}\label{eqn:piendo}
\begin{tikzcd}
& \overline{\AAA}(U) \ar[dl,"\pi_{0U}"'] \ar[dr,"\pi_U"] \\
\pi_{0U}\big(\overline{\AAA}(U)\big)^{\prime\prime} \ar[rr,dashed, "\rho_U"'] & 
 & \pi_{0U}\big(\overline{\AAA}(U)\big)^{\prime\prime}
\end{tikzcd}
\end{equation}
of the bicommutant $\pi_{0U}\big(\overline{\AAA}(U)\big)^{\prime\prime}\subseteq B(H)$.

\item[(3)] Suppose that the geometric Assumption \ref{assu:0} (1) is satisfied.
Given any morphism $U\subseteq V$ in $\CC(S)$ and any
$\pi\in \SSS_{(\AAA,\pi_0)}(U)\subseteq \SSS_{(\AAA,\pi_0)}(V)$,
then applying the construction in item (2) to both $U$ and $V$ 
yields the commutative diagram
\begin{equation}
\begin{tikzcd}
\pi_{0U}\big(\overline{\AAA}(U)\big)^{\prime\prime} \ar[r,"\rho_U"] \ar[d, hookrightarrow] & \pi_{0U}\big(\overline{\AAA}(U)\big)^{\prime\prime}\ar[d,hookrightarrow]\\
\pi_{0V}\big(\overline{\AAA}(V)\big)^{\prime\prime} \ar[r,"\rho_V"'] & \pi_{0V}\big(\overline{\AAA}(V)\big)^{\prime\prime}
\end{tikzcd}
\end{equation}
of $\ast$-homomorphisms.

\item[(4)] Any morphism $L : \pi\to \pi^\prime$ in $\SSS_{(\AAA,\pi_0)}(U)$
defines an element $L\in \pi_{0U}\big(\overline{\AAA}(U)\big)^{\prime\prime}\subseteq B(H)$.

\item[(5)] Given any $\pi\in \SSS_{(\AAA,\pi_0)}(U)$ and any $V\in\CC(S)$, 
the unitary isomorphism $u: \pi\to \pi^u$ from \eqref{eqn:transporter} to the associated 
$\pi^{u}\in \SSS_{(\AAA,\pi_0)}(V)$ strictly localized in $V$ defines an element
$u\in  \pi_{0W}\big(\overline{\AAA}(W)\big)^{\prime\prime}\subseteq B(H)$,
for every $W\in\CC(S)$ such that $U\cup V\subseteq W$.
\end{itemize}
\end{lem}
\begin{proof}
Item (1) follows from 
\begin{flalign}
\pi_U\big(\overline{\AAA}(U)\big)\,\subseteq\, \pi_{U^\cc}\big(\overline{\AAA}(U^\cc)\big)^\prime
\,=\,  \pi_{0U^\cc}\big(\overline{\AAA}(U^\cc)\big)^\prime\,=\, \pi_{0U}\big(\overline{\AAA}(U)\big)^{\prime\prime}~,
\end{flalign}
where the first step uses $\perp$-commutativity, the second step uses the strict localization of $\pi$
in $U$ and the third step uses Haag duality.
\sk

To show item (2), we choose any $V\in \CC(S)$ such that $V\subseteq U^\cc$,
which exists by Assumption \ref{assu:0} (1).
Using the localizability condition from \eqref{eqn:localizablerep},
there exists a unitary operator $u:H\to H$ such that 
$u\circ \pi_{s}(-) = \pi_{0s}(-)\circ u$, for all $s\in V^\cc$. Since by hypothesis
$U\subseteq V^\cc$, this implies that $\pi_{U}(-) =u^\ast\circ \pi_{0U}(-)\circ u : \overline{\AAA}(U)\to
B(H)$. Using the $\ast$-isomorphism $\pi_{0U}: \overline{\AAA}(U)\to \pi_{0U}\big(\overline{\AAA}(U)\big)$
to identify $\overline{\AAA}(U)$ with $\pi_{0U}\big(\overline{\AAA}(U)\big)$,
we define the $\ast$-homomorphism
\begin{flalign}
\rho_U\,:=\, \pi_U\, \pi_{0U}^{-1}\,:\,  \pi_{0U}\big(\overline{\AAA}(U)\big)~\xrightarrow{\;\quad\;}~
\pi_{0U}\big(\overline{\AAA}(U)\big)^{\prime\prime}~,~~
a~\xmapsto{\;\quad\;}~ \rho_U(a) \,=\, u^\ast \circ a\circ u ~.
\end{flalign}
Since adjoint actions by unitaries are weakly continuous, $\rho_U$ admits a unique continuous
extension to the desired $\ast$-endomorphism on the weak 
completion $\pi_{0U}\big(\overline{\AAA}(U)\big)^{\prime\prime}$ of $\pi_{0U}\big(\overline{\AAA}(U)\big)\subseteq B(H)$.
\sk

To show item (3), note that by weak continuity it suffices to pre-compose this diagram with
the dense map $\pi_{0U} : \overline{\AAA}(U)\to \pi_{0U}\big(\overline{\AAA}(U)\big)^{\prime\prime}$.
The result then follows from the diagram
\begin{equation}
\begin{tikzcd}
\overline{\AAA}(U) \ar[r,"\pi_{0U}"]\ar[dd] & \pi_{0U}\big(\overline{\AAA}(U)\big)^{\prime\prime} \ar[r,"\rho_U"]\ar[dd,hookrightarrow] & \pi_{0U}\big(\overline{\AAA}(U)\big)^{\prime\prime} \ar[rd, hookrightarrow]\ar[dd, hookrightarrow]\\[-5mm]
& & & B(H)\ar[ld, hookleftarrow]\\[-5mm]
\overline{\AAA}(V) \ar[r,"\pi_{0V}"'] & \pi_{0V}\big(\overline{\AAA}(V)\big)^{\prime\prime} \ar[r,"\rho_V"'] & \pi_{0V}\big(\overline{\AAA}(V)\big)^{\prime\prime}
\end{tikzcd}~,
\end{equation}
which commutes as a consequence of $\pi_U$ and $\pi_V$ satisfying commutative triangles
as in \eqref{eqn:trianglerep}.
\sk

To show item (4), note that the strict localization of 
$\pi$ and $\pi^\prime$ in $U$ and the 
definition of morphisms in \eqref{eqn:SSSmorphisms}
implies that $L\circ \pi_{0U^\cc}(-) = L\circ \pi_{U^\cc}(-) = \pi^\prime_{U^\cc}(-)\circ L = 
\pi_{0U^\cc}(-)\circ L$. Hence, $L\in\pi_{0U^\cc}\big(\overline{\AAA}(U^\cc)\big)^\prime = 
\pi_{0U}\big(\overline{\AAA}(U)\big)^{\prime\prime}$ by Haag duality.
\sk

To show item (5), note that the strict localization of 
$\pi$ in $U$ and of $\pi^u$ in $V$ 
implies that $u\circ \pi_{0(U^\cc\cap V^\cc)}(-) = u\circ \pi_{U^\cc\cap V^\cc}(-) = 
\pi^u_{U^\cc\cap V^\cc}(-)\circ u = \pi_{0(U^\cc\cap V^\cc)}(-)\circ u$. 
Since $U\cup V\subseteq W$ implies $W^\cc \subseteq U^\cc\cap V^\cc$, it follows
that $u\in\pi_{0W^\cc}\big(\overline{\AAA}(W^\cc)\big)^\prime = 
\pi_{0W}\big(\overline{\AAA}(W)\big)^{\prime\prime}$ by Haag duality.
\end{proof}

Using these tools we can now define the object-wise monoidal structure.
\begin{lem}\label{lem:monoidal}
Suppose that the geometric Assumptions \ref{assu:0} (1) and (2) are satisfied, and that 
$\overline{\AAA}\in \CastAQFT(S)$ satisfies Haag duality
in the reference $\ast$-representation $(H,\pi_0)$, 
for all $U\in \CC(S)$. For each $U\in \CC(S)$, consider the assignment
from the algebraic tensor product
\begin{subequations}\label{eqn:monoidal}
\begin{flalign}
\cdiamond \,:\, \SSS_{(\AAA,\pi_0)}(U)\otimes \SSS_{(\AAA,\pi_0)}(U)~&\xrightarrow{\;\quad\;}~\SSS_{(\AAA,\pi_0)}(U)~,\\
\nn (\pi,\dot{\pi})~&\xmapsto{\;\quad\;}~\pi\cdiamond\dot{\pi}~,\\
\nn L \otimes \dot{L}~&\xmapsto{\;\quad\;}~L\cdiamond\dot{L}
\end{flalign}
given by
\begin{flalign}\label{eqn:monoidalformula}
(\pi\cdiamond \dot{\pi})_s\,:=\,\begin{cases}
\rho_U \, \dot{\rho}_U\,\pi_{0s} &~~\text{if } s\in U\,,\\
\pi_{0s} &~~\text{if }s\in U^\cc\,,
\end{cases}
\end{flalign}
for all $s\in S$, and
\begin{flalign}\label{eqn:monoidalformulamorphism}
L\cdiamond \dot{L}\,:=\, L\circ\rho_U(\dot{L})~,
\end{flalign}
\end{subequations}
where $\rho_U, \dot{\rho}_U:\pi_{0U}\big(\overline{\AAA}(U)\big)^{\prime\prime}\to 
\pi_{0U}\big(\overline{\AAA}(U)\big)^{\prime\prime}$ denote the $\ast$-endomorphisms
associated uniquely to $\pi$ and $\dot{\pi}$ via Lemma \ref{lem:Haagconsequences}.
Then \eqref{eqn:monoidal} is a well-defined $\ast$-functor.
\end{lem}
\begin{proof}
The family of $\ast$-homomorphisms \eqref{eqn:monoidalformula} 
is mutually commuting: For $s,s^\prime\in U$ or $s,s^\prime\in U^\cc$ with 
$s\neq s^\prime$, this follows from the fact that 
$\pi_{0s}$ and $\pi_{0s^\prime}$ are mutually commuting.
For $s\in U$ and $s^\prime\in U^\cc$, this follows from 
$\pi_{0s^\prime}\big(\AAA(s^\prime)\big) \subseteq \pi_{0U}\big(\overline{\AAA}(U)\big)^{\prime}=
\pi_{0U}\big(\overline{\AAA}(U)\big)^{\prime\prime\prime}$,
where the first step uses $\perp$-commutativity. Note that
$(\pi\cdiamond \dot{\pi})_s = \pi_{0s}$, for all $s\in U^\cc$,
is by definition strictly localized in $U$. We claim 
that it satisfies also the localizability condition from Definition \ref{def:SSSloc}, for all other 
$V\in\CC(S)$. Using Assumption \ref{assu:0} (2) and Remark 
\ref{rem:smalllocalizableisenough}, it suffices to prove this claim  
only for those objects $V\in\CC(S)$ such that 
there exists $W\in\CC(S)$ with $U\cup V\subseteq W$.
Since $\pi,\dot{\pi}\in \SSS_{(\AAA,\pi_0)}(U)$, there exist by definition unitary operators
$u,\dot{u} : H\to H$ such that $\pi^u,\dot{\pi}^{\dot{u}}\in \SSS_{(\AAA,\pi_0)}(V)$ are
strictly localized in $V$. Using further that $U\cup V\subseteq W$, 
we obtain from Lemma \ref{lem:Haagconsequences} (5) that
$u,\dot{u}\in \pi_{0W}\big(\overline{\AAA}(W)\big)^{\prime\prime}\subseteq B(H)$.
We claim that the unitary operator $\tilde{u} := u\circ \rho_W(\dot{u}): H\to H$,
with $\rho_W$ described in Lemma \ref{lem:Haagconsequences} (3), defines a unitary isomorphism
$\tilde{u} : \pi\cdiamond \dot{\pi}\to (\pi\cdiamond\dot{\pi})^{\tilde{u}}$ 
with $(\pi\cdiamond\dot{\pi})^{\tilde{u}}$ 
localized strictly in $V$, i.e.\
\begin{subequations}
\begin{flalign}\label{eqn:monoidalcheckT}
\tilde{u}\circ (\pi\cdiamond \dot{\pi})_s(-)\,=\, \pi_{0s}(-)\circ \tilde{u}~,
\end{flalign}
for all $s\in V^\cc$. In order to verify these identities, it is convenient to
observe that Lemma \ref{lem:Haagconsequences} (3) for $U\subseteq W$ allows us to 
rewrite \eqref{eqn:monoidalformula} equivalently as
\begin{flalign}
(\pi\cdiamond \dot{\pi})_s\,=\,\begin{cases}
\rho_W \, \dot{\rho}_W\,\pi_{0s} ~&~~\text{if } s\in W\,,\\
\pi_{0s}~&~~\text{if }s\in W^\cc\,.
\end{cases}
\end{flalign}
For $s\in V^\cc\cap W^\cc$, one then verifies \eqref{eqn:monoidalcheckT} by
\begin{flalign}
u\circ \rho_W(\dot{u})\circ \pi_{0s}(-)\,=\, \pi_{0s}(-)\circ u\circ \rho_W(\dot{u})~,
\end{flalign}
where we use $\perp$-commutativity in the form $\pi_{0s}\big(\AAA(s)\big)\subseteq
\pi_{0W}\big(\overline{\AAA}(W)\big)^\prime = \pi_{0W}\big(\overline{\AAA}(W)\big)^{\prime\prime\prime}$
together with Lemma \ref{lem:Haagconsequences} (5). For $s\in V^\cc\cap W$, we compute,
for all $a\in \AAA(s)$,
\begin{flalign}
\nn u\circ \rho_W(\dot{u})\circ \rho_W\big(\dot{\rho}_W\big(\pi_{0s}(a)\big)\big)
\,&=\,u\circ \rho_W(\dot{u})\circ \rho_W\big(\dot{\pi}_{s}(a)\big)
\,=\, u\circ \rho_W\big(\dot{u}\circ \dot{\pi}_{s}(a)\big)\big) \\
\nn \,&=\, u\circ \rho_W\big(\pi_{0s}(a)\circ \dot{u}\big)\,=\, 
u\circ \rho_W\big(\pi_{0s}(a)\big)\circ \rho_W(\dot{u}) \\
\,&=\, u\circ \pi_{s}(a)\circ \rho_W(\dot{u}) \,=\, \pi_{0s}(a)\circ u\circ \rho_W(\dot{u})~.
\end{flalign}
\end{subequations}
Steps one and five use Lemma \ref{lem:Haagconsequences} (2), 
steps two and four use that the $\rho$'s are $\ast$-endomorphisms,
and steps three and six use the intertwining properties of the $u$'s.
\sk

To verify that $L\cdiamond \dot{L} :\pi\cdiamond\dot{\pi}\to \pi^\prime\cdiamond \dot{\pi}^\prime$
defines a morphism in $\SSS_{(\AAA,\pi_0)}(U)$, we have to show that
\begin{subequations}
\begin{flalign}
(L\cdiamond \dot{L}) \circ (\pi\cdiamond \dot{\pi})_s(-)\,=\, 
(\pi^\prime\cdiamond \dot{\pi}^\prime)_s(-)\circ (L\cdiamond \dot{L}) ~,
\end{flalign}
for all $s\in S$. For $s\in U^\cc$, this follows directly from Lemma \ref{lem:Haagconsequences} (4).
For $s\in U$, we compute, for all $a\in \AAA(s)$,
\begin{flalign}
L\circ \rho_U(\dot{L})\circ \rho_U\big(\dot{\rho}_U\big(\pi_{0s}(a)\big)\big)
\,&=\,L\circ \rho_U\big(\dot{L}\circ \dot{\rho}_U\big(\pi_{0s}(a)\big)\big)
\,=\, L\circ \rho_U\big(\dot{\rho}^\prime_U\big(\pi_{0s}(a)\big)\circ \dot{L}\big)\\
\,&=\,   L\circ \rho_U\big(\dot{\rho}^\prime_U\big(\pi_{0s}(a)\big)\big)\circ \rho_U(\dot{L})
\,=\,  \rho^\prime_U\big(\dot{\rho}^\prime_U\big(\pi_{0s}(a)\big)\big)\circ L\circ \rho_U(\dot{L})~. \nn
\end{flalign}
\end{subequations}
Steps one and three use that the $\rho$'s are $\ast$-endomorphisms.
Steps two and four use the intertwining properties \eqref{eqn:SSSmorphisms} 
between the $L$'s and the $\pi$'s, which imply via Lemma \ref{lem:Haagconsequences} (2)
analogous intertwining properties between the $L$'s and the $\rho$'s.
\sk

The assignment \eqref{eqn:monoidal} clearly 
preserves the identity morphisms. To show that it preserves compositions, 
we have to check the interchange law
\begin{subequations}
\begin{flalign}
(L^\prime\circ L)\cdiamond (\dot{L}^\prime\circ\dot{L}) \,=\,
(L^\prime\cdiamond \dot{L}^\prime) \circ (L\cdiamond  \dot{L})\,:\, \pi\cdiamond \dot{\pi}~\xrightarrow{\;\quad\;}~
\pi^{\prime\prime}\cdiamond \dot{\pi}^{\prime\prime} ~,
\end{flalign}
for any four morphisms
$L : \pi \to \pi^\prime$, $L^\prime : \pi^\prime\to \pi^{\prime\prime}$,
$\dot{L} : \dot{\pi} \to \dot{\pi}^\prime$ and 
$\dot{L}^\prime : \dot{\pi}^\prime\to \dot{\pi}^{\prime\prime}$.
This is shown by the following computation
\begin{flalign}
\nn (L^\prime\circ L)\cdiamond (\dot{L}^\prime\circ\dot{L})\,&=\,
L^\prime\circ L\circ \rho_U(\dot{L}^\prime)\circ\rho_U(\dot{L})
\,=\, L^\prime\circ \rho^\prime_U(\dot{L}^\prime) \circ L \circ\rho_U(\dot{L})\\
\,&=\,(L^\prime\cdiamond \dot{L}^\prime) \circ (L\cdiamond  \dot{L})~,
\end{flalign}
\end{subequations}
where the second step uses the intertwining property \eqref{eqn:SSSmorphisms} 
of $L : \pi \to \pi^\prime$.
One also easily checks that this linear functor
preserves the $\ast$-involutions 
\begin{flalign}
(L\cdiamond \dot{L})^\ast \,=\, 
(L\circ \rho_U(\dot{L}))^\ast
\,=\, (\rho^\prime_U(\dot{L})\circ L)^\ast \,=\, 
L^\ast\circ\rho^\prime_U(\dot{L}^\ast) \,=\,
L^\ast\cdiamond \dot{L}^\ast~.\qedhere
\end{flalign}
\end{proof}

\begin{propo}\label{prop:monoidal}
Suppose that the geometric Assumptions \ref{assu:0} (1) and (2) are satisfied, and that 
$\overline{\AAA}\in \CastAQFT(S)$ satisfies Haag duality in the 
reference $\ast$-representation $(H,\pi_0)$, for all $U\in \CC(S)$. 
Then the $\ast$-functor from Lemma \ref{lem:monoidal}
and the distinguished object \eqref{eqn:pointing} endow $\SSS_{(\AAA,\pi_0)}(U)\in\CastCat$ with 
the structure of a strict monoidal $C^\ast$-category, for all $U\in\CC(S)$.
\end{propo}
\begin{proof}
This is a simple check using the definitions in \eqref{eqn:monoidal}.
To check associativity on objects $\pi\cdiamond (\dot{\pi}\cdiamond \ddot{\pi})
=(\pi\cdiamond \dot{\pi})\cdiamond \ddot{\pi}$, we observe that the
$s\in U^\cc$ components of both sides give $\pi_{0s}$, while for $s\in U$ we have
\begin{flalign}
\big(\pi \cdiamond (\dot{\pi} \cdiamond \ddot{\pi})\big)_s
\,=\,\rho_U\,\dot{\rho}_U\,\ddot{\rho}_U\,\pi_{0s}
\,=\, \big((\pi \cdiamond \dot{\pi})\cdiamond \ddot{\pi}\big)_s~.
\end{flalign}
To check associativity on morphisms, we compute
\begin{flalign}
L \cdiamond (\dot{L} \cdiamond \ddot{L})\,=\,
L\circ \rho_U(\dot{L}) \circ \rho_U\big(\dot{\rho}_U(\ddot{L})\big) \, =\,
(L \cdiamond \dot{L}) \cdiamond \ddot{L}~.
\end{flalign}
Unitality  $\pi\cdiamond \pi_0 =\pi= \pi_0\cdiamond\pi$ is obvious
since the $\ast$-endomorphism corresponding to $\pi_0$ via Lemma \ref{lem:Haagconsequences}
is given by the identity $\rho_{0U} = \id$.
\end{proof}

\subsection{\label{subsec:compatibility}Compatibility of the two structures}
We will now show that the prefactorization algebra and object-wise monoidal 
structures from the previous two subsections are compatible with each other. 
This leads to the following main result.
\begin{theo}\label{theo:PFAinMonCat}
Suppose that the geometric Assumption \ref{assu:0} is satisfied
and that the extended $\overline{\AAA}\in \CastAQFT(S)$ in \eqref{eqn:Aoverline}
satisfies Haag duality in the reference $\ast$-representation $(H,\pi_0)$, for all $U\in \CC(S)$. 
Then the object-wise monoidal structures from Proposition \ref{prop:monoidal}
are compatible with the prefactorization algebra structure from Theorem \ref{theo:HaagDualityPFA}
in the sense that they define a locally constant prefactorization algebra
\begin{subequations}
\begin{flalign}
\SSS_{(\AAA,\pi_0)}\,:\,\P_{\CC(S)^\perp}~\xrightarrow{\;\quad\;}~\Alg_{\mathsf{uAs}}\big(\CastCat\big)
\end{flalign}
with values in the symmetric monoidal category $\Alg_{\mathsf{uAs}}\big(\CastCat\big)$
of strict monoidal $C^\ast$-categories and strict monoidal $\ast$-functors. 
This is equivalent to the datum of a unital associative algebra
\begin{flalign}
\SSS_{(\AAA,\pi_0)}\,\in\,\Alg_{\mathsf{uAs}}\Big(C^\ast\mathbf{PFA}_{\CC(S)^\perp}\Big)
\end{flalign}
\end{subequations}
in the category of $C^\ast$-categorical prefactorization algebras $C^\ast\mathbf{PFA}_{\CC(S)^\perp}$
from Definition \ref{def:PFA} (endowed with the object-wise symmetric monoidal structure), 
whose underlying prefactorization algebra is locally constant.
\end{theo}
\begin{proof}
We have to show that the structure maps 
\begin{flalign}
\bullet \,:\, \bigotimes_{i=1}^n\SSS_{(\AAA,\pi_0)}(U_i)~\xrightarrow{\;\quad\;}~
\SSS_{(\AAA,\pi_0)}(V)
\end{flalign} 
in \eqref{eqn:PFAstructure} are strict monoidal $\ast$-functors, for all operations
$\und{U} = (U_1,\dots,U_n)\to V$ in the operad $\P_{\CC(S)^\perp}$. For arity $n=0$ 
and $n=1$ operations this is evident, so we consider the case of $n\geq 2$.
As a direct consequence of its definition in \eqref{eqn:PFAstructureformula},
it follows that $\bullet$ preserves the monoidal units 
$\pi_0\bullet\cdots\bullet\pi_0 = \pi_0$. Concerning the monoidal product
on objects, we have to verify the interchange law
\begin{flalign}\label{eqn:interchangecompatibility}
\big((\pi_1\cdiamond \dot{\pi}_1)\bullet \cdots\bullet(\pi_n\cdiamond \dot{\pi}_n)\big)_s
\, =\,\big((\pi_1\bullet \cdots\bullet\pi_n)\cdiamond(\dot{\pi}_1\bullet \cdots\bullet\dot{\pi}_n)\big)_s~,
\end{flalign}
for all $s\in S$. This is obvious for $s\in V^\cc$.
For $s\in V$, the right-hand side of 
\eqref{eqn:interchangecompatibility} reads as
\begin{flalign}
\big((\pi_1\bullet \cdots\bullet\pi_n)\cdiamond(\dot{\pi}_1\bullet \cdots\bullet\dot{\pi}_n)\big)_s
\,=\, \rho_V\,\dot{\rho}_V\,\pi_{0s}~,
\end{flalign}
where $\rho_V$ and $\dot{\rho}_V$ denote the $\ast$-endomorphisms
corresponding to $\pi_1\bullet \cdots\bullet\pi_n$ and $
\dot{\pi}_1\bullet \cdots\bullet\dot{\pi}_n$ respectively.
These are defined uniquely according to Lemma \ref{lem:Haagconsequences} (2) by 
the factorizations $(\pi_1\bullet \cdots\bullet\pi_n)_V = \rho_V\,\pi_{0V}$
and $(\dot{\pi}_1\bullet \cdots\bullet\dot{\pi}_n)_V = \dot{\rho}_V\,\pi_{0V}$.
Given any $s\in V\cap U_1^\cc\cap \cdots\cap U_n^\cc$, restricting along $\AAA(s)\to \overline{\AAA}(V)$
yields the identities $\pi_{0s} = \rho_V\,\pi_{0s}$ and $\pi_{0s}=\dot{\rho}_V\,\pi_{0s}$,
hence $\rho_V\,\dot{\rho}_V\,\pi_{0s} = \pi_{0s}$, which implies that \eqref{eqn:interchangecompatibility}
holds true also for all  $s\in V\cap U_1^\cc\cap \cdots\cap U_n^\cc$.
By restricting along  $\overline{\AAA}(U_i)\to \overline{\AAA}(V)$ we obtain
the identities $\pi_{iU_i}=\rho_{iU_i}\,\pi_{0U_i} = \rho_V\,\pi_{0U_i}$ and $\dot{\pi}_{iU_i}=
\dot{\rho}_{iU_i}\,\pi_{0U_i}=\dot{\rho}_V\,\pi_{0U_i}$. Hence, for all $s\in U_i$, 
we have that $\rho_V\,\dot{\rho}_V\,\pi_{0s}
=\rho_{V}\,\dot{\rho}_{iU_i}\,\pi_{0s}
=\rho_{i U_i}\,\dot{\rho}_{iU_i}\,\pi_{0s}$, because the image of
$\dot{\rho}_{iU_i}$ lies in $\pi_{0U_i}\big(\overline{\AAA}(U_i)\big)^{\prime\prime}$. This implies that 
\eqref{eqn:interchangecompatibility} holds true also for all  $s\in U_i$, for $i=1,\dots,n$.
This exhausts all possible cases, so we can conclude that \eqref{eqn:interchangecompatibility} holds true
for all $s\in S$.
\sk

Concerning the monoidal product on morphisms, we have to verify that
\begin{flalign}\label{eqn:interchangecompatibilitymorphism}
(L_1\cdiamond \dot{L}_1)\bullet \cdots\bullet(L_n\cdiamond \dot{L}_n)
\, =\,(L_1\bullet \cdots\bullet L_n)\cdiamond(\dot{L}_1\bullet \cdots\bullet\dot{L}_n)~.
\end{flalign}
Since $L_1\bullet \cdots\bullet L_n : \pi_1\bullet \cdots\bullet\pi_n\to \pi^\prime_1\bullet \cdots\bullet\pi^\prime_n$
is a morphism from $\pi_1\bullet \cdots\bullet\pi_n\in \SSS_{(\AAA,\pi_0)}(V)$,
the right-hand side of \eqref{eqn:interchangecompatibilitymorphism} 
is defined according to \eqref{eqn:monoidalformulamorphism} and \eqref{eqn:PFAstructureformulamorphism} by
\begin{flalign}\label{eqn:interchangecompatibilitymorphismtmp}
\nn (L_1\bullet \cdots\bullet L_n)\cdiamond(\dot{L}_1\bullet \cdots\bullet\dot{L}_n)
\,&=\,L_1\circ \cdots\circ L_n \circ\rho_V\big(\dot{L}_1\circ \cdots\circ \dot{L}_n\big)\\
\,&=\,L_1\circ \cdots\circ L_n \circ\rho_V(\dot{L}_1)\circ \cdots\circ \rho_V(\dot{L}_n)~,
\end{flalign}
where $\rho_V$ is the $\ast$-endomorphism corresponding to $\pi_1\bullet \cdots\bullet\pi_n
\in \SSS_{(\AAA,\pi_0)}(V)$ discussed in the paragraph above.
Since $\dot{L}_i\in \pi_{0U_i}\big(\overline{\AAA}(U_i)\big)^{\prime\prime}$ 
by Lemma \ref{lem:Haagconsequences} (4),
we can  write $\rho_V(\dot{L}_i) = \rho_{iU_i}(\dot{L}_i)$.
Using also $\perp$-commutativity and the fact that $U_i\cap U_j=\varnothing$ are disjoint, for all $i\neq j$,
we can reorder the factors in \eqref{eqn:interchangecompatibilitymorphismtmp} according to  
\begin{flalign}
(L_1\bullet \cdots\bullet L_n)\cdiamond(\dot{L}_1\bullet \cdots\bullet\dot{L}_n)
\,=\,L_1\circ \rho_{1U_1}(\dot{L}_1)\circ \cdots\circ L_n\circ \rho_{nU_n}(\dot{L}_n)~,
\end{flalign}
which implies that the identity \eqref{eqn:interchangecompatibilitymorphism} holds true.
\end{proof}

One can rephrase the result of Theorem \ref{theo:PFAinMonCat} 
in the language of categorified AQFTs introduced in \cite{BPSWcategorified}.
The key ingredient for this reinterpretation is \cite[Theorem 2.9]{BPSWcategorified},
which shows that the Boardman-Vogt tensor product 
$\P_{\CC(S)^\perp}\otimes_{\mathrm{BV}}^{} \mathsf{uAs}
\cong \O_{\CC(S)^\perp}$ of the prefactorization operad from Definition \ref{def:Poperad}
and the unital associative operad $\mathsf{uAs}$ is equivalent to the AQFT operad
$\O_{\CC(S)^\perp}$ from \cite{BSWoperad}. As a direct consequence,
we obtain the following result.
\begin{cor}\label{cor:2AQFT}
Suppose that the geometric Assumption \ref{assu:0} is satisfied
and that the extended $\overline{\AAA}\in \CastAQFT(S)$ in \eqref{eqn:Aoverline}
satisfies Haag duality in the reference $\ast$-representation $(H,\pi_0)$, for all $U\in \CC(S)$. 
Then the locally constant prefactorization algebra from Theorem \ref{theo:PFAinMonCat} is equivalent
to the datum of a  locally constant categorified AQFT
\begin{flalign}
\SSS_{(\AAA,\pi_0)}\,:\,\O_{\CC(S)^\perp}~\xrightarrow{\;\quad\;}~\CastCat
\end{flalign}
over the orthogonal category $\CC(S)^\perp$ with values in the symmetric monoidal category
$\CastCat$ of $C^\ast$-categories and $\ast$-functors.
\end{cor}

\begin{rem}\label{rem:2AQFT}
We believe that this corollary is quite curious and interesting.
It shows that applying superselection theory \cite{DHR,BuchholzFredenhagen},
in the specific form presented above, to an ordinary AQFT $\overline{\AAA}$ 
results in a categorified AQFT $\SSS_{(\AAA,\pi_0)}$. This categorified
AQFT is locally constant, which means that it describes only (some part of) the topological content
of the original not necessarily topological AQFT $\overline{\AAA}$.
Hence, one may interpret the theory of superselection sectors as a 
concrete and powerful tool to ``extract the topological content of a not necessarily topological AQFT''. 
This perspective also matches perfectly the way how superselection theory is
used in practice for the description of topological order in (non-topological) quantum
lattice models, see e.g.\ \cite{Naaijkens1,Naaijkens2,NaaijkensChapter,FiedlerNaaijkens,Ogata}.
\end{rem}


\section{\label{sec:lattice}Quantum systems on the lattice $\mathbb{Z}^n$}
We will now specialize the results from Section \ref{sec:PFA}
to the case where the set $S=\bbZ^n$ is the $n$-dimensional lattice.
Motivated by earlier works on topological order in $2$-dimensional quantum spin 
systems \cite{Naaijkens1,Naaijkens2,NaaijkensChapter,FiedlerNaaijkens,Ogata},
we consider as localization regions the category of cone-shaped subsets of $\bbZ^n$,
which we define by intersecting open cones in the Euclidean space $\bbR^n$
with the lattice $\bbZ^n\subset \bbR^n$. A convenient way to parametrize
open cones in $\bbR^n$ is by their apex $p\in \bbR^n$,
normalized center direction $t\in\bbS^{n-1}\subset \bbR^n$
and opening angle $\alpha\in (0,\pi)$. Explicitly, the open cone associated
to such data is given by the open subset
\begin{flalign}
\nn C_{(p,t,\alpha)}\,&:=\, \big\{x\in\bbR^n\backslash\{p\}\,:\, \mathrm{angle}(x-p,t)\,<\,\alpha \big\}\\[3pt]
\,&\phantom{:}=\,\big\{x \in\bbR^n\,:\,(x-p)\cdot t\,>\,\norm{x-p}\,\cos(\alpha)  \big\}\,\subseteq\,\bbR^n~,\label{eqn:cone}
\end{flalign}
where $\cdot$ and $\norm{\cdot}$ denote the standard inner product and its
associated norm on $\bbR^n$. For an opening angle $\alpha\neq \frac{\pi}{2}$, the tuple of data
$(p,t,\alpha)\in \bbR^n\times \bbS^{n-1}\times (0,\pi)$ specifies open cones in $\bbR^n$ faithfully, i.e.\
$C_{(p,t,\alpha)} = C_{(p^\prime,t^\prime,\alpha^\prime)}$
if and only if $(p^\prime,t^\prime,\alpha^\prime)=(p,t,\alpha)$.
For $\alpha = \frac{\pi}{2}$, i.e.\ in the case where open cones are half-spaces,
there are some mild degeneracies in this parametrization, namely
$C_{(p,t,\frac{\pi}{2})} = C_{(p^\prime,t^\prime,\alpha^\prime)}$
if and only if $\alpha^\prime=\frac{\pi}{2}$, $t^\prime=t$ and
$p^\prime = p + q$ for some $q\in \bbR^n$ with $q\cdot t=0$.
\begin{defi}\label{def:cones}
The \textit{category of cone-shaped subsets of $\bbZ^n\subset \bbR^n$} is defined
as the full subcategory $\Cone(\bbZ^n)\subseteq \Sub(\bbZ^n)$ whose objects
are of the form $U = C_{(p,t,\alpha)}\cap \bbZ^n\subseteq \bbZ^n$, for some 
open cone $C_{(p,t,\alpha)}\subseteq \bbR^n$ as in \eqref{eqn:cone}.
\end{defi}

\begin{propo}\label{prop:cones}
The full subcategory $\Cone(\bbZ^n)\subseteq\Sub(\bbZ^n)$ from Definition \ref{def:cones}
satisfies the geometric Assumption \ref{assu:0}.
\end{propo}
\begin{proof}
We make use of the following geometric facts about open cones in $\bbR^n$:
\begin{itemize}
\item[(G1)] Given any open cone $C_{(p,t,\alpha)}\subseteq \bbR^n$, its complement
in $\bbR^n$ is given by the closure
\begin{flalign}
(C_{(p,t,\alpha)})^\cc\,=\,\overline{C_{(p,-t,\pi-\alpha)}}\,\subseteq\,\bbR^n
\end{flalign}
of the open cone $C_{(p,-t,\pi-\alpha)}\subseteq\bbR^n$ which is
associated with the complementary parameters $(p,-t,\pi-\alpha)\in\bbR^{n}\times\bbS^{n-1}\times(0,\pi)$.

\item[(G2)] Given any open cone $C_{(p,t,\alpha)}\subseteq \bbR^n$,
any point $q\in \overline{C_{(p,t,\alpha)}}\subseteq\bbR^n$ in its closure and 
any normalized direction $u\in\bbS^{n-1}$ such that $\mathrm{angle}(t,u)<\alpha$,
then $C_{(q,u,\beta)}\subseteq C_{(p,t,\alpha)}$ for all $\beta< \alpha-\mathrm{angle}(t,u)$.

\item[(G3)] Given any open cone $C_{(p,t,\alpha)}\subseteq \bbR^n$,
any point $q\in \bbR^n$ and any normalized direction $u\in\bbS^{n-1}$ such
that $\mathrm{angle}(t,u)<\alpha$, then the half-line 
$\bbR_{\geq 0}\ni \lambda\mapsto q+\lambda\,u\in\bbR^n$ is contained eventually in 
$C_{(p,t,\alpha)}\subseteq\bbR^n$, i.e.\ there exists $\lambda_\ast \in\bbR_{\geq 0}$ such that 
$q+\lambda\,u\in C_{(p,t,\alpha)}$ for all $\lambda\geq\lambda_\ast$.
\end{itemize}

Item (1): Given any $U = C_{(p,t,\alpha)}\cap \bbZ^n\in \Cone(\bbZ^n)$,
we use the geometric fact (G1) to define $V = C_{(p,-t,\pi-\alpha)}\cap \bbZ^n\in \Cone(\bbZ^n)$
which satisfies $V\subseteq U^\cc$.
\sk

Item (2): Note that it suffices to consider the 
case where $U = C_{(p,t,\alpha)}\cap \bbZ^n\in \Cone(\bbZ^n)$
and $V = C_{(q,u,\beta)}\cap \bbZ^n\in \Cone(\bbZ^n)$ are such that 
the normalized directions are not antipodal $t\neq -u$. 
Indeed, if $t=-u$, we can apply the geometric fact (G2) to obtain a 
smaller open cone $C_{(q,\tilde{u},\tilde{\beta})}\subseteq C_{(q,u,\beta)}$
with $t\neq -\tilde{u}$ and prove item (2) by exhibiting cone-shaped 
subsets $V^\prime\subseteq \tilde{V}\subseteq V$ which are contained in 
$\tilde{V} = C_{(q,\tilde{u},\tilde{\beta})}\cap \bbZ^n\in\Cone(\bbZ^n)$.
\sk

Therefore, assuming without loss of generality that $t\neq -u$, there exists
an opening angle $\gamma \in (0,\pi)$ with $\gamma>\alpha$ and $\gamma > \mathrm{angle}(t,u)$.
Using (G2) and (G3), this implies that $C_{(p,t,\alpha)}\subseteq C_{(p,t,\gamma)}$
and that the half-line $\bbR_{\geq 0}\ni \lambda\mapsto q+\lambda\,u\in\bbR^n$
is contained eventually in $C_{(p,t,\gamma)}\subseteq \bbR^n$. We set
\begin{flalign}
W\,:=\, C_{(p,t,\gamma)}\cap \bbZ^n\,\in\,\Cone(\bbZ^n)
\end{flalign}
and observe that $U\subseteq W$ holds true by construction.
In order to exhibit $V^\prime \in\Cone(\bbZ^n)$ of the form $V^{\prime}=  C_{(q^\prime,u^\prime,\beta^\prime)}\cap \bbZ^n$,
satisfying $V^\prime\subseteq V$ and $V^\prime\subseteq W$, we use the eventual containedness property of
the half-line and choose any $q^\prime\in C_{(p,t,\gamma)}\cap C_{(q,u,\beta)}$ such that 
$q^\prime + \lambda\,u\in C_{(p,t,\gamma)}\cap C_{(q,u,\beta)}$, for all $\lambda\geq 0$.
Choosing $u^\prime = u$ and any $\beta^\prime\in(0,\pi)$ such that 
$\beta^\prime<\beta$ and $\beta^\prime< \gamma - \mathrm{angle}(t,u)$,
we obtain from (G2) that $C_{(q^\prime,u,\beta^\prime)}\subseteq  C_{(q,u,\beta)}$
and $C_{(q^\prime,u,\beta^\prime)}\subseteq  C_{(p,t,\gamma)}$, hence 
$V^\prime\subseteq V$ and $V^\prime\subseteq W$.
\sk

Item (3): Note that it suffices to consider the 
case where $U = C_{(p,t,\alpha)}\cap \bbZ^n\in \Cone(\bbZ^n)$
and $U_i = C_{(p_i,t_i,\alpha_i)}\cap \bbZ^n\in \Cone(\bbZ^n)$
are such that $\mathrm{angle}(t,t_i)\neq \alpha_i$, for all $i=1,\dots,n$. This follows as above
by applying the geometric fact (G2) in order to shrink the open cone representing
$U$ to a smaller cone with a slightly altered normalized direction which satisfies this property. 
To carry out the proof, we make a case distinction.
In the case where $\mathrm{angle}(t,t_i)>\alpha_i$, for all $i=1,\dots,n$,
one obtains from (G3) that the half-line $\bbR_{\geq 0}\ni \lambda\mapsto p+\lambda\,t\in\bbR^n$ is 
contained eventually in the intersection $\bigcap_{i=1}^n C_{(p_i,-t_i,\pi-\alpha_i)}\subseteq \bbR^n$
of the complementary open cones from (G1). Choosing any 
$q\in \bigcap_{i=1}^n C_{(p_i,-t_i,\pi-\alpha_i)}\cap C_{(p,t,\alpha)}$
such that $q+\lambda\,t\in  \bigcap_{i=1}^n C_{(p_i,-t_i,\pi-\alpha_i)}\cap C_{(p,t,\alpha)}$, for all $\lambda\geq 0$,
and any $\beta\in(0,\pi)$ such that $\beta<\alpha$ and $\beta<\mathrm{angle}(t,t_i)-\alpha_i$, for 
all $i=1,\dots,n$, we obtain from (G2) that $C_{(q,t,\beta)}\subseteq C_{(p,t,\alpha)}$
and $C_{(q,t,\beta)}\cap C_{(p_i,t_i,\alpha_i)}=\varnothing$, for all $i=1,\dots,n$.
Hence, $V = C_{(q,t,\beta)}\cap \bbZ^n\in \Cone(\bbZ^n)$ satisfies
$V\subseteq U$ and $V\cap U_i=\varnothing$, for all $i=1,\dots,n$.
\sk

Consider now the case where $\mathrm{angle}(t,t_i)<\alpha_i$ 
for some $i\in \{1,\dots,n\}$. Then (G3) implies that the 
half-line $\bbR_{\geq 0}\ni \lambda\mapsto p+\lambda\,t\in\bbR^n$ is contained eventually in 
$C_{(p_i,t_i,\alpha_i)}\subseteq \bbR^n$, so we can choose
$q\in C_{(p_i,t_i,\alpha_i)}\cap C_{(p,t,\alpha)}$ such that 
$q+\lambda\,t\in C_{(p_i,t_i,\alpha_i)}\cap C_{(p,t,\alpha)}$,
for all $\lambda\geq 0$. Choosing further any $\beta\in(0,\pi)$ 
such that $\beta<\alpha$ and $\beta<\alpha_i-\mathrm{angle}(t,t_i)$,
we obtain from (G2) that that $C_{(q,t,\beta)}\subseteq C_{(p,t,\alpha)}$
and $C_{(q,t,\beta)}\subseteq C_{(p_i,t_i,\alpha_i)}$. Hence, 
$V = C_{(q,t,\beta)}\cap \bbZ^n\in \Cone(\bbZ^n)$ satisfies
$V\subseteq U$ and $V\subseteq U_i$, which via mutual disjointness 
of $(U_1,\dots,U_n)$ implies that $V\cap U_j=\varnothing$, for all $j\neq i$.
\end{proof}

The following result is a direct consequence of Theorem \ref{theo:PFAinMonCat}
and Proposition \ref{prop:cones}.
\begin{cor}\label{cor:latticePFA}
Let $\AAA = \{\AAA(x)\}_{x\in \bbZ^n}\in\CastAlg^{\bbZ^n}$ be any family of 
$C^\ast$-algebras which is indexed by the lattice $\bbZ^n$ and let $(H,\pi_0)\in {}^\ast\mathbf{Rep}_{\AAA}$
be any faithful $\ast$-representation such that the extended $\overline{\AAA}\in \CastAQFT(\bbZ^n)$ 
in \eqref{eqn:Aoverline} satisfies Haag duality, for all cone-shaped subsets $U\in \Cone(\bbZ^n)$. 
Then the $C^\ast$-categories of localized superselection sectors
carry via Theorem \ref{theo:PFAinMonCat} the structure of a unital associative algebra
\begin{flalign}
\SSS_{(\AAA,\pi_0)}\,\in\,\Alg_{\mathsf{uAs}}\Big(C^\ast\mathbf{PFA}_{\Cone(\bbZ^n)^\perp}\Big)
\end{flalign}
in the category of $C^\ast$-categorical prefactorization algebras 
$C^\ast\mathbf{PFA}_{\Cone(\bbZ^n)^\perp}$ from Definition \ref{def:PFA}, 
whose underlying prefactorization algebra is locally constant.
\end{cor}

\begin{ex}\label{ex:latticePFA}
A concrete example satisfying the hypotheses of Corollary \ref{cor:latticePFA}
is given by Kitaev's quantum double model \cite{Kitaev} on the $2$-dimensional lattice $\bbZ^2$
for a finite group $G$. In this example one considers the matrix algebra
$\AAA(s)=\mathrm{Mat}_{\vert G\vert}(\bbC)\in\CastAlg$, 
for all $s\in\bbZ^2$, and $(H,\pi_0)\in {}^\ast\mathbf{Rep}_{\AAA}$
the GNS representation of the (unique) translation invariant ground state of the Kitaev model Hamiltonian.
Haag duality for this model, which is crucial for Corollary \ref{cor:latticePFA}, 
has been proven in \cite{FiedlerNaaijkens} for Abelian $G$ and more recently in \cite{OgataHD}
for general $G$. Furthermore, an explicit description of the superselection
sectors of Kitaev's quantum double model is presented in \cite{Naaijkens1,FiedlerNaaijkens} 
for Abelian $G$ and in \cite{NaaijkensChapter,Bols} for general finite groups $G$.
\end{ex}

\subsection{\label{subsec:infty}An $\infty$-categorical analysis of the algebraic structures}
In this subsection we explain how from our Corollary \ref{cor:latticePFA} one can recover the results
from \cite{Naaijkens1,Naaijkens2,NaaijkensChapter,FiedlerNaaijkens,Ogata}
that superselection sectors of lattice $C^\ast$-AQFTs on $\bbZ^2$ form 
a braided monoidal $C^\ast$-category.
Our approach will also provide a generalization of these results to lattices $\bbZ^n$
of arbitrary dimension $n\in\bbZ_{\geq 1}$ and identify new algebraic structures 
which do not seem to have been observed in the literature before, 
even for the $2$-dimensional lattice $\bbZ^2$.
\sk

In order to state and prove our results, it will be convenient to use some aspects of
Lurie's theory of $\infty$-categories and $\infty$-operads \cite{LurieHA}.
This is undeniably a vast and highly technical subject, but fortunately the reader does not necessarily 
require a background on these techniques in order to follow our arguments.
Despite the fact that the target $\CastCat$ in our present scenario 
is only a $2$-categorical object, we prefer to give more general 
$\infty$-categorical arguments in order to accommodate the possibility
of higher-categorical lattice models arising in future works, 
i.e.\ models whose representation categories are higher categories or even dg-categories. 
The key facts about $\infty$-categories and 
$\infty$-operads which are used below are as follows:
\begin{itemize}
\item[(F1)] The category $\CastCat$ of $C^\ast$-categories and $\ast$-functors carries the structure
of a combinatorial simplicial symmetric monoidal model
category, see \cite{CastCat,Bunke} and also Theorem~\ref{theo:CastCatmodel} for a review.
Hence, it defines a presentably symmetric monoidal $\infty$-category $\CCastCat_\bbC$, see e.g.\ \cite{NikolausSagave}.
Since our results below do not depend on specific details of
this $\infty$-category, we will often replace
$\CCastCat_\bbC$ by an arbitrary presentably 
symmetric monoidal $\infty$-category $\V$.

\item[(F2)] The prefactorization operads $\P_{\CC^\perp}$ from Definition \ref{def:Poperad}
can be regarded, according to \cite[Example 2.1.1.21]{LurieHA}, as $\infty$-operads, 
which we denote with abuse of notation by the same symbols.
We denote by $\AAlg_{\P_{\CC^\perp}}(\V)$ the $\infty$-category of
$\V$-valued algebras over the $\infty$-operad $\P_{\CC^\perp}$, 
see e.g.\ \cite[Definition 2.1.2.7]{LurieHA}. 
In the case where $\V = \CCastCat_\bbC$, or more generally
when $\V$ is presented by a symmetric monoidal model category, 
every strict prefactorization algebra as in Definition \ref{def:PFA} defines
an object in this $\infty$-category, see e.g.\ \cite{Haugseng}.

\item[(F3)] Given any orthogonal category $\CC^\perp$ and any 
subset $W\subseteq \mathrm{Mor}(\CC)$ of morphisms, one can consider the full $\infty$-subcategory
\begin{flalign}
\AAlg_{\P_{\CC^\perp}}^{W\text{-}\mathrm{l.c.}}(\V)\,\subseteq\,\AAlg_{\P_{\CC^\perp}}(\V)
\end{flalign}
spanned, in the sense of \cite[Section 1.2.11]{LurieHTT}, by the 
 $\P_{\CC^\perp}$-algebras $\FFF : \P_{\CC^\perp}\to \V$
which send every $1$-ary operation $(f:U\to V)\in W$ to an equivalence in $\V$. 
By definition, our prefactorization algebra from Proposition~\ref{prop:PFA} yields an object
\begin{flalign}
\SSS_{(\AAA,\pi_0)}\,\in\, \AAlg_{\P_{\CC(S)^\perp}}^{\,\mathrm{l.c.}}\big(\CCastCat_\bbC\big)~,
\end{flalign}
where we use the abbreviated superscript 
$\mathrm{l.c.}:=\mathrm{Mor}(\CC(S))\text{-}\mathrm{l.c.}$ (locally constant) 
in the case where $W=\mathrm{Mor}(\CC(S))$ is the set of all morphisms.

\item[(F4)] Every orthogonal functor $F : \CC^\perp\to \DD^{\perp}$,
i.e.\ $F(f_1)\perp_\DD F(f_2)$ for all $f_1\perp_\CC f_2$, defines
a morphism $\P_{\CC^\perp}\to \P_{\DD^\perp}$ of operads, see e.g.\ \cite{BPSWcategorified},
and hence of $\infty$-operads. By precomposition, the latter 
induces a pullback $\infty$-functor
\begin{flalign}
F^\ast\,: \, \AAlg_{\P_{\DD^\perp}}(\V)~\xrightarrow{\;\quad\;}~\AAlg_{\P_{\CC^\perp}}(\V)~.
\end{flalign} 
Given further any subsets $W_\CC\subseteq \mathrm{Mor}(\CC)$ and $W_\DD\subseteq \mathrm{Mor}(\DD)$ 
such that $F(W_\CC)\subseteq W_\DD$, then this pullback $\infty$-functor restricts
to an $\infty$-functor 
\begin{flalign}
F^\ast\,: \, \AAlg_{\P_{\DD^\perp}}^{W_\DD\text{-}\mathrm{l.c.}}(\V)
~\xrightarrow{\;\quad\;}~\AAlg_{\P_{\CC^\perp}}^{W_\CC\text{-}\mathrm{l.c.}}(\V)
\end{flalign}
between the corresponding full $\infty$-subcategories of locally constant objects.

\item[(F5)] Regarding the unital associative operad $\mathsf{uAs}$ as an $\infty$-operad,
the result in \cite[Theorem 5.4.5.9]{LurieHA} identifies its $\infty$-category of algebras 
\begin{flalign}
\AAlg_{\mathsf{uAs}}(\V)\,\simeq\, \AAlg_{\P_{\Disk(\bbR^1)^\perp}}^{\,\mathrm{l.c.}}(\V)
\end{flalign}
with locally constant prefactorization algebras over 
the orthogonal category $\Disk(\bbR^1)^\perp$
of all open disks in $\bbR^1$ (i.e.\ open intervals $I\subseteq \bbR^1$)
and orthogonality relation given by disjointness of disks. Consequently, our strict 
unital associative algebra in prefactorization algebras from Theorem \ref{theo:PFAinMonCat} 
defines an object
\begin{flalign}
\SSS_{(\AAA,\pi_0)}\,\in\, \AAlg_{\P_{\Disk(\bbR^1)^\perp}}^{\,\mathrm{l.c.}}\Big(
\AAlg_{\P_{\CC(S)^\perp}}^{\,\mathrm{l.c.}}\big(\CCastCat_\bbC\big)\Big)~.
\end{flalign}
\end{itemize}

As a first step towards identifying the algebraic structure underlying the prefactorization 
algebra from Corollary \ref{cor:latticePFA}, we apply the construction in (F4)
to the diagram of orthogonal functors
\begin{equation}\label{eqn:latticetosphereorthogonal}
\begin{tikzcd}
\Cone(\bbZ^n)^\perp \ar[rr, hookleftarrow] && \Cone_0^{}(\bbZ^n)^\perp \ar[rr,"(-)\cap\, \bbS^{n-1}"] && \Disk(\bbS^{n-1})^\perp
\end{tikzcd}~.
\end{equation}
Here $\Cone_0^{}(\bbZ^n)^\perp \subseteq \Cone(\bbZ^n)^\perp$ denotes the full orthogonal subcategory
of cone-shaped subsets $U = C_{(0,t,\alpha)}\cap \bbZ^n\subseteq \bbZ^n$ which
are represented by open cones in $\bbR^n$ whose apex is the origin $0\in\bbR^n$.
(Note that the representing open cone in $\bbR^n$ is unique in this case,
i.e.\ $C_{(0,t,\alpha)}\cap \bbZ^n = C_{(0,t^\prime,\alpha^\prime)}\cap \bbZ^n$
if and only if $t^\prime = t$ and $\alpha^\prime=\alpha$.) We further denote
by $\Disk(\bbS^{n-1})^\perp$ the orthogonal category of open disks in the ${(n{-}1)}$-dimensional
sphere $\bbS^{n-1}$, i.e.\ objects are all open subsets $U\subseteq \bbS^{n-1}$ which are diffeomorphic
$U\cong \bbR^{n-1}$ to the ${(n{-}1)}$-dimensional Euclidean space, morphisms are subset inclusions
and the orthogonality relation is given by disjointness of disks.
The orthogonal functor $(-)\cap\, \bbS^{n-1} : \Cone_0^{}(\bbZ^n)^\perp\to  \Disk(\bbS^{n-1})^\perp$
sends each object $U = C_{(0,t,\alpha)}\cap \bbZ^n\in \Cone_0^{}(\bbZ^n)^\perp$
to the intersection $C_{(0,t,\alpha)}\cap \bbS^{n-1}\in \Disk(\bbS^{n-1})^\perp$
of its (unique) representing open cone in $\bbR^n$ with the unit sphere $\bbS^{n-1}\subseteq\bbR^n$.
\begin{propo}\label{prop:latticetospherePFA}
The orthogonal functors in \eqref{eqn:latticetosphereorthogonal} induce via
pullback as in (F4) a diagram of $\infty$-functors
\begin{equation}\label{eqn:latticetospherePFA1}
\begin{tikzcd}
\AAlg_{\P_{\Cone(\bbZ^n)^\perp}}^{\,\mathrm{l.c.}}(\V) \ar[r]&
\AAlg_{\P_{\Cone_0^{}(\bbZ^n)^\perp}}^{\,\mathrm{l.c.}}(\V) &\ar[l,"\sim"'] 
\AAlg_{\P_{\Disk(\bbS^{n-1})^\perp}}^{\,\mathrm{l.c.}}(\V)
\end{tikzcd}
\end{equation}
between the corresponding $\infty$-categories of locally constant prefactorization algebras.
The left-pointing $\infty$-functor is an equivalence of $\infty$-categories,
hence one obtains an $\infty$-functor
\begin{equation}\label{eqn:latticetospherePFA2}
\begin{tikzcd}
\AAlg_{\P_{\Cone(\bbZ^n)^\perp}}^{\,\mathrm{l.c.}}(\V) \ar[r] &
\AAlg_{\P_{\Disk(\bbS^{n-1})^\perp}}^{\,\mathrm{l.c.}}(\V)
\end{tikzcd}
\end{equation}
from locally constant prefactorization algebras over $\Cone(\bbZ^n)^\perp$
to locally constant prefactorization algebras over $\Disk(\bbS^{n-1})^\perp$.
\end{propo}
\begin{proof}
The image of the orthogonal functor  
$(-)\cap \bbS^{n-1} : \Cone_0^{}(\bbZ^n)^\perp\to  \Disk(\bbS^{n-1})^\perp$
is given by the full orthogonal subcategory 
$\Disk_{\mathrm{eucl}}(\bbS^{n-1})^\perp\subseteq \Disk(\bbS^{n-1})^\perp$ 
of all non-dense Euclidean open disks in $\bbS^{n-1}$, 
i.e.\ open subsets of the form $U_{(t,\alpha)} = 
\big\{x\in\bbS^{n-1}\,:\, \mathrm{angle}(x,t)<\alpha\big\}\subseteq \bbS^{n-1}$
for some $t\in\bbS^{n-1}$ and $\alpha\in(0,\pi)$.
This induces a factorization of the orthogonal functor
\begin{subequations}
\begin{equation}
\begin{tikzcd}
(-)\cap\, \bbS^{n-1} \,:\, \Cone_0^{}(\bbZ^n)^\perp \ar[r,"\cong"] &
\Disk_{\mathrm{eucl}}(\bbS^{n-1})^\perp \ar[r, hookrightarrow] &
\Disk(\bbS^{n-1})^\perp
\end{tikzcd}
\end{equation}
into an isomorphism of orthogonal categories followed by a full orthogonal subcategory inclusion.
Passing to the $\infty$-categories of locally constant prefactorization algebras via pullback,
we obtain
\begin{equation}\label{eqn:EuclDiskchain}
\begin{tikzcd}
\AAlg_{\P_{\Disk(\bbS^{n-1})^\perp}}^{\,\mathrm{l.c.}}(\V) \ar[r] &
\AAlg_{\P_{\Disk_{\mathrm{eucl}}(\bbS^{n-1})^\perp}}^{\,\mathrm{l.c.}}(\V)  \ar[r,"\sim"] &
\AAlg_{\P_{\Cone_0^{}(\bbZ^n)^\perp}}^{\,\mathrm{l.c.}}(\V)
\end{tikzcd}~,
\end{equation}
\end{subequations}
where the second arrow is an equivalence of $\infty$-categories because the underlying orthogonal functor
is an isomorphism of orthogonal categories. To show that the first arrow is an equivalence of 
$\infty$-categories, we use \cite[Definition 5.4.5.1 and Theorem 5.4.5.9]{LurieHA} 
to obtain the commutative diagram
\begin{equation}\label{eqn:comdiagramEsphere}
\begin{tikzcd}
\AAlg_{\P_{\Disk(\bbS^{n-1})^\perp}}^{\,\mathrm{l.c.}}(\V) \ar[rr] && 
\AAlg_{\P_{\Disk_{\mathrm{eucl}}(\bbS^{n-1})^\perp}}^{\,\mathrm{l.c.}}(\V)\\
& \ar[lu,"\sim"' sloped] \AAlg_{\mathbb{E}_{\bbS^{n-1}}}(\V)\ar[ru,"\sim"' sloped]&
\end{tikzcd}~.
\end{equation}
It is shown in \cite[Theorem 5.4.5.9]{LurieHA} that the 
upward-left pointing arrow is an equivalence. 
To prove that the upward-right pointing arrow is an equivalence too, 
one mimicks the proof of \cite[Theorem 5.4.5.9]{LurieHA} while taking into account the following basic 
geometric facts about disks in $\Disk_{\mathrm{eucl}}(\bbS^{n-1})$: 
\begin{enumerate}[1.)]
\item For every point $x\in\bbS^{n-1}$, the poset 
\begin{flalign}
\Big\{U\in \Disk_{\mathrm{eucl}}(\bbS^{n-1})\,:\,x\in U\Big\}
\end{flalign}
of Euclidean open disks containing $x$ is cofiltered. 
\item For every $U\in \Disk_{\mathrm{eucl}}(\bbS^{n-1})$
and every finite family $x_1,\dots,x_m\in U$ of distinct points, i.e.\ $x_i\neq x_j$ for $i\neq j$, the poset 
\begin{flalign}
\Big\{\und{V}\in \Disk_{\mathrm{eucl}}(\bbS^{n-1})^{\times m}\,:\,(\und{V}\to U)\in\P_{\Disk_{\mathrm{eucl}}(\bbS^{n-1})^\perp}\text{ and } x_i\,\in\,V_i~\forall i \Big\}
\end{flalign}
of tuples of mutually disjoint Euclidean open disks in $U$ containing the $x_i$ is cofiltered.
\end{enumerate}
It then follows that the horizontal arrow in \eqref{eqn:comdiagramEsphere}, and hence
the first arrow in \eqref{eqn:EuclDiskchain}, is an equivalence too.
The $\infty$-functor in \eqref{eqn:latticetospherePFA2} is obtained by choosing any quasi-inverse
for the equivalence of $\infty$-categories given by the left-pointing $\infty$-functor 
in \eqref{eqn:latticetospherePFA1}.
\end{proof}

\begin{rem}
In simpler words, the result of Proposition \ref{prop:latticetospherePFA}
states that every locally constant prefactorization algebra over
$\Cone(\bbZ^n)^\perp$ has an underlying locally constant prefactorization algebra
over $\Disk(\bbS^{n-1})^\perp$ which is obtained by applying the 
$\infty$-functor \eqref{eqn:latticetospherePFA2}. In particular,
forgetting for the moment the unital associative algebra structure
in Corollary \ref{cor:latticePFA}, we can assign to our
locally constant prefactorization algebra
of localized superselection sectors $\SSS_{(\AAA_,\pi_0)} \in 
\AAlg_{\P_{\Cone(\bbZ^n)^\perp}}^{\,\mathrm{l.c.}}\big(\CCastCat_\bbC\big) $
over $\Cone(\bbZ^n)^\perp$ a locally constant prefactorization algebra over $\Disk(\bbS^{n-1})^\perp$.
Note that this passage from cone-shaped subsets in $\bbZ^n$ to open disks in $\bbS^{n-1}$ might be
forgetful because the right-pointing $\infty$-functor in \eqref{eqn:latticetospherePFA1}, and hence
also the composite $\infty$-functor in \eqref{eqn:latticetospherePFA2}, are a priori not equivalences.
Answering the question of whether or not this construction is forgetful is conceptually 
interesting, but not needed here to 
link our result in Corollary \ref{cor:latticePFA} to the earlier works in
\cite{Naaijkens1,Naaijkens2,NaaijkensChapter,FiedlerNaaijkens,Ogata}. 
\end{rem}

The second step towards identifying the algebraic structure underlying the 
prefactorization algebra from Corollary \ref{cor:latticePFA} is concerned
with the additional unital associative algebra structure. For this we use
item (F5) from above and results about additivity of locally constant
prefactorization algebras.
\begin{propo}\label{prop:latticetospherePFAuAs}
The $\infty$-functor in \eqref{eqn:latticetospherePFA2} induces an $\infty$-functor
\begin{flalign}\label{eqn:latticetospherePFAuAs}
\AAlg_{\P_{\Disk(\bbR^1)^\perp}}^{\,\mathrm{l.c.}}\big(
\AAlg_{\P_{\Cone(\bbZ^n)^\perp}}^{\,\mathrm{l.c.}}(\V)\big)  ~\xrightarrow{\;\quad\;}~
\AAlg_{\P_{\Disk(\bbR^1)^\perp}}^{\,\mathrm{l.c.}}\big(
\AAlg_{\P_{\Disk(\bbS^{n-1})^\perp}}^{\,\mathrm{l.c.}}(\V)\big)\,\simeq\, 
\AAlg_{\P_{\Disk(\bbR^1\times \bbS^{n-1})^\perp}}^{\,\mathrm{l.c.}}(\V)
\end{flalign}
to the $\infty$-category of locally constant prefactorization algebras over 
$\Disk(\bbR^1\times \bbS^{n-1})^\perp$, i.e.\ the orthogonal category
of open disks in the cylinder $\bbR^1\times\bbS^{n-1}$.
\end{propo}
\begin{proof}
Both pullback $\infty$-functors in \eqref{eqn:latticetospherePFA1} are symmetric monoidal
$\infty$-functors, hence so is \eqref{eqn:latticetospherePFA2}. 
This implies that the first arrow in \eqref{eqn:latticetospherePFAuAs} is well-defined.
The equivalence in the second step of \eqref{eqn:latticetospherePFAuAs} 
follows from additivity of locally constant prefactorization algebras in the form
of \cite[Example 5.4.5.5]{LurieHA}. More explicitly, we have a chain of equivalences
of $\infty$-categories
\begin{flalign}
\nn \AAlg_{\P_{\Disk(\bbR^1)^\perp}}^{\,\mathrm{l.c.}}\big(
\AAlg_{\P_{\Disk(\bbS^{n-1})^\perp}}^{\,\mathrm{l.c.}}(\V)\big)
\,&\simeq\, \AAlg_{\mathbb{E}_{\bbR^1}}\big(\AAlg_{\mathbb{E}_{\bbS^{n-1}}}(\V)\big)\\
\,&\simeq\,\AAlg_{\mathbb{E}_{\bbR^1\times \bbS^{n-1}}}(\V)
\,\simeq\,
\AAlg_{\P_{\Disk(\bbR^1\times \bbS^{n-1})^\perp}}^{\,\mathrm{l.c.}}(\V)~,
\end{flalign}
where the first and last equivalence are from \cite[Theorem 5.4.5.9]{LurieHA}
and the middle equivalence is from \cite[Example 5.4.5.5]{LurieHA}.
\end{proof}

Applying this construction to Corollary \ref{cor:latticePFA}, we obtain the following result.
\begin{cor}\label{cor:latticePFAcylinder}
Let $\AAA = \{\AAA(x)\}_{x\in \bbZ^n}\in\CastAlg^{\bbZ^n}$ be any family of $C^\ast$-algebras
which is indexed by the lattice $\bbZ^n$ and let $(H,\pi_0)\in {}^\ast\mathbf{Rep}_{\AAA}$
be any faithful $\ast$-representation such that the extended $\overline{\AAA}\in \CastAQFT(\bbZ^n)$ 
in \eqref{eqn:Aoverline} satisfies Haag duality, for all $U\in \Cone(\bbZ^n)$. 
Then the locally constant prefactorization algebra of localized superselection sectors
\begin{flalign}
\SSS_{(\AAA,\pi_0)}\,\in\, \AAlg_{\P_{\Disk(\bbR^1)^\perp}}^{\,\mathrm{l.c.}}\big(
\AAlg_{\P_{\Cone(\bbZ^n)^\perp}}^{\,\mathrm{l.c.}}\big(\CCastCat_\bbC\big)\big) 
\end{flalign}
from Corollary \ref{cor:latticePFA} has an underlying locally constant prefactorization algebra
\begin{flalign}
\und{\SSS}_{(\AAA,\pi_0)}\,\in\, 
\AAlg_{\P_{\Disk(\bbR^1\times \bbS^{n-1})^\perp}}^{\,\mathrm{l.c.}}\big(\CCastCat_\bbC\big)
\end{flalign}
over the orthogonal category $\Disk(\bbR^1\times \bbS^{n-1})^\perp$
of open disks in the cylinder $\bbR^1\times\bbS^{n-1}$.
\end{cor}

To relate this result to \cite{Naaijkens1,Naaijkens2,NaaijkensChapter,FiedlerNaaijkens,Ogata},
we observe that there exists an open embedding $\bbR^n\to \bbR^1\times\bbS^{n-1}$ which is determined
by removing one point of the sphere $\bbR^1\times (\bbS^{n-1}\setminus\pt)\cong \bbR^n$. 
This induces an orthogonal functor
\begin{flalign}
\Disk(\bbR^n)^\perp ~\xrightarrow{\;\quad\;}~ \Disk(\bbR^1\times \bbS^{n-1})^\perp~,
\end{flalign}
and hence via (F4) a pullback $\infty$-functor between the corresponding
$\infty$-categories of locally constant prefactorization algebras. Using
further the identification between locally constant $\P_{\Disk(\bbR^n)^\perp}$-algebras and 
$\mathbb{E}_n$-algebras from \cite[Theorem 5.4.5.9]{LurieHA}, we obtain
an $\infty$-functor
\begin{flalign}\label{eqn:forgetEn}
\AAlg_{\P_{\Disk(\bbR^1\times \bbS^{n-1})^\perp}}^{\,\mathrm{l.c.}}(\V)
~\xrightarrow{\;\quad\;}~
\AAlg_{\P_{\Disk(\bbR^n)^\perp}}^{\,\mathrm{l.c.}}(\V)\,\simeq\,\AAlg_{\mathbb{E}_n}(\V)
\end{flalign}
to the $\infty$-category of $\mathbb{E}_n$-algebras in $\V$.
\begin{cor}\label{cor:latticeEn}
The locally constant prefactorization algebra 
from Corollary \ref{cor:latticePFAcylinder} has 
an underlying $\mathbb{E}_n$-monoidal $C^\ast$-category
\begin{flalign}
\und{\und{\SSS}}_{(\AAA,\pi_0)}\,\in\,
\AAlg_{\mathbb{E}_n}\big(\CCastCat_\bbC\big)
\end{flalign}
which is obtained by applying the $\infty$-functor \eqref{eqn:forgetEn}.
More concretely, this yields
a monoidal $C^\ast$-category for the $1$-dimensional lattice $\bbZ^1$,
a braided monoidal $C^\ast$-category for the $2$-dimensional lattice $\bbZ^2$,
and a symmetric monoidal $C^\ast$-category for the $(n\geq 3)$-dimensional lattice $\bbZ^{n}$.
\end{cor}

\begin{rem}
We would like to emphasize that the emergence of an underlying 
braided or symmetric monoidal $C^\ast$-category is \textit{not} a general 
consequence of our geometric Assumption~\ref{assu:0}, but 
it is linked to more specific features of the category $\Cone(\bbZ^n)$
of cone-shaped subsets. We expect that the geometric axioms (GA0--GA3) of~\cite{SSS},
which are stronger than our Assumption~\ref{assu:0} and also satisfied 
for cones (see \cite[Appendix A]{SSS}), could be the key to answer the following
interesting question: Under which additional conditions on the orthogonal category
of localization regions $\CC(S)^\perp$ has the locally constant prefactorization algebra
$\SSS_{(\AAA,\pi_0)}$ of superselection sectors an underlying braided or symmetric monoidal category?
\end{rem}

\begin{rem}
In applications to topological order, one of the motivations behind our work,
the category of superselection sectors typically has more structure, namely
that of a modular tensor category~\cite[Appendix E]{Kitaev06}. This richer
structure is recognized \textit{a posteriori} in concrete models. Under an
additional assumption, namely that the von Neumann algebras $\pi_0(\AAA(V))^{\prime\prime}$
for $V \in \Cone(\mathbb{Z}^d)$ are properly infinite factors, it can be shown
that $\und{\und{\SSS}}_{(\AAA,\pi_0)}$ has direct sums and subobjects (i.e., we can project
onto invariant subspaces). If one further restricts to dualizable objects,
which physically have the interpretation of conjugate charges, this implies
semi-simplicity~\cite{LongoRoberts}. The category is then modular if there are
only finitely many simple objects, and in addition the braiding is
non-degenerate. As far as we are aware, for the lattice systems we are
interested in, there are no known general criteria which imply dualizability,
finiteness of sectors, or non-degeneracy of the braiding, but these properties
can be verified in many concrete models. (Note however that in algebraic
quantum field theory~\cite{GuidoLongo} or for conformal nets~\cite{KLM}, this
question is better understood.) Since these properties do not have their origin
in the prefactorization algebra structure, we refer the interested reader to
the references cited above for details.
\end{rem}

It is important to stress that the $\infty$-functor \eqref{eqn:forgetEn} is forgetful,
which implies that the underlying $\mathbb{E}_n$-monoidal $C^\ast$-category from Corollary 
\ref{cor:latticeEn} does not capture the entire algebraic structure 
of the locally constant prefactorization algebra over $\Disk(\bbR^1\times \bbS^{n-1})^\perp$
from Corollary \ref{cor:latticePFAcylinder}. The additional algebraic structures are related
to the homotopy groups of the sphere $\bbS^{n-1}$, see e.g.\ \cite[Remark 27]{Ginot}, 
and they are in general difficult to determine concretely.
In the case of the $2$-dimensional lattice $\bbZ^2$, which is most relevant
for applications to topological order \cite{Naaijkens1,Naaijkens2,NaaijkensChapter,FiedlerNaaijkens,Ogata},
one can explicitly characterize these additional algebraic structures 
by using \cite[Corollary 4]{Ginot}.\footnote{Note that Ginot's 
result is stated only in the context of cochain complexes, 
but it is valid for factorization algebras valued in any 
presentably symmetric monoidal $\infty$-category. Indeed, the proof 
relies on extension constructions from a factorizing basis,
and such tools were developed in full generality in 
\cite[Proposition 4.18]{Karlsson-Scheimbauer-Walde}.}
\begin{cor}\label{cor:selfequivalence}
For the $2$-dimensional lattice $\bbZ^2$, the datum of the 
locally constant prefactorization algebra 
$\und{\SSS}_{(\AAA,\pi_0)}\in
\AAlg_{\P_{\Disk(\bbR^1\times \bbS^{1})^\perp}}^{\,\mathrm{l.c.}}\big(\CCastCat_\bbC\big)$
from Corollary \ref{cor:latticePFAcylinder} is equivalent 
to its underlying braided monoidal $C^\ast$-category from Corollary \ref{cor:latticeEn}
together with a self-equivalence
\begin{equation}
\begin{tikzcd}
T \,:\, \und{\und{\SSS}}_{(\AAA,\pi_0)}\ar[r,"\sim"] & \und{\und{\SSS}}_{(\AAA,\pi_0)}
\end{tikzcd}
\end{equation} 
of braided monoidal $C^\ast$-categories. 
\end{cor}

\begin{rem}
To the best of our knowledge, the additional structure of a self-equivalence
on the braided monoidal $C^\ast$-categories of superselection sectors in cone-shaped
subsets of $\bbZ^2$ has not yet been observed in the literature. We will explain in the
next subsection how to compute it in terms of 
a kind of holonomy which arises by rotating cones around their apex. 
\sk

It is important to stress that this self-equivalence is of a 
topological origin, rooted in the fact that our prefactorization algebra
is intrinsically defined on the category of cone-shaped subsets of $\bbZ^2$, 
hence it can not be detected when one `punctures the circle' or `introduces a forbidden direction',
as usually done in the literature \cite{BuchholzFredenhagen,Naaijkens1,Ogata},
with a notable exception given by the recent work \cite{SSS}. 
In this restricted context, one has only access to the $\mathbb{E}_2$-monoidal 
$C^\ast$-category $\und{\und{\SSS}}_{(\AAA,\pi_0)}\in \AAlg_{\mathbb{E}_2}\big(\CCastCat_\bbC\big)$ 
from Corollary \ref{cor:latticeEn}, but one loses the datum of the
self-equivalence $T$ from Corollary \ref{cor:selfequivalence}. 
Once this datum is lost, it clearly cannot be restored through universal constructions or `tricks',
such as taking the factorization homology $\int_{\bbR^1\times \bbS^1}\und{\und{\SSS}}_{(\AAA,\pi_0)}$ 
on the cylinder (endowed with a choice of framing) 
of the underlying $\mathbb{E}_2$-monoidal $C^\ast$-category, see e.g.\ \cite{AyalaFrancis}
and also \cite{FHTopOrder} for the context of topological order. 
Indeed, the locally constant prefactorization algebra 
$\int_{\bbR^1\times \bbS^1}\und{\und{\SSS}}_{(\AAA,\pi_0)}\in 
\AAlg_{\P_{\Disk(\bbR^1\times \bbS^{n-1})^\perp}}^{\,\mathrm{l.c.}}\big(\CCastCat_\bbC\big)$ 
obtained from factorization homology has necessarily a trivial self-equivalence, 
hence it will in general not be equivalent to our intrinsically 
constructed one $\und{\SSS}_{(\AAA,\pi_0)}\in
\AAlg_{\P_{\Disk(\bbR^1\times \bbS^{n-1})^\perp}}^{\,\mathrm{l.c.}}\big(\CCastCat_\bbC\big)$
from Corollary \ref{cor:latticePFAcylinder}.
\end{rem}

\subsection{\label{subsec:explicit}Concrete computations for $\bbZ^2$}
In this last subsection we will specialize to the case of the $2$-dimensional 
lattice $\bbZ^2$ and explain how the braiding and self-equivalence from 
Corollary \ref{cor:selfequivalence} can be computed geometrically by using 
suitable cone configurations. Since the prefactorization algebra
\begin{flalign}\label{eqn:PFAfinalsubsection}
\SSS_{(\AAA,\pi_0)}\, \in \,\Alg_{\mathsf{uAs}}\Big(C^\ast\mathbf{PFA}_{\Cone(\bbZ^2)^\perp}\Big)
\end{flalign}
from Corollary \ref{cor:latticePFA} is locally constant, one can realize 
the braided monoidal structure and self-equivalence from Corollary \ref{cor:selfequivalence}
on any of the unitarily equivalent $C^\ast$-categories $\SSS_{(\AAA,\pi_0)}(V)$. 
(Note that this is completely analogous to the observations in \cite{SSS}.)
In order to obtain a model which is manifestly independent of the choice of 
cone-shaped subset $V\in\Cone(\bbZ^2)$, we can also realize this structure
on the unitarily equivalent $C^\ast$-category 
\begin{flalign}\label{eqn:undundSSS}
\und{\und{\SSS}}_{(\AAA,\pi_0)}\,:=\,\SSS_{(\AAA,\Cone(\bbZ^2),\pi_0)}\,\in\,\CastCat
\end{flalign}
of localizable (but not necessarily strictly localized) superselection sectors
from Definition \ref{def:SSSloc}.
\sk

In order to transfer the algebraic structures from the locally constant prefactorization algebra
$\SSS_{(\AAA,\pi_0)}\colon \P_{\Disk(\bbR^1\times \bbS^{1})^\perp}\to \CCastCat_\bbC$ to the single $C^\ast$-category \eqref{eqn:undundSSS}, we choose
quasi-inverses for the unitary equivalences induced by the inclusion $\ast$-functors
$\iota_V^{} : \SSS_{(\AAA,\pi_0)}(V)\to \und{\und{\SSS}}_{(\AAA,\pi_0)}$ and 
$\iota_U^V :\SSS_{(\AAA,\pi_0)}(U)\to \SSS_{(\AAA,\pi_0)}(V)$,
for all $V\in\Cone_0(\bbZ^2)$ and all morphisms $U\subseteq V$ in $\Cone_0(\bbZ^2)$.
Up to equivalence, the result of this transfer does not depend on these choices.
Convenient models for such quasi-inverses can be obtained by choosing
unitary charge transporters as in Remark \ref{rem:transporter}.
Explicitly, to define a quasi-inverse for 
$\iota_V^{} : \SSS_{(\AAA,\pi_0)}(V)\to \und{\und{\SSS}}_{(\AAA,\pi_0)}$,
choose for each object $(K,\pi)\in\und{\und{\SSS}}_{(\AAA,\pi_0)}$ a unitary
$u_{(K,\pi)} : K\to H$ such that  $\big(H,\pi^{u_{(K,\pi)}}\big) \in \SSS_{(\AAA,\pi_0)}(V)$
defined in \eqref{eqn:piuformula} is strictly localized in $V$. When 
$(H,\pi)\in \SSS_{(\AAA,\pi_0)}(V)\subseteq \und{\und{\SSS}}_{(\AAA,\pi_0)}$ is already strictly localized in
$V$, we choose the identity $u_{(H,\pi)}=\id_H : H\to H$. This defines
a quasi-inverse $\ast$-functor
\begin{flalign}
p_V\,:\,\und{\und{\SSS}}_{(\AAA,\pi_0)} ~&\xrightarrow{\;\quad\;}~\SSS_{(\AAA,\pi_0)}(V)~,\\
\nn (K,\pi)~&\xmapsto{\;\quad\;}~\big(H,\pi^{u_{(K,\pi)}}\big)~,\\
\nn \Big(L : (K,\pi)\to(K^\prime,\pi^\prime)\Big) ~&\xmapsto{\;\quad\;}~
\Big(u_{(K^\prime,\pi^\prime)}\circ L\circ u_{(K,\pi)}^\ast :\big(H,\pi^{u_{(K,\pi)}}\big)\to
\big(H,\pi^{\prime u_{(K^\prime,\pi^\prime)}}\big) \Big)
\end{flalign}
satisfying $p_V\,\iota_V =\id$ and such that $\alpha_V : \id \stackrel{\cong}{\Longrightarrow} \iota_V\,p_V$
via the natural isomorphism given by the unitary components $(\alpha_V)_{(K,\pi)} = u_{(K,\pi)} : (K,\pi)
\to \big(H,\pi^{u_{(K,\pi)}}\big)$. With a similar construction one obtains quasi-inverses
\begin{flalign}
p_U^V \,:\, \SSS_{(\AAA,\pi_0)}(V)~\xrightarrow{\;\quad\;}~ \SSS_{(\AAA,\pi_0)}(U)
\end{flalign} 
for the inclusion $\ast$-functors
$\iota_U^V :\SSS_{(\AAA,\pi_0)}(U)\to \SSS_{(\AAA,\pi_0)}(V)$ satisfying $p_U^V\,\iota_U^V = \id$
and such that $\alpha_U^V : \id \stackrel{\cong}{\Longrightarrow} \iota_U^V\,p_U^V$ via a natural isomorphism
with unitary components.
\sk

The monoidal structure from Corollary \ref{cor:selfequivalence}
can be realized on the $C^\ast$-category \eqref{eqn:undundSSS} 
by taking as monoidal unit the reference $\ast$-representation 
$\pi_0\in\und{\und{\SSS}}_{(\AAA,\pi_0)}$ and as monoidal product the $\ast$-functor
\begin{equation}\label{eqn:starmonoidal}
\begin{tikzcd}
\ar[d,"p_V\boxtimes p_V"'] \und{\und{\SSS}}_{(\AAA,\pi_0)}\boxtimes \und{\und{\SSS}}_{(\AAA,\pi_0)}\ar[r,dashed,"\star"]& 
\und{\und{\SSS}}_{(\AAA,\pi_0)}  \\
\SSS_{(\AAA,\pi_0)}(V) \boxtimes \SSS_{(\AAA,\pi_0)}(V) \ar[r,"\diamond	\hspace*{-1.16mm}\cdot	\hspace*{0.6mm}"'] & 
\SSS_{(\AAA,\pi_0)}(V) \ar[u,"\iota_V"']
\end{tikzcd}
\end{equation}
which is obtained by transferring the strict monoidal product from 
Lemma \ref{lem:monoidal} along any of the unitary equivalences
$\iota_V : \SSS_{(\AAA,\pi_0)}(V) \rightleftarrows \und{\und{\SSS}}_{(\AAA,\pi_0)} : p_V$,
for some $V\in\Cone_0(\bbZ^2)$.
The existence of an associator and unitors for this (not necessarily strict) 
monoidal structure is a direct consequence of Corollary \ref{cor:selfequivalence}. 
\sk 

To exhibit a model for the braiding $\tau: \star\Rightarrow \star\circ\mathrm{flip}$ 
for this monoidal structure, we leverage the fact that
the prefactorization algebra structure and object-wise monoidal structures 
of \eqref{eqn:PFAfinalsubsection} are strictly compatible by Theorem \ref{theo:PFAinMonCat}. 
Using an Eckmann-Hilton-type argument, one observes that the diagram
\begin{equation}\label{eqn:EH}
\begin{tikzcd}
\ar[dr,"\boxtimes_i \iota_{U_i}^V"'] \bigboxtimes\limits_{i=1}^n \SSS_{(\AAA,\pi_0)}(U_i) \ar[rr,"\bullet"] && 
\SSS_{(\AAA,\pi_0)}(V)\\
& \bigboxtimes\limits_{i=1}^n \SSS_{(\AAA,\pi_0)}(V) \ar[ru,"\diamond	\hspace*{-1.16mm}\cdot	\hspace*{0.6mm}^n"'] &
\end{tikzcd}
\end{equation}
in $\CastCat$ commutes strictly, for all operations $(U_1,\dots,U_n)\to V$ in the operad
$\P_{\Cone(\bbZ^2)^\perp}$, where $\cdiamond^n$ denotes the $n$-fold
object-wise monoidal product from Lemma \ref{lem:monoidal}. Choosing any
binary operation $(U_1,U_2)\to V$ in $\P_{\Cone_0(\bbZ^2)^\perp}$, we obtain the
natural isomorphism
\begin{subequations}\label{eqn:braiding}
\begin{equation}
\begin{tikzcd}[column sep=tiny]
		\cdiamond 
		\ar[d,Rightarrow,"\diamond\hspace*{-1.16mm}\cdot\hspace*{0.6mm}\circ (\alpha_{U_1}^V\boxtimes \alpha_{U_2}^V)"']
		\ar[rrr,Rightarrow,dashed,"\tilde{\tau}"]&[-10mm]&&[-15mm] 
		\cdiamond\circ\, \mathrm{flip}\\[4mm]
		\cdiamond\circ\big(\iota_{U_1}^Vp_{U_1}^V\boxtimes \iota_{U_2}^{V}p_{U_2}^V\big) \ar[rd,equal]&&
		&
		\cdiamond\circ\big(\iota_{U_2}^Vp_{U_2}^V\boxtimes \iota_{U_1}^{V}p_{U_1}^V\big)\circ\mathrm{flip}
		\ar[u,Rightarrow,"\diamond	\hspace*{-1.16mm}\cdot	\hspace*{0.6mm} \circ (\alpha_{U_2}^V\boxtimes \alpha_{U_1}^V)^{-1}\circ\Id"']\\[4mm]
		& \bullet\circ \big(p_{U_1}^V\boxtimes p_{U_2}^V\big) \ar[r,equal]&
		\bullet\circ \mathrm{flip}\circ \big(p_{U_1}^V\boxtimes p_{U_2}^V\big) \ar[ru,equal] &
\end{tikzcd}~,
\end{equation}
where the diagonal equalities use the Eckmann-Hilton argument \eqref{eqn:EH}
and the lower horizontal equality uses the permutation equivariance property \eqref{eqn:PFAdiagramperm} 
of a prefactorization algebra. This natural isomorphism transfers to the braiding
\begin{flalign}
\tau\,:=\, \iota_V\circ \tilde{\tau}\circ (p_V\boxtimes p_V) \,:\, \star ~\Longrightarrow~\star\circ\,\mathrm{flip}
\end{flalign}
\end{subequations}
for the monoidal structure \eqref{eqn:starmonoidal}, which as a direct 
consequence of Corollary \ref{cor:selfequivalence} satisfies the relevant hexagon identities.
\begin{rem}
The braiding \eqref{eqn:braiding} agrees with the standard
one from traditional superselection theory, see e.g.\ \cite{DHR} and \cite{BuchholzFredenhagen}.
To verify this claim, it suffices to compute the components $\tau_{\pi,\dot{\pi}}$ of the braiding
for any two objects $\pi,\dot{\pi}\in \SSS_{(\AAA,\pi_0)}(V)\subseteq \und{\und{\SSS}}_{(\AAA,\pi_0)}$ which
we may assume to be strictly localized in the chosen cone-shaped subset $V\in\Cone_0(\bbZ^2)$.
Abbreviating by $u:= (\alpha_{U_1}^V)_{\pi} : \pi\to \pi^u$
and $\dot{u}:= (\alpha_{U_2}^V)_{\dot{\pi}} : \dot{\pi}\to \dot{\pi}^{\dot{u}}$
the corresponding components of the unitary natural isomorphisms $\alpha_{U_1}^V : 
\id \Rightarrow \iota_{U_1}^V\,p_{U_1}^V$ and $\alpha_{U_2}^V : \id \Rightarrow \iota_{U_2}^V\,p_{U_2}^V$,
we find that
\begin{subequations}
\begin{equation}
\begin{tikzcd}
\tau_{\pi,\dot{\pi}}\,:\,\pi\star \dot{\pi}\,=\,\pi\cdiamond\dot{\pi}
\ar[r,"u \diamond	\hspace*{-1.16mm}\cdot	\hspace*{0.6mm}\dot{u}"] & 
\pi^u\cdiamond\dot{\pi}^{\dot{u}}\,=\, \dot{\pi}^{\dot{u}}\cdiamond\pi^u
\ar[r,"\dot{u}^\ast\diamond	\hspace*{-1.16mm}\cdot	\hspace*{0.6mm} u^\ast"] &
\dot{\pi}\cdiamond\pi \,=\, \dot{\pi}\star \pi
\end{tikzcd}~.
\end{equation}
Recalling the explicit formula \eqref{eqn:monoidalformulamorphism} for the monoidal structure 
$\cdiamond$ on morphisms,
we can write this more concretely as the composition of unitary operators
\begin{flalign}
\tau_{\pi,\dot{\pi}}\,=\, \dot{\rho}_V(u^\ast)\circ \dot{u}^\ast\circ u\circ \rho_V(\dot{u})~,
\end{flalign}
\end{subequations}
which coincides with the usual braiding from \cite{DHR,BuchholzFredenhagen}.
Since, by our transfer construction, the resulting braided monoidal $C^\ast$-category is unique (up to equivalence),
it becomes superfluous to analyze as in \cite{DHR,BuchholzFredenhagen} the dependence of the braiding
on the choice of binary operation $(U_1,U_2)\to V$ in $\P_{\Cone_0(\bbZ^2)^\perp}$.
\end{rem}

It remains to exhibit a model for the self-equivalence 
$T : \und{\und{\SSS}}_{(\AAA,\pi_0)}\to \und{\und{\SSS}}_{(\AAA,\pi_0)}$ 
from Corollary \ref{cor:selfequivalence}. 
This is geometrically realized by winding once around a cone's apex 
via a zig-zag of cone inclusions. More explicitly, we choose any sequence of morphisms
\begin{subequations}\label{eqn:Tcones}
\begin{flalign}
V_1 \,\xrightarrow{\;\quad\;}\, V_2\, \xleftarrow{\;\quad\;} \,V_3\, \xrightarrow{\;\quad\;}\, V_4 \,\xleftarrow{\;\quad\;}\, V_1
\end{flalign}
in $\Cone_0(\bbZ^2)$ that winds clockwise around the origin $0\in\bbZ^2$, e.g.\
\begin{flalign}
\begin{gathered}
\begin{tikzpicture}[scale=0.3]
\draw[fill=blue, opacity=0.3] (-2,-2) -- (-2,2) -- (0,0) -- (-2,-2) ;
\draw[thick] (-2,-2) -- (-2,2) -- (2,2) -- (2,-2) -- (-2,-2);
\filldraw (0,0) circle (2pt);
\end{tikzpicture}
\end{gathered}~\longrightarrow~
\begin{gathered}
\begin{tikzpicture}[scale=0.3]
\draw[fill=blue, opacity=0.3]  (-2,-2) -- (-2,2) -- (2,2) -- (2,-2) -- (0,0) -- (-2,-2);
\draw[thick] (-2,-2) -- (-2,2) -- (2,2) -- (2,-2) -- (-2,-2);
\filldraw (0,0) circle (2pt);
\end{tikzpicture}
\end{gathered}~\longleftarrow~
\begin{gathered}
\begin{tikzpicture}[scale=0.3]
\draw[fill=blue, opacity=0.3]  (0,0) -- (2,2) -- (2,-2) -- (0,0);
\draw[thick] (-2,-2) -- (-2,2) -- (2,2) -- (2,-2) -- (-2,-2);
\filldraw (0,0) circle (2pt);
\end{tikzpicture}
\end{gathered}~\longrightarrow~
\begin{gathered}
\begin{tikzpicture}[scale=0.3]
\draw[fill=blue, opacity=0.3]  (-2,2) -- (-2,-2) -- (2,-2) -- (2,2) -- (0,0) -- (-2,2);
\draw[thick] (-2,-2) -- (-2,2) -- (2,2) -- (2,-2) -- (-2,-2);
\filldraw (0,0) circle (2pt);
\end{tikzpicture}
\end{gathered}
~\longleftarrow~
\begin{gathered}
\begin{tikzpicture}[scale=0.3]
\draw[fill=blue, opacity=0.3] (-2,-2) -- (-2,2) -- (0,0) -- (-2,-2) ;
\draw[thick] (-2,-2) -- (-2,2) -- (2,2) -- (2,-2) -- (-2,-2);
\filldraw (0,0) circle (2pt);
\end{tikzpicture}
\end{gathered}\quad.
\end{flalign}
\end{subequations}
This defines a model (unique up to equivalence) for the self-equivalence by taking the holonomy
\begin{equation}\label{eqn:T}
\begin{tikzcd}[column sep=tiny]
	\ar[d,"p_{V_1}"'] \und{\und{\SSS}}_{(\AAA,\pi_0)} \ar[rrrr,dashed,"T"] &&&&\und{\und{\SSS}}_{(\AAA,\pi_0)}\\
	\SSS_{(\AAA,\pi_0)}(V_1) \ar[rd,"\iota_{V_1}^{V_2}"'] & &
	\SSS_{(\AAA,\pi_0)}(V_3) \ar[rd,"\iota_{V_3}^{V_4}"'] & &
	\SSS_{(\AAA,\pi_0)}(V_1) \ar[u,"\iota_{V_1}"']\\
	& 	\SSS_{(\AAA,\pi_0)}(V_2) \ar[ru,"p_{V_3}^{V_2}"'] & & \SSS_{(\AAA,\pi_0)}(V_4) \ar[ru,"p_{V_1}^{V_4}"']
\end{tikzcd}~.
\end{equation}
As a direct consequence of Corollary \ref{cor:selfequivalence},
this $\ast$-functor can be endowed with the 
structure of a braided monoidal self-equivalence
of the braided monoidal $C^\ast$-category
$\big(\und{\und{\SSS}}_{(\AAA,\pi_0)},\star,\pi_0,\tau\big)$.
\begin{ex}\label{ex:Kitaevmonodromy}
A concrete example of the braided monoidal 
$C^\ast$-category $\und{\und{\SSS}}_{(\AAA,\pi_0)}$ 
of localizable superselection sectors is constructed 
in \cite{Naaijkens1} for  Kitaev's quantum double model with $G=\mathbb{Z}_2$, i.e.\ the toric code.
It is therefore natural to ask how the additional 
braided monoidal self-equivalence $T$ from \eqref{eqn:T} 
behaves in this concrete example and the answer we obtain below is that $T$ is trivial.
Our argument below readily generalizes to all finite Abelian groups using the results from~\cite{FiedlerNaaijkens},
however for a generalization to non-Abelian groups one would require 
more sophisticated techniques as in \cite{Bols}.
As already highlighted in the introduction, we believe that this is not
an artifact of the simplicity of this model, but that the triviality of $T$ is rooted in the insufficiency
of traditional superselection theory \cite{DHR,BuchholzFredenhagen} to detect non-trivial topology.
In particular, we expect that adapting the generalized superselection theory from \cite{BrunettiRuzzi}
will lead to a richer braided monoidal $C^\ast$-category for Kitaev's quantum double model, 
endowed with a potentially non-trivial self-equivalence $T$.
\sk

Our goal is to show that $T(\pi) = \pi$, 
for all irreducible stringlike localized representations 
$\pi \in \SSS_{(\AAA,\pi_0)}(V_1)\subseteq \und{\und{\SSS}}_{(\AAA,\pi_0)}$ as introduced 
in \cite{Naaijkens1}. Let us recall that $\pi$ is defined by a choice of path $\gamma$ 
extending to infinity and contained in $V_1 \in \Cone_0(\bbZ^2)$, 
see \cite[Proposition 3.4]{Naaijkens1}. 
In order to compute $T(\pi)$ explicitly, 
recall from \eqref{eqn:Tcones} the sequence of morphisms 
in $\Cone_0(\bbZ^2)$ used in the definition \eqref{eqn:T} of $T$
and that the quasi-inverses $p_{V_1}$, 
$p_{V_3}^{V_2}$ and $p_{V_1}^{V_4}$ involve choices. 
We already defined $p_{V_1}$ so that it acts 
as the identity on objects localized in $V_1$. 
It is convenient to define $p_{V_3}^{V_2}$ so that 
it acts on irreducible stringlike representations 
localized in $V_1$ by a clockwise rotation 
of the underlying path by an angle $\pi$ around the origin $0 \in \bbZ^2$. 
(The required unitary charge transporter 
can be constructed by choosing a path contained in $V_2$ 
from the basepoint of the original path to the basepoint 
of the rotated path, see \cite[Lemma 4.2]{Naaijkens1}.) 
Similarly, it is convenient to define $p_{V_1}^{V_4}$ so that 
it acts on irreducible stringlike representations 
localized in $V_3$ by a clockwise rotation 
of the underlying path by an angle $\pi$ around the origin $0 \in \bbZ^2$. 
(The required unitary charged transporter can be constructed again
by choosing a similar path contained in $V_4$.) 
To transport $\pi \in \SSS_{(\AAA,\pi_0)}(V_1)$ 
clockwise along the chosen sequence of morphisms \eqref{eqn:Tcones} in $\Cone_0(\bbZ^2)$, 
we proceed in two steps.
First, we transport the irreducible representation $\pi = p_{V_1}(\pi)$ from 
$V_1$ to $V_3$ through $V_2$. The outcome is the object 
$\dot{\pi} := p_{V_3}^{V_2} (\iota_{V_1}^{V_2} (\pi)) \in 
\SSS_{(\AAA,\pi_0)}(V_3)$, whose underlying path $\dot{\gamma} = \gamma^\pi$ is by construction 
the clockwise rotation of the path $\gamma$  by an angle $\pi$ around the origin $0 \in \bbZ^2$.
Second, we transport $\dot{\pi}$ from $V_3$ back to $V_1$ through $V_4$. 
The outcome is the object 
$T(\pi) = p_{V_1}^{V_4} (\iota_{V_3}^{V_4} (\dot{\pi})) \in \SSS_{(\AAA,\pi_0)}(V_1)$, 
whose underlying path $\widetilde{\gamma} = \dot{\gamma}^\pi$ 
is by construction the clockwise rotation of the path $\dot{\gamma}$
by an angle $\pi$ around the origin $0 \in \bbZ^2$. 
It follows that $\widetilde{\gamma} = \gamma^{2\pi} = \gamma$ is obtained 
by a clockwise rotation of the original path $\gamma$ by an angle $2\pi$ around the origin $0 \in \bbZ^2$.
Hence, the associated irreducible stringlike localized representations $T(\pi) = \pi$ coincide. 
\end{ex}


\section*{Acknowledgments}
We would like to thank Sebastiano Carpi, Owen Gwilliam, Robin Hillier, Corey Jones
and David Penneys for useful discussions and comments. We also would
like to thank the anonymous referees for their comments which helped us to improve our paper.
The work of M.B.\ is supported in part by the MUR Excellence 
Department Project awarded to Dipartimento di Matematica, 
Universit{\`a} di Genova (CUP D33C23001110001) and it is fostered by 
the National Group of Mathematical Physics (GNFM-INdAM (IT)). 
V.C.\ is partially supported by the grant PID2020-117971GB-C21 funded by
MCIN/AEI/10.13039/501100011033.


\section*{Data availability statement}
All data generated or analyzed during this study are contained in this document.

\section*{Conflict of interest statement}
The authors have no conflict of interest to declare that are relevant to the content of this article. 



\end{document}